%% file: main.tex
\documentclass[a4paper,fleqn]{cas-sc}

\usepackage[numbers,sort&compress]{natbib}

\usepackage[section]{placeins}
\usepackage{lineno}

\ExplSyntaxOn
\cs_gset:Npn \__reset_fig:
{
  \tl_set:Nx \l_fig_pos_tl { htbp }
  \tl_set:Nx \l_fig_cols_tl { 1 }
  \tl_set:Nn \l_fig_align_tl { \centering }
  \skip_set:Nn \l_fig_abovecap_skip { 6pt }
  \skip_set:Nn \l_fig_belowcap_skip { 6pt }
  \skip_set:Nn \l_fig_abovefig_skip { 6pt }
  \skip_set:Nn \l_fig_belowfig_skip { 6pt }
}
\cs_gset:Npn \__reset_tbl:
{
  \tl_set:Nx \l_tbl_pos_tl { htbp }
  \tl_set:Nx \l_tbl_cols_tl { 1 }
  \tl_set:Nn \l_tbl_align_tl { \centering }
  \skip_set:Nn \l_tbl_abovecap_skip { 6pt }
  \skip_set:Nn \l_tbl_belowcap_skip { 0pt }
  \skip_set:Nn \l_tbl_abovetbl_skip { 6pt }
  \skip_set:Nn \l_tbl_belowtbl_skip { 6pt }
}
\ExplSyntaxOff

\newtheorem{theorem}{Theorem}[section]

\newtheorem{proposition}[theorem]{Proposition}
\newdefinition{definition}{Definition}[section]
\newtheorem{assumption}{Assumption}[section]
\newdefinition{remark}{Remark}[section]
\newproof{proof}{Proof}

\begin{document}
\let\WriteBookmarks\relax

\shorttitle{GeONet for Predicting Flow in Heterogeneous Porous Media}
\shortauthors{M.L. Mamud et~al.}

\title [mode = title]{GeONet: Geometry-Encoded Operator Network for Predicting Flow in Heterogeneous Porous Media}

\author[1]{Md Lal Mamud}
\cormark[1]
\ead{lal.mamud@pnnl.gov}

\author[1]{Maruti K. Mudunuru}
\cormark[1]
\ead{maruti@pnnl.gov}

\author[2]{Yucheng Fu}

\author[3,4]{Aramy Troung}

\author[4,5]{Satish Karra}

\author[6]{Ivani Patel}

\author[3,4]{Arunima Bhattacharjee}

\author[3,4]{Tamas Varga}

\author[7]{Pratanu Roy}

\author[8]{Michael Penwarden}

\author[9]{Rajesh K. Singh}

\author[4,5]{Arjun Chakrawal}

\author[10]{Kalyana B. Nakshatrala}

\author[11]{Christof Meile}

\author[12]{Erin Rooney}

\author[1]{Glenn Hammond}

\affiliation[1]{
    organization={Subsurface Science Group, Energy and Environment Directorate, Pacific Northwest National Laboratory},
    city={Richland},
    state={WA},
    postcode={99354},
    country={USA}
}

\affiliation[2]{
    organization={Computational Mathematics Group, Integrated Discovery Sciences Directorate, Pacific Northwest National Laboratory},
    city={Richland},
    state={WA},
    postcode={99354},
    country={USA}
}

\affiliation[3]{
    organization={Environmental Transformations \& Interactions Group, Integrated Discovery Sciences Directorate, Pacific Northwest National Laboratory},
    city={Richland},
    state={WA},
    postcode={99354},
    country={USA}
}

\affiliation[4]{
    organization={Environmental Molecular Sciences Laboratory, Pacific Northwest National Laboratory},
    city={Richland},
    state={WA},
    postcode={99352},
    country={USA}
}

\affiliation[5]{
    organization={Computing, Analytics, and Modeling Group, Integrated Discovery Sciences Directorate, Pacific Northwest National Laboratory},
    city={Richland},
    state={WA},
    postcode={99352},
    country={USA}
}

\affiliation[6]{
    organization={Earth Systems Predictability \& Resilience Group, Energy and Environment Directorate, Pacific Northwest National Laboratory},
    city={Richland},
    state={WA},
    postcode={99354},
    country={USA}
}

\affiliation[7]{
    organization={Physical and Life Sciences Directorate, Lawrence Livermore National Laboratory},
    city={Livermore},
    state={CA},
    postcode={94550},
    country={USA}
}

\affiliation[8]{
    organization={Sandia National Laboratories},
    city={Albuquerque},
    state={NM},
    postcode={87185},
    country={USA}
}

\affiliation[9]{
    organization={Operational Systems Engineering Group, Energy and Environment Directorate, Pacific Northwest National Laboratory},
    city={Richland},
    state={WA},
    postcode={99354},
    country={USA}
}

\affiliation[10]{
    organization={Department of Civil and Environmental Engineering, University of Houston},
    city={Houston},
    state={Texas},
    postcode={77204},
    country={USA}
}

\affiliation[11]{
    organization={Department of Marine Sciences, University of Georgia},
    city={Athens},
    state={Georgia},
    postcode={30602},
    country={USA}
}

\affiliation[12]{
    organization={U.S. Department of Agriculture, Natural Resources Conservation Service},
    city={Knoxville},
    state={TN},
    postcode={37918},
    country={USA}
}

\cortext[1]{Corresponding authors}

\input{sections/abstract}

%\linenumbers

% \begin{highlights}
% \item GeONet explicitly encodes pore geometry for neural operator-based flow prediction.
% \item Leaf, Branch, and Trunk networks jointly encode geometry, properties, and location.
% \item Approximately 5000 PFLOTRAN simulations support neural operator training.
% \item GeONet outperforms DeepONet with NRMSE $\leq 0.05$ and MAPE $\leq 2.59\%$.
% \item GeONet generalizes to unseen geometries with $R^2$ exceeding 0.95.
% \end{highlights}


\maketitle

\input{sections/s1_intro}
\input{sections/s2_data}
\input{sections/s3_methods}
\input{sections/s4_results}
\input{sections/s5_conclusions}

\section*{Acknowledgments}

\paragraph{Funding:}
The research described in this paper was conducted under the auspices of Lawrence Livermore National Laboratory (LLNL, Contract No. DE-AC52-07NA27344), Pacific Northwest National Laboratory (PNNL, Contract No. DE-AC05--76RL01830), and Sandia National Laboratories (SNL, Contract No. DE-NA0003525) as part of the DOE Interlaboratory (IL) Laboratory Directed Research and Development (LDRD) Pilot Program.
The LLNL LDRD project number is 25-ERD-052, and the release number is LLNL-JRNL-2024560.
MKM, AT, SK, IP, AB, TV, RKS, and AC's research was performed under a project award (Award DOIs: 10.46936/intm.proj.2023.60904/60008965) from the Environmental Molecular Sciences Laboratory, a DOE Office of Science User Facility sponsored by the Biological and Environmental Research program under contract no. DE-AC05-76RL01830.

\paragraph{Competing interests:}
There are no competing interests to declare.

\paragraph{Disclaimer.}
This research work was prepared as an account of work sponsored by an agency of the United States Government.
Neither the United States Government nor any agency thereof, nor any of their employees, makes any warranty, express or implied, or assumes any legal liability or responsibility for the accuracy, completeness, or usefulness of any information, apparatus, product, or process disclosed, or represents that its use would not infringe privately owned rights.
Reference herein to any specific commercial product, process, or service by trade name, trademark, manufacturer, or otherwise does not necessarily constitute or imply its endorsement, recommendation, or favoring by the United States Government or any agency thereof.
The views and opinions of authors expressed herein do not necessarily state or reflect those of the United States Government or any agency thereof.

\paragraph*{Data, High-Performance Computing, and Software Availability:}
\label{Sec:data_availability}
The PFLOTRAN simulation datasets used for training, validation, and testing
will be made available through the EMSL Science Central Data Platform upon
publication, with a persistent data identifier provided in the final
manuscript. 
The computational workflows and software supporting this study
are developed within the Chip2Flow framework
(\url{https://github.com/EMSL-Computing/Chip2Flow}), a Python-based framework for integrating multimodal experimental data with calibration, simulation, validation, and data-worth analysis of coupled flow, thermal, and reactive-transport processes.
The GeONet implementation, AI artifacts and model configurations, and supporting scripts
required to reproduce the reported results will be released with the
associated software repository on EMSL computing GitHub.

Large-scale PFLOTRAN simulation ensembles were generated using the entire EMSL
Tahoma high-performance computing system (\url{https://www.emsl.pnnl.gov/MSC/UserGuide/tahoma/tahoma_overview.html}).
Tahoma is a high-performance computing system with 184 Intel Cascade Lake compute nodes, each containing 36 Intel Xeon Gold 6254 CPU cores at 3.1 GHz and 384 GB of memory. 
The system provides approximately 98 TB of aggregate memory, 10 PB of global BeeGFS storage, and a peak performance of approximately 1.015 PFLOPS. 
Tahoma also includes 24 ML/AI nodes equipped with two NVIDIA Tesla V100 32-GB GPUs per node.
Simulation data generation and extensive computational testing on Tahoma required approximately 13 million CPU core-hours in total (needed for GeoNet model development).
Additionally, EMSL Tahoma and the PNNL Deception high-performance computing systems (CPUs and GPUs) were used for AI/ML model development,
hyperparameter optimization, training, and evaluation. 
Deception, a HPC cluster comprising approximately 284 compute nodes and 18,000 AMD EPYC CPU cores, with 256 GB of memory per standard node, high-speed InfiniBand networking, and parallel BeeGFS storage. 
The system provides NVIDIA A100 and H100 GPU accelerators for large-scale AI/ML model training, along with CPU computing resources managed through the SLURM workload scheduler. 
These computational resources supported large-scale neural-operator training and testing.
On Deception, we trained and evaluated GeoNet using NVIDIA A100 GPUs to accelerate model training and performance evaluation.

\section*{Author contributions:}
MLM contributed to conceptualization, development of the research idea, problem formulation, GeONet code development and implementation, technical execution, and preparation of the original manuscript draft.
MKM contributed to funding acquisition; conceptualization and development of the research idea; PFLOTRAN data generation; mathematical formulation and analysis of GeONet; testing and review of the GeONet code; technical execution; project supervision and oversight; and manuscript drafting, editing, and revision.
AT contributed to generating Pore2Chip-enabled geometry realizations derived from experimental data needed for PFLOTRAN simulations.
IP contributed to GeONet code development, technical execution, results generation and analysis, CPU and GPU scaling, and manuscript drafting.
YF, PR, and MP contributed to funding acquisition, technical supervision, and manuscript review and editing.
AB, AT, and ER contributed to the generation and analysis of experimental and Pore2Chip datasets needed for PFLOTRAN simulations.
AC and RKS contributed to pore-scale modeling and analysis using two-dimensional Stokes-flow simulations.
AT, SK, AB, TV, RKS, AC, KBN, CM, ER, and GH contributed to manuscript review and editing.

%====================================;
%  Section: AI usage and disclosure  ;
%====================================;
\section*{AI usage and disclosure}
Enterprise LLMs (e.g., Grok, ChatGPT) were used to assist with editing (e.g., focusing on grammar) this manuscript draft. 
The authors reviewed the edited text, checked technical claims against the software and documentation, and are responsible for the final content.

\section*{Abbreviations}
\label{Sec:abbreviations}
\begin{description}
    \item[AI] Artificial Intelligence
    \item[BC] Boundary Condition
    \item[CNN] Convolutional Neural Network
    \item[CPU] Central Processing Unit
    \item[DeepONet] Deep Operator Network
    \item[DeepONet-BC] Boundary-Condition-Aware Deep Operator Network
    \item[DOE] U.S. Department of Energy
    \item[EMSL] Environmental Molecular Sciences Laboratory
    \item[FNO] Fourier Neural Operator
    \item[GELU] Gaussian Error Linear Unit
    \item[GeONet] Geometry-Encoded Operator Network
    \item[GPU] Graphics Processing Unit
    \item[HDF5] Hierarchical Data Format Version 5
    \item[HPC] High-Performance Computing
    \item[HPO] Hyperparameter Optimization
    \item[JIT] Just-In-Time Compilation
    \item[LDRD] Laboratory Directed Research and Development
    \item[LLNL] Lawrence Livermore National Laboratory
    \item[LoC] Lab-on-Chip
    \item[MAE] Mean Absolute Error
    \item[MAPE] Mean Absolute Percentage Error
    \item[ML] Machine Learning
    \item[MLP] Multilayer Perceptron
    \item[MSE] Mean-Squared Error
    \item[NRMSE] Normalized Root-Mean-Squared Error
    \item[PDE] Partial Differential Equation
    \item[PFLOTRAN] Massively parallel flow and reactive-transport simulator
    \item[PNNL] Pacific Northwest National Laboratory
    \item[RCP] Reduced Complexity Platform
    \item[RFF] Random Fourier Features
    \item[RMSE] Root-Mean-Squared Error
    \item[SDF] Signed-Distance Field
    \item[SNL] Sandia National Laboratories
    \item[XCT] X-ray Computed Tomography
    \item[XLA] Accelerated Linear Algebra
    \item $\mathbf{a}\odot\mathbf{b} \; \text{= element-wise multiplication (Hadamard product), produces a vector}$
\end{description}

\appendix
\input{sections/s6_math_appendix}

\bibliographystyle{model1-num-names}
\bibliography{references}

\end{document}

%% file: sections/abstract.tex
\begin{abstract}
\boldmath
Rapid prediction of pressure-driven flow in microfluidic porous-media analogs is important for experimental design, tracer-test interpretation, and real-time digital-twin development for subsurface energy and critical mineral science applications. High-fidelity numerical simulators resolve geometry-controlled flow in heterogeneous pore networks, but their computational cost limits repeated evaluation across large geometry ensembles. Here, we introduce GeONet, a Geometry-Encoded Operator Network that predicts steady-state pressure fields by jointly representing material properties, spatial coordinates, and explicit solid--void geometry. GeONet extends the conventional Branch--Trunk architecture with a dedicated Leaf Network that encodes the binary domain mask through multiscale features. The Branch Network represents porosity and permeability, the Trunk Network combines coordinates with local geometric descriptors and random Fourier features, and a gated DeepMixer integrates the three latent representations before decoding pressure. GeONet was trained and evaluated using 5,000 PFLOTRAN simulations of heterogeneous Reduced Complexity Platform domains, comprising nearly 30~TB of data. Direct comparisons with pore-scale Stokes solutions across three geometries yielded Pearson correlations of 0.987--0.991 and $R^2$ values of 0.974--0.982, supporting PFLOTRAN solutions as reference data. Under matched data partitions and two-stage hyperparameter optimization, GeONet substantially outperformed a boundary-condition-aware DeepONet baseline. Across three held-out test domains, GeONet achieved NRMSE of 0.0381--0.0500, MAPE of 1.76--2.59\%, and $R^2$ of 0.9673--0.9822, compared with 0.1483--0.2040, 7.81--10.45\%, and 0.5658--0.7126 for DeepONet, respectively. GeONet generalized to three additional unseen geometries, achieving NRMSE of 0.0448--0.0676 and $R^2$ of 0.9511--0.9770. These results demonstrate that explicit geometry encoding improves accuracy, pressure-range preservation, and cross-geometry generalization of neural operators for flow in heterogeneous porous domains.
\end{abstract}

%% file: sections/s1_intro.tex
\section{Introduction}
\label{Sec:introduction}
Fluid flow and solute transport in porous media depend not only on material parameters, such as porosity and permeability, but also on the geometry of the connected pore space~\cite{Bear1972Dynamics,Blunt2017Multiphase,Berkowitz2002FracturedReview}. 
This coupling is especially important in Lab-on-Chip (LoC) and Reduced Complexity Platform (RCP) experiments, where the chip's solid--void pattern defines connected pathways, constrictions, dead zones, and preferential channels~\cite{Karadimitriou2012Micromodels,Anbari2018Microfluidic,Gerami2019MicrofluidicsPorous}. 
These geometric features create sharp pressure gradients and localized flow structures that a surrogate model must resolve to support experimental design, inverse analysis, or real-time control~\cite{Kapteyn2021DigitalTwin}.
High-fidelity numerical models provide the most direct way to resolve these effects. 
Pore-scale Stokes simulations describe creeping flow in the explicitly resolved void space~\cite{Whitaker1986Darcy}, while continuum models such as PFLOTRAN solve Darcy-scale flow and transport equations on structured and unstructured heterogeneous grids~\cite{BearCheng2010Groundwater,Hammond2014PFLOTRAN,Lichtner2015PFLOTRAN}. 
Both approaches are accurate but expensive for high-resolution chips or large geometry ensembles. 
In our generated dataset, a single PFLOTRAN simulation requires approximately two days on two compute nodes (72 CPU cores). 
This cost motivates surrogate models that can rapidly approximate the solution operator while retaining the geometry-dependent structure of the pressure field~\cite{Chan2018MLUncertainty,Wen2022UFNO,Diab2023UDeepONet,Tang2021DLCarbonStorage,Kadeethum2024NeuralOperatorCO2,Goswami2022TransferOperator}.

Neural operators provide a natural framework for this task because they learn mappings between input functions and output solution fields~\cite{Kovachki2023NeuralOperator,Lu2021DeepONet,Li2021FNO}. 
DeepONet \cite{Lu2021DeepONet} and the Fourier Neural Operator \cite{Li2021FNO} have shown that neural networks can approximate solution operators across broad classes of partial differential equations (PDEs). 
Subsequent developments, including physics-informed neural operators~\cite{Wang2022PINO}, improved DeepONet training strategies~\cite{Wang2022ImprovedDeepONet}, multi-input operator networks~\cite{MIONet2022}, U-Net-enhanced DeepONet architectures~\cite{Diab2023UDeepONet}, discontinuity-capturing neural operator~\cite{roy2026phi}, and arbitrary-geometry transformer formulations~\cite{Chen2026ArGEnT} have increased the flexibility and accuracy of operator-learning models.
Despite these advances, many operator-learning architectures remain primarily coordinate-centric. 
In a standard Branch--Trunk formulation, the Branch Network encodes input functions into global coefficients and the Trunk Network evaluates basis-like functions at spatial coordinates, often enriched with random Fourier features to represent high-frequency spatial variation~\cite{Tancik2020FourierFeatures}. 
This structure is effective for smooth solution manifolds but is less effective when the dominant source of variability is the local interaction between pore geometry and flow. 
If geometric information is omitted, compressed into a single global embedding, or supplied through a fixed descriptor, the model can reproduce the large-scale pressure drop but often smooths over pore-scale discontinuities, constrictions, and localized extrema.

Geometry-encoded variants have begun to address this limitation. 
Recent extensions of DeepONet have incorporated geometric information to enable operator learning across parameterized and non-parametric domains. 
Geom-DeepONet introduced a point-cloud-based DeepONet framework for predicting fields on 3D parameterized geometries by combining spatial coordinates, signed-distance information, and geometry-conditioned branch--trunk representations~\cite{he2024geomdeeponet}. 
Point-DeepONet further extended this idea to non-parametric 3D geometries by integrating PointNet-based geometry encoding with DeepONet for nonlinear structural-response prediction under varying shapes and loading conditions~\cite{park2024pointdeeponet}. 
Related Signed Distance Function (SDF)-based geometric DeepONet formulations have also been developed for 3D flow surrogates around complex objects, where signed-distance fields and derivative constraints improve geometric awareness and boundary-gradient prediction~\cite{rabeh2025neuraloperatorflowsurrogates}. 
Broader geometry-informed neural operators, including Geometry-informed neural operator (GINO) for large-scale 3D PDEs and Geo-FNO for operators on general geometries, similarly emphasize explicit geometric encodings rather than coordinate-only evaluation~\cite{Li2023GINO,Li2023GeoFNO, Yin2024DIMON}. 
Physics- and geometry-informed DeepONet variants, such as PGI-DeepONet and GAPI-DeepONet, incorporate geometric parameterizations or geometry-adaptive conditioning, along with physics-informed losses, to improve generalization across arbitrary computational domains~\cite{nair2024physicsgeometryinformed,lee2026gapideeponet}. 
More recently, ArGEnT introduced an attention-based geometry-aware transformer architecture that separates query-point sampling from geometry encoding through cross-attention, improving robustness for field prediction on complex geometries and arbitrary spatial evaluation locations~\cite{Chen2026ArGEnT}.

Beyond general geometric conditioning, recent work has begun tailoring geometry-aware operator learning to transport phenomena in complex microstructures. 
In this direction, PRT-DeepONet extends DeepONet to pore-scale reactive transport by learning geometry-dependent concentration fields in porous media, with a recent extension incorporating pore-scale velocity prediction through physically informed representations of local pore geometry, nonlocal connectivity, wall proximity, and pressure-driven~\cite{Kim2026PRTDeepONet,Jo2026PRTDeepONetVelocity}. 
PRT-DeepONet incorporates geometric information through convolutional feature extraction and a geodesic-distance descriptor in the trunk input. 
Complementary CNN-enhanced operator frameworks such as Porous-DeepONet further demonstrate that convolutional geometry encoders improve learning of parametric reactive-transport operators in porous media~\cite{Huang2024PorousDeepONet}. 
Related graph- and finite-element-based treatments of flow and transport in complex connectivity structures likewise highlight that topology and local conservation structure control solution behavior~\cite{karra2018modeling,mudunuru2016enforcing}. 
Such methods demonstrate the value of geometry-aware operator learning, but they still rely partly on predefined descriptors or global geometry summaries. 
For heterogeneous chip geometries, the desired surrogate should learn how local geometry, material heterogeneity, and the query location interact to determine the solution value.

We address this need with \textit{GeONet}, a Geometry-Encoded Operator Network for pressure-driven flow in heterogeneous porous-media chips (Figure \ref{fig:geonet_architecture}).
GeONet introduces a dedicated Leaf Network that encodes the binary solid--void geometry, alongside a Branch Network for porosity and permeability and a Trunk Network for spatial--temporal coordinates. 
The three latent streams are combined through a gated bilinear fusion module, allowing the predicted solution at a location to depend jointly and nonlinearly on geometry, material properties, and coordinates.
The central hypothesis is that explicit geometry encoding improves neural-operator prediction of microscale pressure fields by preserving localized, topology-induced structure that standard DeepONet architectures tend to smooth. 
\begin{figure*}[width=\linewidth,pos=htbp]
    \centering
    \includegraphics[width=0.95\textwidth]{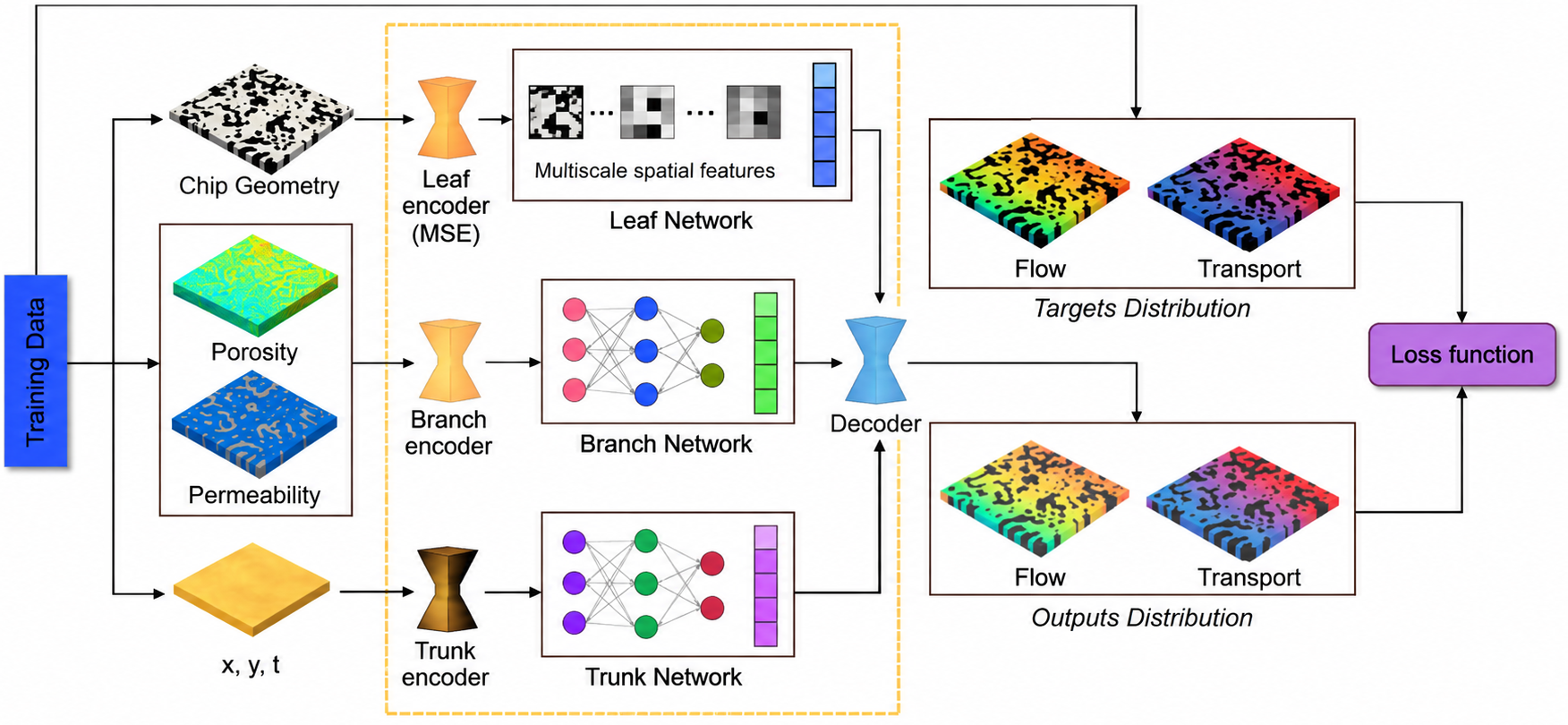}
    \caption{\textbf{Architecture of the GeONet framework.}
    GeONet contains three complementary encoding pathways. The Leaf Network encodes the complete binary solid--void geometry using a multiscale spatial encoder that retains coarse spatial features across multiple resolutions
    before forming a realization-level geometry representation. The Branch Network encodes the hydraulic material properties, represented here by porosity and permeability, while the Trunk Network represents the spatial query coordinates together with local geometry-derived descriptors and random Fourier features. The DeepMixer integrates the resulting geometry, material-property, and query-dependent latent representations, and the prediction head maps the fused representation to the target field.}
    \label{fig:geonet_architecture}
\end{figure*}

We test this hypothesis using thousands of PFLOTRAN simulations on heterogeneous RCP geometries~\cite{Hammond2014PFLOTRAN,Lichtner2015PFLOTRAN}. 
The results show that GeONet accurately reconstructs pressure fields, preserves the pressure distribution, reduces bias in error profiles, and generalizes to chip geometries that were unseen during training. 
Together, these results establish geometry-conditioned operator learning as a practical route toward fast surrogate modeling for LoC and porous-media flow systems for subsurface energy and critical mineral science applications~\cite{Bickle2009CO2Storage,Pires2011CCSOverview}.

%% file: sections/s2_data.tex
\section{Dataset Generation and Physical Parameters}
\label{Sec:S2_data_generation}
The dataset exposes the neural operator to the same sources of variability that control flow in physical chip experiments, including pore geometry, pore connectivity, and heterogeneous material properties. 
Each sample pairs an RCP geometry and its associated porosity and permeability fields with high-resolution PFLOTRAN outputs of pressure, velocity, and tracer concentration. Both pressure, velocity, and tracer concentration fields are stored in
the HDF5 dataset (see Data availability
section~\ref{Sec:data_availability}). The present study focuses on pressure as the prediction target, while the dataset retains the velocity and tracer concentration fields for future extensions of GeONet to coupled flow--transport prediction.

\subsection{Reduced Complexity Platform geometries}
The computational domains are heterogeneous LoC geometries represented by binary solid--void masks (Figure \ref{fig:dataset_generation}a). 
Void pixels denote connected fluid pathways while solid pixels denote impermeable obstacles that redirect flow and impose no-flow barriers. 
The resulting RCP domains are sufficiently simplified to support large simulation ensembles while retaining pore-scale features central to porous-media transport, including channel connectivity, constrictions, disconnected or weakly connected regions, and tortuous pathways.
%
% \begin{figure}
%     \centering
%     \includegraphics[width=0.5\textwidth]{figures/rcp_geometries.png}
%     \caption{\textbf{Example heterogeneous RCP (or microfluidic) geometry.}
%     A representative Lab-on-Chip domain is shown as a binary solid--void structure. 
%     The geometry contains connected pathways, constrictions, and isolated solid regions that control pressure gradients and preferential flow. I will combine this fig the fig 2}
%     \label{fig:rcp_geometries}
% \end{figure}
%
This representation is deliberately explicit: the geometry is not treated as an implicit consequence of permeability alone. Instead, it is provided directly to the model as a spatial mask so that the neural operator can learn how pore topology affects the local solution. 
Across the ensemble, the geometries vary in connectivity, channel width, and spatial heterogeneity, producing a range of flow regimes, including preferential pathways, hydraulic bottlenecks, and low-flow regions.

\subsection{Porosity and permeability fields}
Each geometry is paired with spatially heterogeneous porosity $\phi(\mathbf{x})$ (Figure \ref{fig:dataset_generation}b) and intrinsic permeability $k(\mathbf{x})$ fields (Figure \ref{fig:dataset_generation}c). 
Porosity controls local fluid storage and appears in the transport equations, while permeability determines hydraulic conductance and, therefore, the magnitude and direction of Darcy flux. 
Regions of high permeability form preferential pathways, whereas low-permeability or solid regions restrict flow and generate strong pressure gradients.
The assigned physical fields are heterogeneous at the same spatial scale as the chip geometry. 
As illustrated in Figure~\ref{fig:dataset_generation}, porosity spans low- and high-storage regions and permeability spans orders of magnitude, ranging from effectively impermeable solid zones to highly conductive pore channels. 
This combination of binary geometry and continuous material fields allows the dataset to represent both topological constraints and contrasts in hydraulic properties.

\subsection{PFLOTRAN simulations}
We simulated reference pressure and tracer fields with PFLOTRAN, a massively parallel simulator for coupled flow and reactive transport in porous media \cite{Hammond2014PFLOTRAN}. 
For each RCP geometry, PFLOTRAN solves the flow problem on a structured grid with the assigned porosity and permeability fields. 
The simulations use steady injection from the inlet boundary, with fluid propagating through the heterogeneous pore network toward the outlet. 
The reported dataset includes pressure fields (Figure \ref{fig:dataset_generation}e) and tracer concentration fields over a two-hour injection period (Figure \ref{fig:dataset_generation}f).

\begin{figure}[width=\linewidth,pos=htbp]
    \centering
    \includegraphics[width=0.97\textwidth]{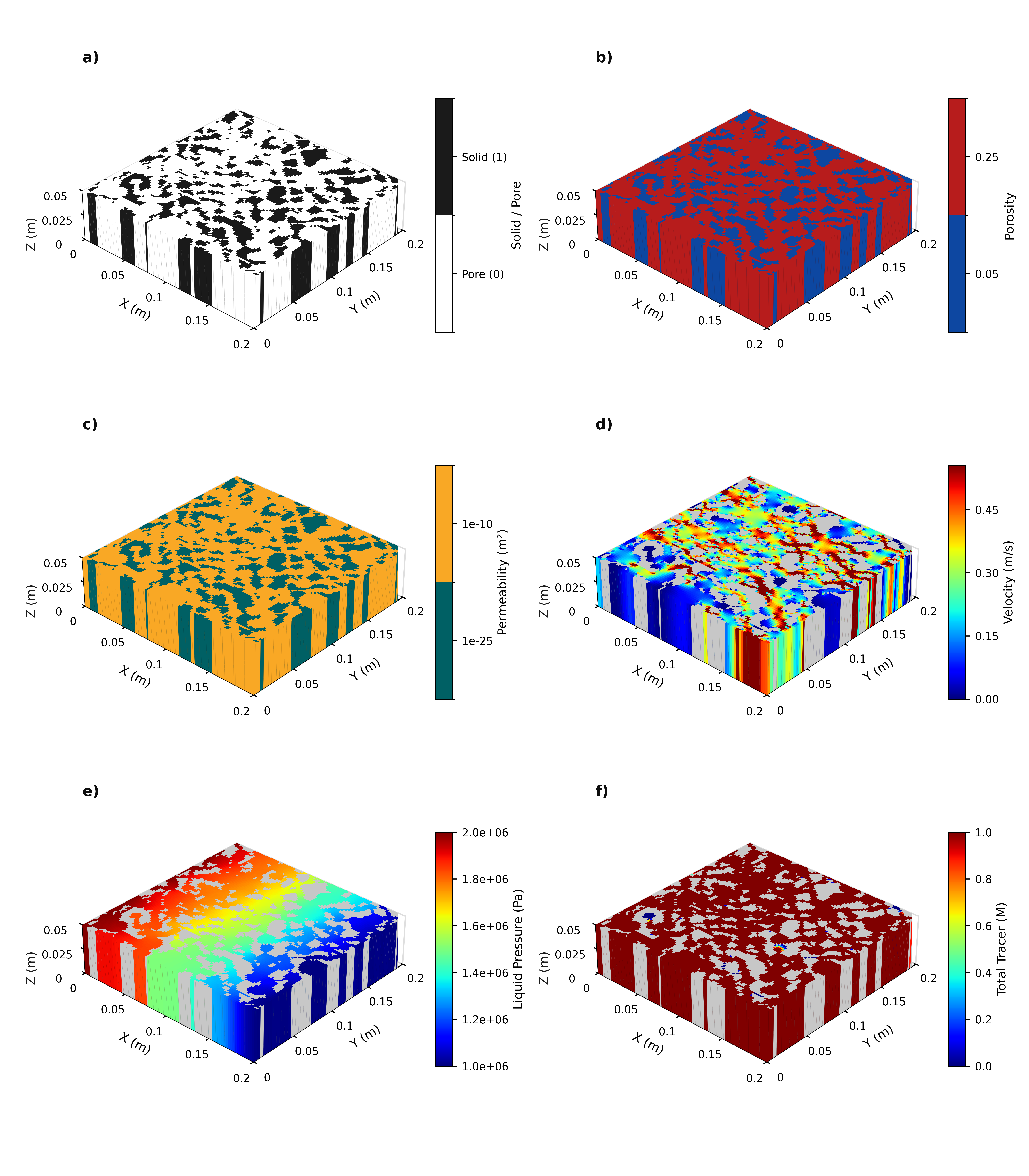}
    \caption{\textbf{Dataset-generation workflow for heterogeneous RCP domains.}
    A representative chip geometry is converted into spatially heterogeneous material properties and used as input to PFLOTRAN Richards-flow and tracer-transport simulations. The resulting velocity, pressure, and concentration fields provide the supervised outputs for operator-learning model development.
    \textbf{a)}~Chip geometry mask (solid vs.\ pore).
    \textbf{b)}~Porosity field,
    \textbf{c)}~Permeability field,
    \textbf{d)}~Liquid-velocity magnitude,
    \textbf{e)}~Liquid pressure, and
    \textbf{f)}~Total tracer concentration (top slice).
    Flow is driven West-East by fixed inlet/outlet pressures; tracer is injected at the West face according to the prescribed inlet schedule. Gray/white regions in panels~\textbf{d)}--\textbf{f)} denote solid, no-flow domains. Axes are shown in meters.}
    \label{fig:dataset_generation}
\end{figure}

The simulations are computationally demanding because they resolve microscale heterogeneity over many distinct geometries. 
A representative PFLOTRAN run required approximately two days on two compute nodes, each with dual-socket Intel Xeon Gold~6254  (Cascade Lake) CPUs at $3.1$~GHz ($36$ cores; $384$~GB RAM per node).
The cost of producing these reference fields motivates the use of a neural surrogate whose low inference cost can amortize the initial simulation expense.
Once trained, the model can evaluate new geometries and parameter fields much faster than running a new high-fidelity simulation.
The resulting supervised learning problem can be written as
\begin{equation}
\left(G,\phi,k;\mathbf{s}\right) \longmapsto \left(p(x,y),c(x,y)\right),
\qquad \mathbf{s}=(x,y),
\end{equation}
where $G$ is the binary solid--void geometry, $\phi$ is porosity, $k$ is permeability, $\mathbf{s}$ is the query coordinate, $p$ is pressure, and $c$ is tracer concentration. 
The simulations reported below focus on the pressure component of this mapping.

%% file: sections/s3_methods.tex
\section{Methodology}
\label{Sec:S3_methods}
GeONet is formulated as a supervised neural operator that maps a geometry-conditioned physical state to a pressure value or concentration value at a queried space--time coordinate. 
This section defines the governing flow models used for verification and data generation, the normalization procedure, the baseline DeepONet comparison, the GeONet architecture, the training setup, and the evaluation metrics.

\subsection{PFLOTRAN Darcy flow and pore-scale Stokes verification}
\label{Sec:darcy_stokes}
We use two flow descriptions on the same chip geometry.
The first is a pore-scale incompressible Stokes model solved directly in the resolved void space.
The second is a continuum Darcy model solved by PFLOTRAN on a structured grid with defined heterogeneous porosity and permeability.
The Stokes comparison is used only to verify that the PFLOTRAN pressure fields preserve the expected pore-scale pressure structure.

Let $\Omega_f$ denote the connected fluid region and $\Gamma_s$ the solid--fluid interface, i.e., the pore wall.
Under the creeping-flow conditions, the governing equations for steady incompressible Stokes flow are
\begin{equation}
-\nabla p + \mu \nabla^2 \mathbf{u}=\mathbf{0}, \qquad
\nabla\cdot\mathbf{u}=0 \qquad \mathrm{in}\, \; \Omega_f,
\label{eq:stokes}
\end{equation}
where $\mathbf{u}$ is the velocity, $p$ is the pressure, and $\mu$ is dynamic viscosity of fluid.
The boundary conditions are
\begin{equation}
\mathbf{u}=\mathbf{0} \quad \mathrm{on}\quad \Gamma_s, \qquad
p=p_{\mathrm{in}} \quad \mathrm{on}\quad \Gamma_{\mathrm{in}}, \qquad
p=p_{\mathrm{out}} \quad \mathrm{on}\quad \Gamma_{\mathrm{out}},
\label{eq:stokes_bc}
\end{equation}
with no penetration on any remaining exterior boundaries.

PFLOTRAN solves single-phase flow in a porous-medium representation with porosity $\phi(\mathbf{x})$ and intrinsic permeability $k(\mathbf{x})$.
The governing equations are mass conservation and Darcy's law,
\begin{equation}
\frac{\partial}{\partial t}\left(\phi\rho\right)+\nabla\cdot\left(\rho\mathbf{q}\right)=Q,
\qquad
\mathbf{q}=-\frac{k}{\mu}\left(\nabla p-\rho\mathbf{g}\right),
\label{eq:darcy}
\end{equation}
where $\mathbf{q}$ is Darcy flux, $\rho$ is fluid density, $Q$ is a source or sink term, and $\mathbf{g}$ is gravitational acceleration.
For the steady, incompressible, source-free, and negligible gravity conditions, the Darcy equation reduces to
\begin{equation}
\nabla\cdot\mathbf{q}=0, \qquad
\nabla\cdot\left(\frac{k}{\mu}\nabla p\right)=0.
\label{eq:darcy_steady}
\end{equation}
The pressure boundary conditions mirror the Stokes condition, with prescribed inlet and outlet pressures and no-flow conditions elsewhere.

Both models are run with the same Dirichlet pressures,
$p_{\mathrm{in}}=2\times10^{6}\,\mathrm{Pa}$ and
$p_{\mathrm{out}}=1\times10^{6}\,\mathrm{Pa}$.
The Stokes and PFLOTRAN pressure fields are compared directly in pascals; no additive offset or pressure-drop scaling is applied.
Pore cells are identified with a porosity threshold $\phi_{\mathrm{th}}$ ($\phi>\phi_{\mathrm{th}}$).

\subsection{Normalization of Physical Variables}
\label{sec:normalization}

Model inputs and targets are normalized using statistics estimated exclusively from the training split.
We then apply the fitted transformations without modification to the validation and test samples.
Network outputs are transformed back to physical units before calculating the reported metrics and generating the figures.

Porosity was transformed using global min--max normalization:
\begin{equation}
\widetilde{\phi}
=
\frac{\phi-\phi_{\min}^{\mathrm{tr}}}
{\phi_{\max}^{\mathrm{tr}}-\phi_{\min}^{\mathrm{tr}}+\epsilon},
\label{eq:porosity_normalization}
\end{equation}
where $\phi_{\min}^{\mathrm{tr}}$ and $\phi_{\max}^{\mathrm{tr}}$ are the
minimum and maximum porosity values in the training data set, respectively,
and $\epsilon=10^{-12}$ prevents division by zero.
The same training-derived extrema were used to normalize the validation and
test realizations.

Because permeability spans several orders of magnitude, the selected GeONet
and DeepONet configurations use model-specific permeability transformations.
For DeepONet, permeability was first transformed logarithmically:
\begin{equation}
k_{\log}
=
\log_{10}\!\left[\max(k,\varepsilon_k)\right],
\label{eq:permeability_log}
\end{equation}
where $\varepsilon_k$ is a small positive threshold introduced to prevent
evaluation of $\log_{10}(0)$.
The transformed permeability was then normalized using the training-set
range:
\begin{equation}
\widetilde{k}_{\log}
=
\frac{k_{\log}-k_{\log,\min}^{\mathrm{tr}}}
{k_{\log,\max}^{\mathrm{tr}}-k_{\log,\min}^{\mathrm{tr}}+\epsilon}.
\label{eq:permeability_log_normalization}
\end{equation}

In contrast, the selected GeONet configuration uses linear min--max
normalization of the untransformed permeability:
\begin{equation}
\widetilde{k}
=
\frac{k-k_{\min}^{\mathrm{tr}}}
{k_{\max}^{\mathrm{tr}}-k_{\min}^{\mathrm{tr}}+\epsilon}.
\label{eq:permeability_linear_normalization}
\end{equation}

Thus, both models receive permeability inputs scaled approximately to the
unit interval, but DeepONet operates on logarithmically transformed
permeability, whereas GeONet operates on linearly scaled permeability.
These model-specific transformations correspond to the configurations
selected during model development and were retained for the final models
used in the subsequent analyses.

Pressure values, which are of order $10^{6}$~Pa, were normalized using a
training-derived affine transformation:
\begin{equation}
\widetilde{p}
=
\frac{p-Q_{0.05}^{\mathrm{tr}}(p)}
{\max\!\left[\sigma_p^{\mathrm{tr}},\,10^{5}\ \mathrm{Pa}\right]},
\label{eq:pressure_normalization}
\end{equation}
where $Q_{0.05}^{\mathrm{tr}}(p)$ is the fifth percentile of the training
pressure distribution and $\sigma_p^{\mathrm{tr}}$ is its standard
deviation.
The lower bound of $10^{5}$~Pa prevents amplification of numerical
variations when the pressure distribution has a small standard deviation.
The fifth-percentile offset is preferred over the absolute minimum (or mean)
as a robust location estimate: pressure grids contain solid-masked and
near-boundary extremes, so a low quantile reduces sensitivity to a few
outlier voxels while remaining close to the Dirichlet low-pressure boundary
that dominates the lower end of the field
\cite{huber1981robust,rousseeuw1993alternatives}.
Quantile-based centering of this type is standard practice for robust
feature scaling when sample moments are easily distorted by heavy tails or
isolated extremes~\cite{pedregosa2011scikit}.
This transformation was applied consistently during training, validation,
and inference.

Normalized pressure predictions were returned to physical units using
\begin{equation}
\hat{p}
=
\widetilde{\hat{p}}\,
\max\!\left[\sigma_p^{\mathrm{tr}},\,10^{5}\ \mathrm{Pa}\right]
+
Q_{0.05}^{\mathrm{tr}}(p),
\label{eq:pressure_inverse_normalization}
\end{equation}
where $\widetilde{\hat{p}}$ and $\hat{p}$ denote the normalized and
dimensional model predictions, respectively.
All reported pressure errors and spatial comparisons were calculated after
applying this inverse transformation.

The two models also differ in their treatment of spatial coordinates.
For DeepONet, the query coordinates were independently mapped to the unit
square for each realization:
\begin{equation}
\widetilde{x}
=
\frac{x-x_{\min}}
{x_{\max}-x_{\min}+\epsilon},
\qquad
\widetilde{y}
=
\frac{y-y_{\min}}
{y_{\max}-y_{\min}+\epsilon}.
\label{eq:coordinate_normalization}
\end{equation}
This normalization provides the DeepONet Trunk with dimensionless coordinates
on a common numerical range. In contrast, the GeONet Trunk uses the physical
coordinates directly and augments them with local geometry-derived descriptors
and random Fourier features. This distinction follows from the respective
model formulations: the DeepONet Trunk is coordinate based, whereas the
GeONet Trunk jointly represents spatial position and the local geometric
environment of each query point. Retaining physical coordinates in GeONet
also preserves their spatial scale when constructing the Fourier embedding
and relating query locations to the geometry-derived descriptors. The same
coordinate treatment is applied consistently during GeONet training,
validation, and testing.

\subsection{Baseline DeepONet}
\label{sec:baseline_deeponet}

We use a boundary-condition-aware Deep Operator Network, hereafter denoted
DeepONet, as the baseline for evaluating the contribution of GeONet's
explicit geometry representation.
DeepONet retains the conventional Branch--Trunk structure but does not
include the Leaf Network or the geometry-aware DeepMixer pathway introduced
in GeONet.

In its canonical factorized form, DeepONet approximates an operator by
expressing the solution at a spatial query coordinate $\mathbf{s}$ as
\begin{equation}
\hat{u}_{\boldsymbol{\theta}}(\mathbf{s})
=
\sum_{j=1}^{r}
b_j(\mathbf{a})\,t_j(\mathbf{s}) + b_0,
\label{eq:deeponet_factorization}
\end{equation}
where $\mathbf{a}$ denotes the discretized input functions, $b_j(\mathbf{a})$
are the coefficients produced by the Branch Network, $t_j(\mathbf{s})$ are
the coordinate-dependent basis representations generated by the Trunk
Network, $r$ is the latent dimension, and $b_0$ is a trainable bias.
This factorization separates the input-function representation from the
spatial coordinates at which the output operator is evaluated.

The DeepONet implementation used in this study follows the same underlying
Branch--Trunk principle but employs a learned nonlinear fusion layer instead
of the direct inner product in Eq.~\eqref{eq:deeponet_factorization}.
This modification allows the Branch and Trunk latent dimensions to differ
while retaining the defining separation between realization-level inputs and
spatial query coordinates.
The Branch Network encodes the hydraulic properties of each RCP realization.
It stacks porosity and log-transformed permeability as input channels and
spatially averages them over the domain.
The resulting realization-level descriptors are processed using a multilayer
perceptron to produce the global hydraulic latent vector $\mathbf{b}$.
This representation compactly describes the bulk material state but does not
explicitly retain the spatial arrangement of the pore geometry.

The Trunk Network maps the normalized query coordinate $\mathbf{s}=(x,y)$ to
a spatial latent representation $\mathbf{t}(\mathbf{s})$.
Unlike the GeONet Trunk Network, it does not receive local mask indicators,
distance-to-boundary descriptors, or other geometry-derived features.
Although Stage~II hyperparameter optimization (HPO)\footnote{In this study, HPO refers to the systematic search used to select model-architecture and training hyperparameters based on validation performance, using a broad search on 500 simulations followed by targeted refinement on the full set of 5000 simulations.} identified a mild random Fourier feature setting
($n_f{=}4$, $\sigma_f{=}1$), the boundary condition (BC)-focused training protocol disabled random Fourier features in the Trunk (\texttt{DEEPONET\_USE\_RFF}${}=0$), effectively setting
$n_f{=}0$ and $\sigma_f{=}0$.
Consequently, the final DeepONet model used for the reported results employs a pure coordinate MLP, in contrast to GeONet's selected Fourier representation ($n_f{=}48$, $\sigma_f{=}10$).
The Branch latent vector is broadcast across the spatial query points and
concatenated with the corresponding Trunk representations.
A multilayer perceptron then fuses the two streams:
\begin{equation}
\mathbf{z}_{\mathrm{DON}}(\mathbf{s})
=
\mathcal{F}_{\boldsymbol{\theta}}
\left[
\mathbf{b},
\mathbf{t}(\mathbf{s})
\right],
\label{eq:deeponet_fusion}
\end{equation}
where $\mathcal{F}_{\boldsymbol{\theta}}$ denotes the Branch--Trunk fusion
network.
The predicted pressure is subsequently obtained from
\begin{equation}
\hat{P}_{\boldsymbol{\theta}}(\mathbf{s})
=
\mathcal{D}_{\boldsymbol{\theta}}
\left[
\mathbf{z}_{\mathrm{DON}}(\mathbf{s})
\right],
\label{eq:deeponet_decoder}
\end{equation}
where $\mathcal{D}_{\boldsymbol{\theta}}$ is the pressure-decoder network.
Although this nonlinear fusion differs from the direct dot-product
representation of canonical DeepONet, the model preserves the central
Branch--Trunk decomposition and contains no independent geometry encoder.

The prescribed high- and low-pressure boundaries are incorporated during
training through the boundary-condition penalty defined in
Eq.~\eqref{eq:deeponet_loss}.
The boundary loss is evaluated over three-grid-cell-wide strips adjacent to the Dirichlet boundaries, defining a narrow boundary region where the prescribed pressure values are enforced, while the interior pressure field remains the dominant contribution to the training objective. 
We used a fixed width of three grid cells across all DeepONet experiments as an implementation choice and we did not optimize the boundary-strip width. 
The boundary-condition term contributes to the objective with a weight of $W_{\mathrm{BC}}{=}10$. Boundary-condition enforcement therefore occurs primarily through the training objective rather than through an additional geometry-encoding
pathway.
For the reported DeepONet prediction figures and tables, Dirichlet
high/low edges were additionally overwritten after decoding
(\texttt{ENFORCE\_BC}${}=1$, mode \texttt{auto}, blend width of one
grid cell) so that the displayed fields match the prescribed boundary
values.
Metrics are computed on these post-processed fields unless a
raw-network ablation is stated.

Relative to GeONet, DeepONet differs in three principal respects: it contains
no Leaf geometry encoder, its Branch Network provides only a global
representation of the material properties, and its Trunk Network does not
receive local geometry-derived descriptors (and, for the reported BC
model, uses no Fourier features).
Consequently, the comparison isolates the effect of providing GeONet with
dedicated domain-scale and query-local geometric information.
Both models otherwise use the same PFLOTRAN realizations, data partitions,
pressure targets, validation criterion, and Stage~II training computational cost.
Unless otherwise stated, the reported DeepONet results use the Stage~II HPO tuned models with Branch, Trunk, and fusion dimensions of $128$, $384$, and
$256$, respectively.
The selected model contains five Trunk layers, uses GELU activations, and
employs a coords-only Trunk with Fourier features disabled
($n_f{=}0$, $\sigma_f{=}0$).

\subsection{GeONet Architecture}
\label{sec:geonet_architecture}

GeONet extends the conventional Branch--Trunk neural-operator architecture by explicitly separating material properties, domain-scale geometry, and query-local spatial information into three complementary encoding pathways. In a conventional DeepONet, the Branch Network encodes the input functions, while the Trunk Network represents the coordinates at which the learned operator is evaluated. GeONet introduces an additional geometry pathway, termed the Leaf Network, that encodes the complete solid--void configuration independently of the material-property Branch. The resulting Branch, Leaf, and Trunk representations are subsequently coupled through the DeepMixer before the target field is decoded, as illustrated in Fig.~\ref{fig:geonet_architecture}.

For the pressure-prediction problem considered here, GeONet approximates the nonlinear operator
\begin{equation}
\mathcal{G}_{\boldsymbol{\theta}}:
\left(G,\phi,k;\mathbf{s}\right)
\mapsto
\hat{P}_{\boldsymbol{\theta}}(\mathbf{s}),
\label{eq:geonet_operator}
\end{equation}
where $\boldsymbol{\theta}$ denotes the trainable parameters, $G$ is the binary solid--void mask, $\phi$ and $k$ are the spatially distributed porosity and permeability fields, and $\mathbf{s}=(x,y)$ denotes a query coordinate in the two-dimensional RCP domain. The predicted pressure at $\mathbf{s}$ is denoted by $\hat{P}_{\boldsymbol{\theta}}(\mathbf{s})$. Although the general framework can accommodate additional state variables, including solute concentration, the quantitative evaluation in the present study focuses on steady-state pressure prediction.
As shown in Fig.~\ref{fig:geonet_architecture}, the Leaf Network maps the complete binary geometry mask $G$ to a realization-level latent geometry representation $\mathbf{g}$. In the configuration used for the reported experiments, the Leaf employs a multiscale spatial encoder rather than trainable convolutional layers. The mask is represented at successive spatial resolutions, where summary statistics and coarse spatial representations are retained before further downsampling. The resulting features from three resolution levels are concatenated and processed by a multilayer perceptron to obtain $\mathbf{g}$. Consequently, the Leaf does not reduce the geometry to a single global statistic. Instead, it preserves coarse spatial information about the solid--void organization of the entire domain at multiple scales.

This domain-scale geometry representation serves a different role from both the Branch and Trunk representations. The Branch Network encodes the hydraulic material properties into a latent vector $\mathbf{b}$. Porosity and permeability are stacked as complementary input channels and spatially averaged over the domain, after which the resulting realization-level descriptors are processed by a multilayer perceptron. The Branch therefore provides a compact representation of the global hydraulic material state, whereas the Leaf retains information associated with the spatial organization of the solid--void domain. This distinction is physically relevant because two realizations can have similar globally aggregated material properties while exhibiting different pore pathways, bottlenecks, or disconnected regions as a consequence of their different geometry masks. The binary mask therefore provides structural information that is not, in general, recoverable from the globally aggregated Branch inputs.

The Trunk Network provides the query-dependent representation. For each spatial query location $\mathbf{s}=(x,y)$, it produces a latent vector $\mathbf{t}(\mathbf{s})$ using the coordinates together with local geometric descriptors derived from the binary mask. These descriptors include the local phase indicator and distances to nearby solid and void regions, thereby characterizing the geometric environment surrounding the query point. The coordinate representation is further augmented with random Fourier features to improve the representation of higher-frequency spatial variations and localized pressure changes near geometric constrictions and material interfaces. The final GeONet configuration uses $n_f=48$ Fourier frequencies, a Fourier scale of $\sigma_f=10$, and GELU activation functions.
The Leaf and Trunk pathways therefore provide complementary geometric information at different spatial levels. The Trunk describes where an individual query point lies and the geometry in its immediate neighborhood, whereas the Leaf supplies a realization-level representation of the solid--void organization of the complete domain. Local geometry alone does not uniquely characterize the larger pore configuration: query points with similar phase and distance descriptors can belong to different domain-scale pathways. 
For example, a through-going channel or a dead-end region. By providing the same domain-level geometry representation $\mathbf{g}$ to all queries from a given realization, the Leaf allows the prediction at $\mathbf{s}$ to be conditioned on the larger-scale pore configuration in addition to its local geometric environment. This separation is consistent with the geometry-conditioned formulation discussed further in Appendix~A.

The three latent streams are integrated through the DeepMixer. The Branch and Leaf representations provide realization-level material and geometry information, respectively, while the Trunk representation varies with the query coordinate. The Branch and Leaf vectors are concatenated and broadcast across the spatial query points, after which they interact with the query-dependent Trunk representation through a gated query--key--value mixing operation followed by a multilayer perceptron. Schematically, the fused representation is written as
\begin{equation}
\mathbf{z}(\mathbf{s})
=
\mathcal{M}_{\boldsymbol{\theta}}
\left[
\mathbf{b},
\mathbf{g},
\mathbf{t}(\mathbf{s})
\right],
\label{eq:geonet_fusion}
\end{equation}
where $\mathcal{M}_{\boldsymbol{\theta}}$ denotes the DeepMixer transformation. The prediction head subsequently maps the fused representation to the normalized pressure,
\begin{equation}
\hat{P}_{\boldsymbol{\theta}}(\mathbf{s})
=
\mathcal{D}_{\boldsymbol{\theta}}
\left[\mathbf{z}(\mathbf{s})\right],
\label{eq:geonet_decoder}
\end{equation}
where $\mathcal{D}_{\boldsymbol{\theta}}$ denotes the pressure decoder. This geometry-conditioned nonlinear fusion allows the pressure predicted at each query coordinate to depend jointly on the global hydraulic state, domain-scale solid--void organization, and local geometric environment.

The designation \emph{Leaf Network} therefore reflects the distinct information represented by this pathway rather than a particular neural layer type. It is not simply a second material-property Branch: the Branch compresses porosity and permeability into a global hydraulic representation, whereas the Leaf encodes multiscale spatial information from the complete binary domain mask. Likewise, it is distinct from the Trunk, which represents the query location and its local geometric context. The principal architectural distinction between GeONet and the conventional DeepONet baseline is thus the explicit separation of material state, domain-scale geometry, and query-local information, followed by their interaction through the DeepMixer. The final configuration used in the subsequent analyses has Branch, Leaf, and Trunk latent dimensions of 384, 128, and 128, respectively, together with a DeepMixer latent dimension of 320.

\subsection{Training Setup}
\label{sec:training_setup}

We trained GeONet and DeepONet using pressure fields generated from high-resolution PFLOTRAN simulations of heterogeneous RCP domains. The complete dataset contains approximately 5000 simulations and occupies nearly 30~TB of storage. Generating these realizations required approximately 13 million CPU core-hours on the EMSL Tahoma high-performance computing system.
Unless otherwise stated, we trained the final models on Stage~II tuned models with $N=5000$ realizations, using 3500 samples for training and 1000 for validation. We excluded simulation IDs 91--100 from model training, validation, early stopping, and hyperparameter selection and reserved them for final testing. Stage~I hyperparameter screening used a reduced set of $N=500$ under the same partitioning strategy, as described in Section~\ref{sec:hpo_strategy}.

Both architectures were formulated as neural operators that map realization-specific field and geometry information, together with spatial query coordinates $(x,y)$, to the corresponding pressure values. We evaluated each realization on its native $200\times200$ simulation grid. The resulting full-resolution tile therefore represented one complete pressure field rather than a cropped or downsampled portion of the domain. We estimated input and target normalization statistics exclusively from the training set and applied them to the validation and test sets. Porosity and permeability were the primary material-property inputs, while pressure was the prediction target.

For GeONet, the selected Stage~II architecture comprises a Branch Network that encodes the porosity and permeability fields, a Leaf Network that extracts information from the domain mask, and a Trunk Network that represents the spatial coordinates and local geometric descriptors. The coordinate representation is augmented using random Fourier features with $n_f=48$ frequencies and Fourier scale $\sigma_f=10$. The three pathways combine their outputs through a geometry-aware mixing mechanism to predict pressure at each spatial query location.
For GeONet, the training objective was formulated as the pressure mean-squared error augmented by an optional gradient-matching term:
\begin{equation}
\mathcal{L}_{\mathrm{GeONet}}(\boldsymbol{\theta})
=
\mathcal{L}_{\mathrm{MSE}}(\boldsymbol{\theta})
+
w_{\mathrm{grad}}
\mathcal{L}_{\mathrm{grad}}(\boldsymbol{\theta}),
\label{eq:geonet_loss}
\end{equation}
where $w_{\mathrm{grad}}$ controls the contribution of the gradient-based regularization. The pressure MSE is defined as
\begin{equation}
\mathcal{L}_{\mathrm{MSE}}(\boldsymbol{\theta})
=
\frac{1}{N_q}
\sum_{i=1}^{N_q}
\left[
\hat{p}_{\boldsymbol{\theta}}(\mathbf{x}_i)-p_i
\right]^2,
\label{eq:pressure_mse}
\end{equation}
where $\boldsymbol{\theta}$ denotes the trainable model parameters, $N_q$ is the number of spatial query points, $\hat{p}_{\boldsymbol{\theta}}(\mathbf{x}_i)$ is the normalized pressure predicted at coordinate $\mathbf{x}_i=(x_i,y_i)$, and $p_i$ is the corresponding normalized PFLOTRAN pressure.
The gradient term $\mathcal{L}_{\mathrm{grad}}$ is a normalized finite-difference mismatch between predicted and reference pressure increments on the $200\times200$ tile.
The gradient term was included as a hyperparameter-controlled regularization option during model development, and different values of $w_{\mathrm{grad}}$ were evaluated during HPO (see Figs.~\ref{fig:geonet_param_500} and~\ref{fig:geonet_sensitivity_500} for Stage~I and Figs.~\ref{fig:geonet_param_5000} and~\ref{fig:geonet_sensitivity_5000} for Stage~II). The configuration selected for the final GeONet model has $w_{\mathrm{grad}}=0$. Consequently, Eq.~\eqref{eq:geonet_loss} reduces to the pressure MSE in Eq.~\eqref{eq:pressure_mse} for the final model. Validation, early stopping, and model selection used the pressure MSE alone in all trials.

The selected DeepONet model comprises a Branch Network operating on porosity and log-transformed permeability and a coordinate-based Trunk Network. Unlike GeONet, DeepONet has no explicit Leaf Network or geometry-encoding pathway. The Stage~II tuned models uses a branch latent dimension of 128, a trunk latent dimension of 384, and a mixer dimension of 256. Its coordinate representation includes a mild random Fourier feature expansion with $n_f=4$ and $\sigma_f=1$.
To enforce the prescribed high- and low-pressure boundary conditions, the DeepONet objective combines the pressure MSE with a Dirichlet boundary-condition penalty and an optional gradient-regularization term:

\begin{equation}
\mathcal{L}_{\mathrm{DON}}(\boldsymbol{\theta})
=
\mathcal{L}_{\mathrm{MSE}}(\boldsymbol{\theta})
+
W_{\mathrm{BC}}
\mathcal{L}_{\mathrm{BC}}(\boldsymbol{\theta})
+
w_{\mathrm{grad}}
\mathcal{L}_{\mathrm{grad}}(\boldsymbol{\theta}),
\label{eq:deeponet_loss}
\end{equation}

where $\mathcal{L}_{\mathrm{BC}}$ measures the pressure mismatch within three-grid-cell-wide strips along the prescribed Dirichlet boundaries, $W_{\mathrm{BC}}=10$ controls the contribution of this term, and $\mathcal{L}_{\mathrm{grad}}$ denotes the optional gradient-based regularization. The latter contributes to the objective only when $w_{\mathrm{grad}}>0$.
We trained both models for up to 100 epochs. We monitored validation loss after each epoch and terminated training early if it failed to improve by at least $10^{-5}$ over 20 consecutive epochs. We retained the checkpoints with the minimum validation loss for subsequent prediction, model-comparison, and cross-domain generalization analyses.

\subsection{Hyperparameter Optimization Strategy}
\label{sec:hpo_strategy}

Training neural operators on full-resolution porous-media realizations is computationally demanding, making exhaustive evaluation of all possible architectural and optimization settings impractical. We therefore use a two-stage HPO strategy that first explores a broad configuration space on a reduced dataset and then refines the most promising configurations on the full dataset (Table~\ref{tab:hpo_protocol}). This design enables systematic model selection at a manageable computational cost while strictly separating hyperparameter optimization from final test evaluation.

\begin{table*}[width=\linewidth,pos=t]
\centering
\caption{\textbf{Two-stage hyperparameter optimization protocol.}
Stage~I broadly explores the architecture and optimization space using 500 PFLOTRAN realizations. The five configurations attaining the lowest
validation losses are transferred to Stage~II and refined using the full set of 5000 realizations.}
\label{tab:hpo_protocol}
\setlength{\tabcolsep}{8pt}
\begin{tabular}{@{}p{0.24\textwidth}p{0.34\textwidth}p{0.34\textwidth}@{}}
\toprule
Protocol component
& Stage~I: broad exploration
& Stage~II: transferred refinement \\
\midrule

Dataset size
& $N=500$
& $N=5000$ \\

Dataset partition
& 350 training, 100 validation, and 50 held-out test realizations
& 3500 training, 1000 validation, and 500 held-out test realizations \\

Trials per model
& 100
& 24 \\

Architecture search
& Predefined presets and free-form configurations
& Restricted to architecture families represented among the five
highest-performing Stage~I configurations \\

Configuration transfer
& Not applicable
& Top $K=5$ Stage~I configurations transferred as initialization seeds \\

Refinement strategy
& Broad exploration of architecture and optimization hyperparameters
& Exact retraining of transferred configurations followed by local
optimization refinement \\

Training budget
& Maximum of 100 epochs with early-stopping patience of 20 epochs
& Maximum of 100 epochs with early-stopping patience of 20 epochs \\

Selection criterion
& Minimum validation loss attained during training
& Minimum validation loss attained during refinement \\

Test-data separation
& 50 held-out test realizations excluded from training and HPO
& 500 held-out test realizations excluded from training and HPO \\ 

\bottomrule
\end{tabular}
\end{table*}

Stage~I uses a reduced set of $N=500$ sequential realizations
(\texttt{sim\_total}${}=500$; \texttt{split\_mode}${}=\texttt{seq}$), including 350 training and 100 validation samples. Stage~II uses the full set of $N=5000$, including 3500 training and 1000 validation samples. In both stages, we exclude simulation IDs 91--100 from training and validation and reserve them exclusively for final test-time evaluation. Thus, the reported test realizations do not influence model fitting, early stopping, hyperparameter ranking, or final model selection. All trials predict pressure fields on $200\times200$ spatial grids and use a fixed random seed of 42 for reproducibility.
Each trial is ranked according to the minimum validation loss achieved during training, denoted by \texttt{best\_val}. The training objective consists of the mean-squared pressure error and an optional gradient-based regularization term weighted by $w_{\mathrm{grad}}$. Models are trained for a maximum of 100 epochs, with early stopping applied when the validation loss does not improve by at least $10^{-5}$ over 20 consecutive epochs.

GeONet and DeepONet are optimized independently using the same data partitions, training budget, selection criterion, and applicable hyperparameter ranges. Architecture-specific parameters, such as the Leaf Network dimensions, are optimized only for GeONet because this geometry-encoding component is not present in the classical DeepONet architecture. All shared components are evaluated over equivalent candidate ranges to support a controlled comparison between the two models.

\begin{table*}[t]
\centering
\caption{\textbf{Hyperparameter search space used during Stage~I.}
Candidate architectures are drawn from six predefined configurations or sampled from the free-form ranges listed below. Architecture samples are combined with selected optimization schedules through structured random sampling. Leaf Network parameters apply only to GeONet.
The DeepONet retains the conventional Branch--Trunk architecture.}
\label{tab:hpo_search_space}
\setlength{\tabcolsep}{6pt}
\begin{tabular}{@{}lll@{}}
\toprule
Group & Hyperparameter & Candidate values \\
\midrule
Training
    & Learning rate, $\eta$
    & $\{10^{-5},\,3\times10^{-5},\,10^{-4},\,3\times10^{-4}\}$ \\
    & Weight decay
    & $\{0,\,10^{-6},\,10^{-5}\}$ \\
    & Terminal learning-rate fraction
    & $\{0.1,\,0.3,\,0.5\}$ \\
    & Gradient-loss weight, $w_{\mathrm{grad}}$
    & $\{0,\,0.01\}$ \\
\midrule
Architecture
    & Branch width
    & $\{128,\,256,\,384\}$ \\
    & Trunk width
    & $\{128,\,256,\,384\}$ \\
    & Mixer width
    & $\{128,\,192,\,256,\,320\}$ \\
    & Leaf width
    & $\{128,\,192,\,256\}$ \\
    & Branch depth
    & $\{2,\,3,\,4\}$ \\
    & Leaf depth
    & $\{2,\,3,\,4\}$ \\
    & Trunk depth
    & $\{3,\,4,\,5,\,6\}$ \\
    & Mixer depth
    & $\{2,\,3,\,4\}$ \\
    & Prediction-head width
    & $\{128,\,256,\,384\}$ \\
    & Prediction-head depth
    & $\{2,\,3\}$ \\
    & Number of RFF frequencies, $n_f$
    & $\{16,\,32,\,48\}$ \\
    & Fourier scale, $\sigma_f$
    & $\{5,\,10,\,20\}$ \\
    & Activation function
    & $\{\mathrm{GELU},\,\mathrm{SiLU}\}$ \\
\bottomrule
\end{tabular}
\end{table*}

During Stage~I, we evaluate 100 configurations for each model. Candidate architectures are drawn from six predefined configurations (\texttt{small}, \texttt{base}, \texttt{large}, \texttt{deep\_trunk}, \texttt{wide\_mixer}, and \texttt{silu\_base}) or sampled from a free-form architecture space. The free-form search varies the latent widths and depths of the constituent networks, the number of random Fourier feature frequencies $n_f$, the Fourier scale $\sigma_f$, and the activation function. The optimization search additionally considers the initial learning rate $\eta$, weight decay, terminal learning-rate fraction under cosine decay, and gradient-loss weight $w_{\mathrm{grad}}$ (Table~\ref{tab:hpo_search_space}).
We generate the Stage~I configurations using structured random sampling. We first pair priority optimization schedules with the predefined architecture configurations to ensure representative model families are evaluated under the selected training conditions. The remaining trials are sampled from the complete search space without duplicating identical configurations. This stage intentionally emphasizes architectural diversity, allowing us to identify promising model families before committing the greater computational resources required for full-dataset training.

For each model, we transfer the five Stage~I configurations with the lowest validation losses to Stage~II. Stage~II evaluates 24 configurations per model using an architecture pool restricted to the predefined and free-form configurations represented among these five Stage~I solutions
(\texttt{ARCH\_MODE}${}=\texttt{top\_presets}$). We first retrain the selected configurations without modification using all $N=5000$ realizations to evaluate whether their performance transfers across dataset scales. The remaining trials explore local variations in the optimization parameters around the transferred configurations while retaining the restricted architecture pool.

Stage~II disables unconstrained architectural expansion. Computational resources therefore focus on evaluating and refining the selected model families at the target data scale rather than repeating the broad architecture search performed during Stage~I. For each model, we select the configuration with the lowest Stage~II validation loss for subsequent prediction, comparison, and cross-domain generalization analyses.
The two-stage procedure treats HPO as a computationally expensive black-box optimization problem. The reduced cohort serves as a lower-cost representation of the same operator-learning task, preserving the input--output relationship and spatial resolution while reducing the number of training realizations by one order of magnitude. Transferring the highest-performing configurations concentrates the full-data search on architecture families that have already demonstrated favorable validation performance. This strategy enables broad exploration of architectures without requiring every candidate configuration to be trained on the complete high-resolution dataset.

\subsection{Performance metrics}
\label{Sec:performance_metrics}
All reported pressure metrics are computed after converting model outputs back to physical units. 
The normalized root-mean-squared error is
\[
\mathrm{NRMSE}=\frac{\sqrt{N^{-1}\sum_{i=1}^{N}\left(P_{\mathrm{pred},i}-P_{\mathrm{true},i}\right)^2}}{P_{\max}-P_{\min}},
\]
where $P_{\max}$ and $P_{\min}$ denote the maximum and minimum reference pressures over the evaluated domain. The mean absolute percentage error is
\[
\mathrm{MAPE}=\frac{100}{N}\sum_{i=1}^{N}\left|\frac{P_{\mathrm{pred},i}-P_{\mathrm{true},i}}{P_{\mathrm{true},i}+\epsilon}\right|.
\]
The coefficient of determination is
\[
R^2=1-\frac{\sum_{i=1}^{N}\left(P_{\mathrm{pred},i}-P_{\mathrm{true},i}\right)^2}{\sum_{i=1}^{N}\left(P_{\mathrm{true},i}-\overline{P}_{\mathrm{true}}\right)^2}.
\]
We also examine pressure histograms, pointwise error histograms, one-to-one plots, and horizontal and vertical pressure slices to evaluate whether the model preserves geometry-induced discontinuities and distributional structure.

%% file: sections/s4_results.tex
\section{Results and Discussion}
\label{Sec:S4_results}
The results are organized to follow the modeling pipeline. 
We first verify that the PFLOTRAN pressure fields are consistent with pore-scale Stokes solutions after pressure normalization. 
We then examine optimization behavior, compare GeONet with the DeepONet baseline on pressure reconstruction, analyze distributional errors and profile errors, and finally evaluate generalization on unseen geometries.

\subsection{PFLOTRAN pressure fields are consistent with pore-scale Stokes solutions}
\label{Sec:results_stokes_pflotran}
Figure~\ref{fig:network0_pressure_compare} compares pore-scale Stokes pressure fields with the corresponding PFLOTRAN Darcy pressure fields for three representative RCP geometries (networks~1, 5, and~7). The Stokes and PFLOTRAN solutions are evaluated at the same pore-space grid cells, while soil-matrix cells are excluded from the comparison. The PFLOTRAN arrays are transposed to match the spatial orientation of the corresponding Stokes arrays. Following this alignment, the pressure fields are compared directly in absolute pressure (Pa), without pressure-offset subtraction, normalization, or pressure-drop scaling.

\begin{figure}[width=\linewidth,pos=htbp]
    \centering
    \includegraphics[width=\linewidth]{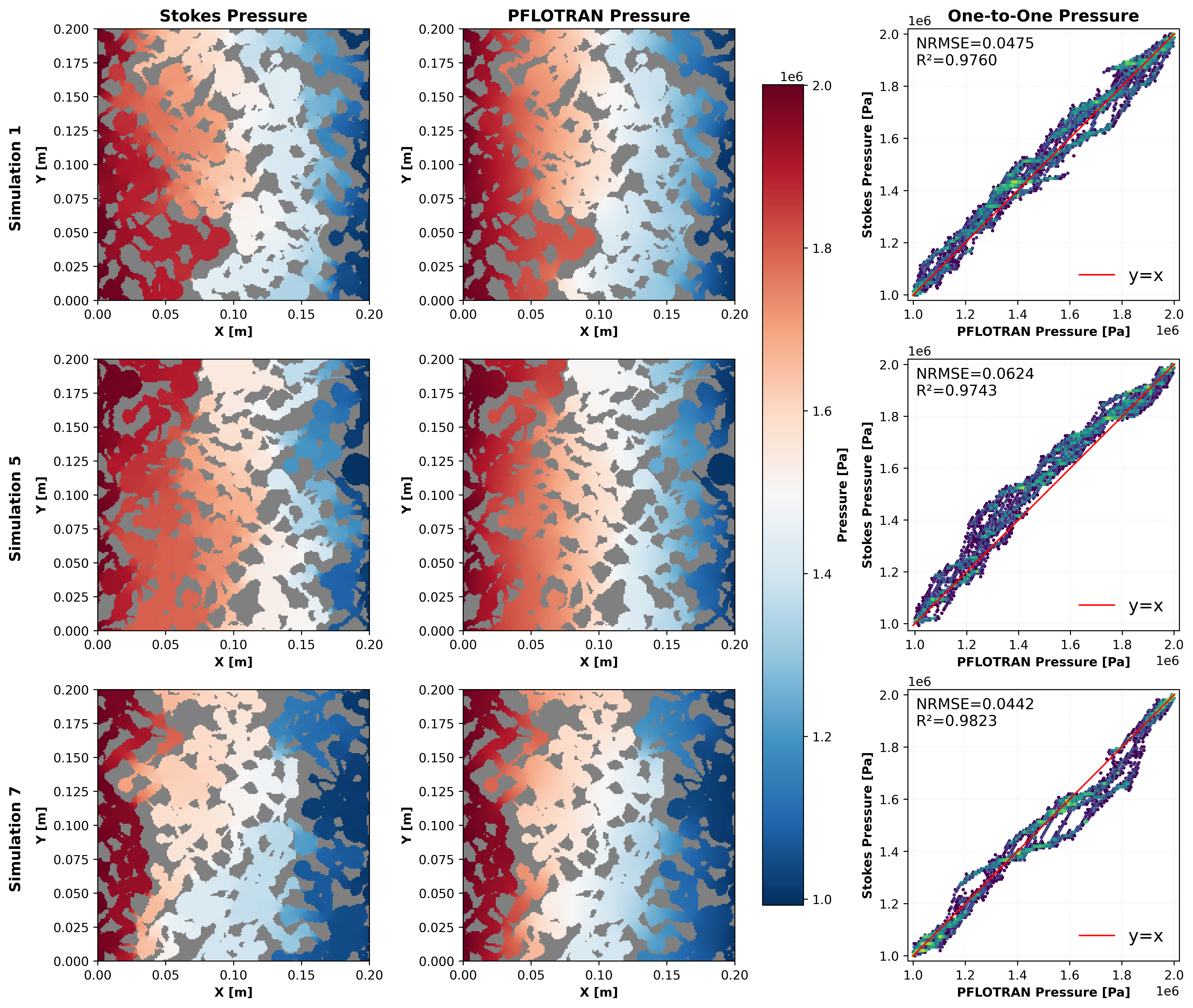}
    \caption{\textbf{Absolute-pressure comparison between the 2-D Stokes and PFLOTRAN simulations.} Rows correspond to networks~1, 5, and~7, respectively. The first column shows the Stokes pressure field, the second column shows the spatially aligned PFLOTRAN pressure field, and the third column presents the corresponding pore-cell one-to-one comparison. Both pressure fields are shown in Pa using a common color scale, without pressure-offset subtraction, normalization, or pressure-drop scaling. Gray regions represent solid cells, and quantitative comparisons are restricted to cells classified as pores in both models. In the one-to-one panels, the red line denotes the ideal $y=x$ relationship, while color indicates the logarithmic sample density.}
    \label{fig:network0_pressure_compare}
\end{figure}

\begin{table}[htbp]
\centering
\caption{Quantitative agreement between the absolute, unscaled pressure fields computed by the 2-D Stokes and PFLOTRAN models for networks~1, 5, and~7. All metrics are evaluated over cells classified as pores in both models. Pearson $r$ measures the linear correlation between the two pressure fields, and $R^2$ denotes the coefficient of determination of the least-squares fit.}
\label{tab:stokes_pflotran_selected}
\begin{tabular}{rrrrrr}
\hline
Network & Common cells & NRMSE & MAE [Pa] & Pearson $r$ & Fit $R^2$ \\
\hline
1 & 28035 & 0.0475 & $3.845 \times 10^{4}$ & 0.9879 & 0.9760 \\
5 & 27112 & 0.0624 & $4.879 \times 10^{4}$ & 0.9870 & 0.9743 \\
7 & 27283 & 0.0442 & $2.895 \times 10^{4}$ & 0.9911 & 0.9823 \\
\hline
\end{tabular}
\end{table}

The Stokes and PFLOTRAN solutions exhibit strong spatial and pointwise agreement across all three representative geometries (Fig.~\ref{fig:network0_pressure_compare}). Pearson correlation coefficients range from 0.9870 to 0.9911, while the coefficients of determination of the least-squares fits range from $R^2=0.9743$ to $0.9823$ (Table~\ref{tab:stokes_pflotran_selected}). The NRMSE ranges from 0.0442 to 0.0624, and the MAE ranges from $2.895\times10^{4}$ to $4.879\times10^{4}$~Pa. Network~7 exhibits the strongest overall agreement, with Pearson $r=0.9911$, $R^2=0.9823$, and NRMSE$=0.0442$. Although localized deviations from the ideal one-to-one relationship remain, the agreement is consistent across geometries with different pore structures. These comparisons support the use of PFLOTRAN pressure fields as reference data for training and evaluating GeONet while recognizing that the Darcy-scale and pore-scale Stokes formulations represent flow at different levels of physical description.

\subsection{Hyperparameter Optimization}
\label{sec:hpo_findings}

The two-stage HPO results for GeONet and the
boundary-condition-focused DeepONet are summarized in
Tables~\ref{tab:hpo_stage_summary}--\ref{tab:hpo_transfer}. All scores denote
the minimum validation loss attained during training
(\texttt{best\_val}), with lower values indicating better performance.
Because the principal objective of this analysis is to establish the
configuration used in the subsequent GeONet experiments, we
examine the GeONet response surface in detail. The corresponding DeepONet
diagnostics are reported in Appendix~\ref{app:deeponet_hpo} to provide a
complete and auditable baseline comparison without interrupting the main
results.

\begin{table}[t]
\centering
\caption{\textbf{Aggregate outcomes of the two-stage hyperparameter
optimization.}
Reported values are the minimum validation losses attained during training
(\texttt{best\_val}). GPU time is the sum of successful trial runtimes
expressed as a single-GPU equivalent.}
\label{tab:hpo_stage_summary}
\begin{tabular}{@{}llrrrrr@{}}
\toprule
Stage & Model & Successful trials & Best $\downarrow$ & Median & Mean & GPU-h \\
\midrule
I ($N=500$)
  & GeONet      & $100$ & $0.02267$ & $0.02902$ & $0.02951$ & $30.0$ \\
  & DeepONet & $100$ & $0.61911$ & $0.62196$ & $0.62214$ & $18.0$ \\
\midrule
II ($N=5000$)
  & GeONet      & $24$  & $0.01944$ & $0.02003$ & $0.02133$ & $87.5$ \\
  & DeepONet & $24$  & $0.56718$ & $0.60091$ & $0.59786$ & $63.1$ \\
\bottomrule
\end{tabular}
\end{table}

During Stage~I, all 100 trials for each model completed successfully. GeONet
attained a minimum validation loss of $0.02267$ with trial~028 using the
\texttt{deep\_trunk} preset. Under the same data partitioning, DeepONet
attained a minimum of $0.61911$ with free-form trial~040, which used branch
and trunk widths of 256 and 384, respectively, six trunk layers, and
effectively disabled random Fourier features ($n_f=0$, $\sigma_f=1$)
(Tables~\ref{tab:hpo_stage_summary} and~\ref{tab:hpo_top5_stage1}).
The median losses of $0.02902$ for GeONet and $0.62196$ for DeepONet
further show that the separation extended across the sampled configuration
space and was not attributable to a single favorable trial.

\begin{table}
\centering
\caption{\textbf{Five highest-performing Stage~I configurations transferred
to Stage~II.}
The configurations are ranked by minimum validation loss at $N=500$.
GeONet and DeepONet use separate sets of random seeds for model initialization. Here, $B$, $T$, and $M$
denote branch, trunk, and mixer latent dimensions, respectively; WD denotes
weight decay.}
\label{tab:hpo_top5_stage1}
\setlength{\tabcolsep}{3.5pt}
\small
\begin{tabular}{@{}clllllrr@{}}
\toprule
Rank & Model & Trial & Preset & Key architecture & $\eta$ / WD & Best val & Epoch \\
\midrule
1 & GeONet & $028$ & \texttt{deep\_trunk}
  & $B=256$, $T=320$ (6 layers), $n_f=48$, $\sigma_f=20$
  & $10^{-4}$ / $0$ & $0.02267$ & $85$ \\
2 & GeONet & $093$ & \texttt{deep\_trunk}
  & Same as rank~1
  & $10^{-4}$ / $10^{-6}$ & $0.02272$ & $82$ \\
3 & GeONet & $056$ & \texttt{small}
  & $B=T=128$ (3 layers), $n_f=16$, $\sigma_f=10$
  & $10^{-4}$ / $10^{-5}$ & $0.02276$ & $98$ \\
4 & GeONet & $097$ & \texttt{deep\_trunk}
  & Same as rank~1
  & $3\times10^{-4}$ / $10^{-6}$ & $0.02285$ & $86$ \\
5 & GeONet & $045$ & \texttt{free}
  & $B=384$, $T=128$, $M=320$, $n_f=48$, $\sigma_f=10$
  & $3\times10^{-4}$ / $10^{-6}$ & $0.02289$ & $56$ \\
\midrule
1 & DeepONet & $040$ & \texttt{free}
  & $B=256$, $T=384$ (6 layers), $n_f=0$, $\sigma_f=1$
  & $3\times10^{-4}$ / $0$ & $0.61911$ & $93$ \\
2 & DeepONet & $052$ & \texttt{free}
  & $B=128$, $T=384$ (5 layers), $n_f=4$, $\sigma_f=1$
  & $3\times10^{-4}$ / $10^{-5}$ & $0.62092$ & $94$ \\
3 & DeepONet & $020$ & \texttt{large}
  & $B=T=384$ (5 layers), $n_f=0$
  & $3\times10^{-4}$ / $10^{-6}$ & $0.62119$ & $84$ \\
4 & DeepONet & $042$ & \texttt{free}
  & $B=128$, $T=384$ (6 layers), $n_f=8$, $\sigma_f=1$
  & $10^{-5}$ / $0$ & $0.62128$ & $59$ \\
5 & DeepONet & $057$ & \texttt{large}
  & $B=T=384$ (5 layers), $n_f=0$
  & $10^{-5}$ / $0$ & $0.62133$ & $100$ \\
\bottomrule
\end{tabular}
\end{table}

The parameter-wise Stage~I results for GeONet reveal that optimization choices
had a stronger influence than most individual width or depth settings
(Fig.~\ref{fig:geonet_param_500}). The lowest-loss trials were concentrated at
learning rates of $10^{-4}$ and $3\times10^{-4}$ and at
$w_{\mathrm{grad}}=0$. By contrast, setting $w_{\mathrm{grad}}=0.01$ or
$0.03$ shifted the trial distribution toward higher validation losses.
Low-loss solutions occurred across several branch, leaf, trunk, and mixer
dimensions, indicating that no single latent width was solely responsible for
the Stage~I optimum. Six trunk layers and 48 random Fourier feature
frequencies were nevertheless associated with several of the best trials,
consistent with the strong performance of the \texttt{deep\_trunk}
configuration.

\begin{figure}[width=\linewidth,pos=htbp]
    \centering
    \includegraphics[width=0.92\textwidth]{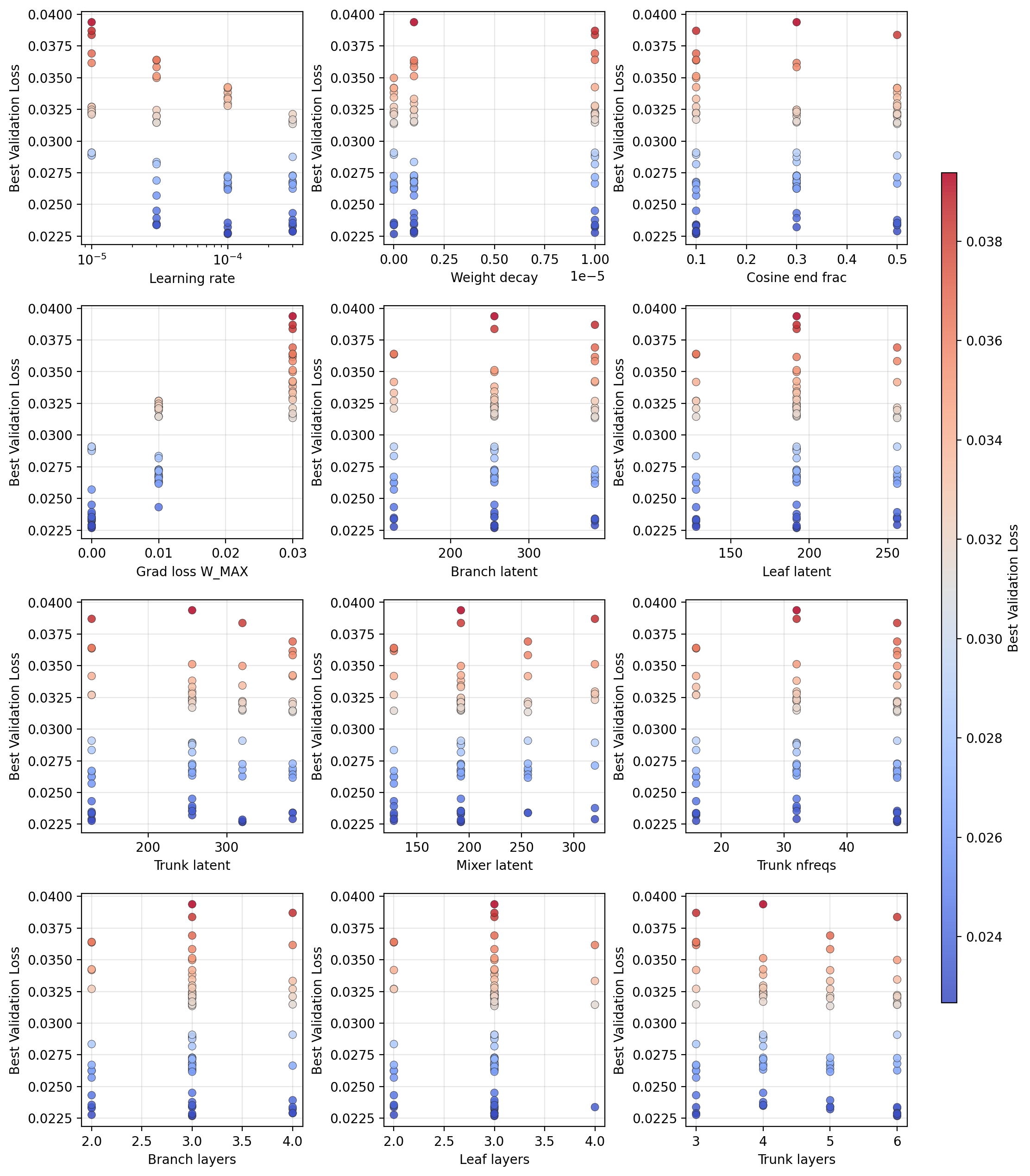}
    \caption{\textbf{Parameter-wise variation in Stage~I GeONet validation
    performance ($N=500$).}
    Best validation loss is shown against the sampled optimization and
    architecture hyperparameters across the Stage~I trials; marker color also
    denotes validation loss. The clearest separation occurs for the learning
    rate and gradient-loss weight, with the lowest losses concentrated at
    $\eta=10^{-4}$--$3\times10^{-4}$ and
    $w_{\mathrm{grad}}=0$. Low-loss trials span multiple latent dimensions
    and network depths, demonstrating that favorable performance is not
    confined to a single model size.}
    \label{fig:geonet_param_500}
\end{figure}

The rank-based analysis provides a complementary global view of these
relationships (Fig.~\ref{fig:geonet_sensitivity_500}a). The gradient-loss
weight was the dominant Stage~I variable
($\rho=+0.74$), followed by learning rate ($\rho=-0.48$); all other
individual correlations had $|\rho|\leq0.11$. The positive correlation for
$w_{\mathrm{grad}}$ indicates increasing validation loss as the gradient
penalty increases, whereas the negative correlation for learning rate reflects
the improved performance at the upper end of the sampled range. Architecture
comparisons (Fig.~\ref{fig:geonet_sensitivity_500}b) show overlapping
distributions and substantial within-family variability. The
\texttt{deep\_trunk} preset produced the overall optimum, while the compact
\texttt{small} model also reached the low-loss regime. The trial sequence
(Fig.~\ref{fig:geonet_sensitivity_500}c) shows that competitive solutions
were recovered repeatedly rather than in one isolated search event. The
learning-rate--weight-decay projection
(Fig.~\ref{fig:geonet_sensitivity_500}d) further confirms that the best trials
occupy several weight-decay settings once an effective learning rate is used.

\begin{figure}[width=\linewidth,pos=htbp]
    \centering
    \includegraphics[width=0.99\textwidth]{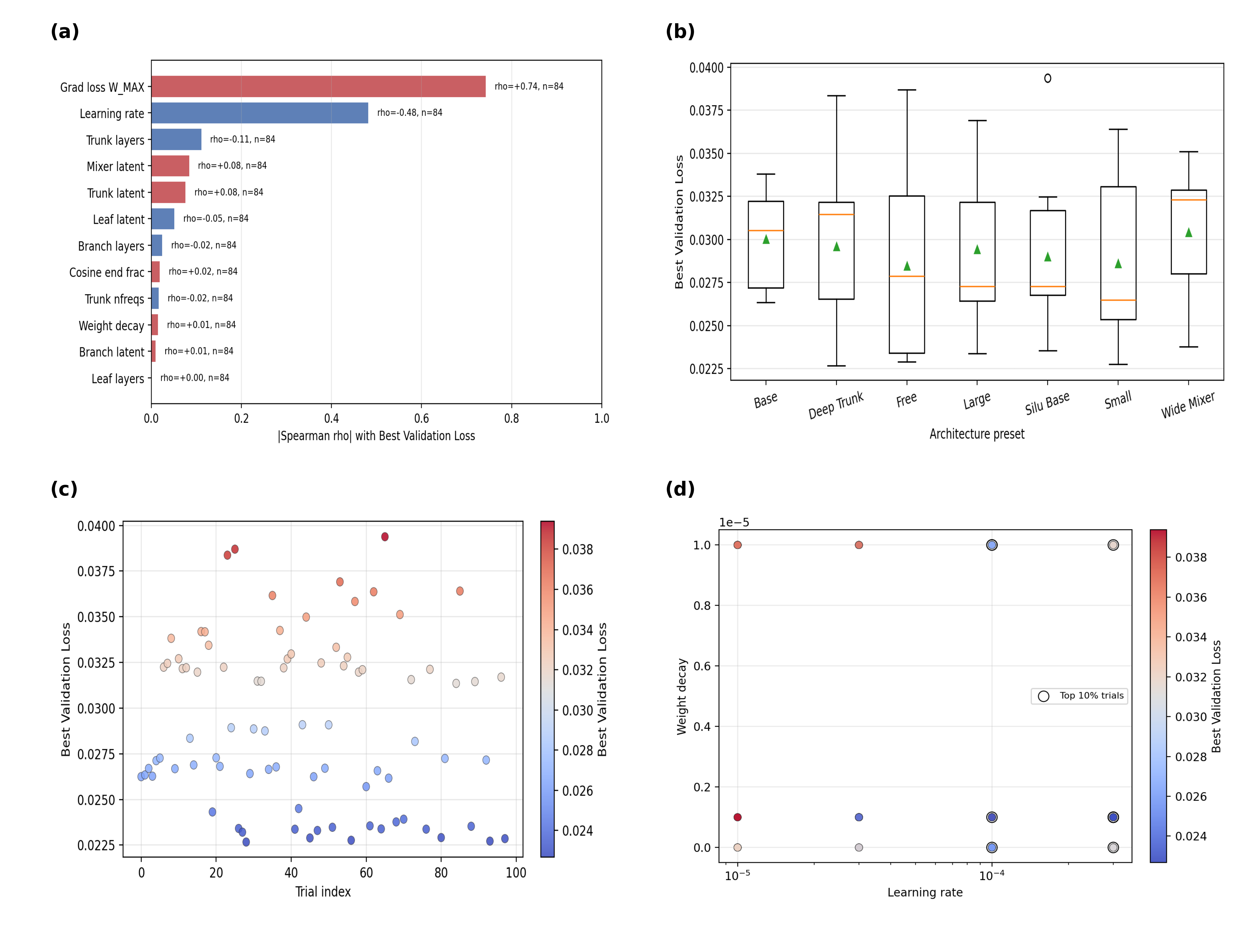}
    \caption{\textbf{Stage~I sensitivity and architecture analysis for
    GeONet ($N=500$).}
    (A) Absolute Spearman rank correlations between each sampled
    hyperparameter and best validation loss; bar color indicates the sign of
    the reported correlation. (B) Validation-loss distributions grouped by
    architecture preset; orange lines and green triangles denote medians and
    means, respectively. (C) Best validation loss across the 100-trial search.
    (D) Joint learning-rate--weight-decay projection, with outlined markers
    identifying the top 10\% of trials. Gradient-loss weight and learning rate
    dominate the Stage~I response, whereas the overlapping architecture
    distributions show that training configuration and architecture interact
    in determining performance.}
    \label{fig:geonet_sensitivity_500}
\end{figure}

Stage~II transferred and refined the five strongest Stage~I configurations
using $N=5000$ realizations. The minimum GeONet validation loss decreased from
$0.02267$ to $0.01944$, a relative reduction of $14.2\%$. The selected model,
trial~011, descended from the Stage~I free-form configuration with branch,
trunk, and mixer widths of 384, 128, and 320, respectively. The transferred
\texttt{deep\_trunk} configuration remained close, reaching $0.01951$.
Together with the Stage~II median of $0.02003$, these results demonstrate that
the reduced-data search identified a capacity regime that transferred
effectively to the full set.

\begin{table*}[t]
\centering
\caption{\textbf{HPO-selected configurations for the reduced- and full-data
stages.}
Stage~I champions seed the transferred refinement, and Stage~II champions are
used in the subsequent full-data experiments. Latent dimensions are reported
as $B/L/T/M$ for the branch, leaf, trunk, and mixer networks; DeepONet does
not contain a leaf network. WD denotes weight decay.}
\label{tab:hpo_best_configs}
\setlength{\tabcolsep}{3.5pt}
\small
\begin{tabular}{@{}llllllll@{}}
\toprule
Stage & Model & Trial & Preset & Latents $B/L/T/M$
 & $\eta$ / WD / $n_f$ / $\sigma_f$ & Best val & Epoch \\
\midrule
I & GeONet & $028$ & \texttt{deep\_trunk}
  & $256/192/320/192$
  & $10^{-4}$ / $0$ / $48$ / $20$
  & $0.02267$ & $85$ \\
I & DeepONet & $040$ & \texttt{free}
  & $256$/--/$384/256$
  & $3\times10^{-4}$ / $0$ / $0$ / $1$
  & $0.61911$ & $93$ \\
\midrule
II & GeONet & $011$ & \texttt{free}\textsuperscript{\textdagger}
  & $384/128/128/320$
  & $10^{-4}$ / $0$ / $48$ / $10$
  & $0.01944$ & $89$ \\
II & DeepONet & $001$ & \texttt{free}\textsuperscript{\textdaggerdbl}
  & $128$/--/$384/256$
  & $3\times10^{-4}$ / $10^{-5}$ / $4$ / $1$
  & $0.56718$ & $97$ \\
\bottomrule
\end{tabular}
\vspace{0.35em}
\begin{minipage}{0.98\textwidth}
\footnotesize
\textsuperscript{\textdagger}Inherited from Stage~I GeONet free-form
trial~045 and refined at $N=5000$.\quad
\textsuperscript{\textdaggerdbl}Inherited from the neighborhood of
Stage~I DeepONet free-form trial~052 and refined at $N=5000$ with the
boundary-condition loss retained. All selected runs use GELU activations.
\end{minipage}
\end{table*}

The Stage~II parameter projections show a compact low-loss band near
$0.0195$--$0.0203$ and a smaller group of clearly inferior trials
(Fig.~\ref{fig:geonet_param_5000}). The largest losses occur at the lowest
learning rate, $10^{-5}$, and when $w_{\mathrm{grad}}=0.01$; the best trials
again use $w_{\mathrm{grad}}=0$ and learning rates of
$10^{-4}$--$3\times10^{-4}$. Several transferred widths and depths reach the
same low-loss band, showing that full-data performance is robust within the
restricted architecture pool rather than dependent on one narrowly tuned
dimension.

\begin{figure}[width=\linewidth,pos=htbp]
    \centering
    \includegraphics[width=0.92\textwidth]{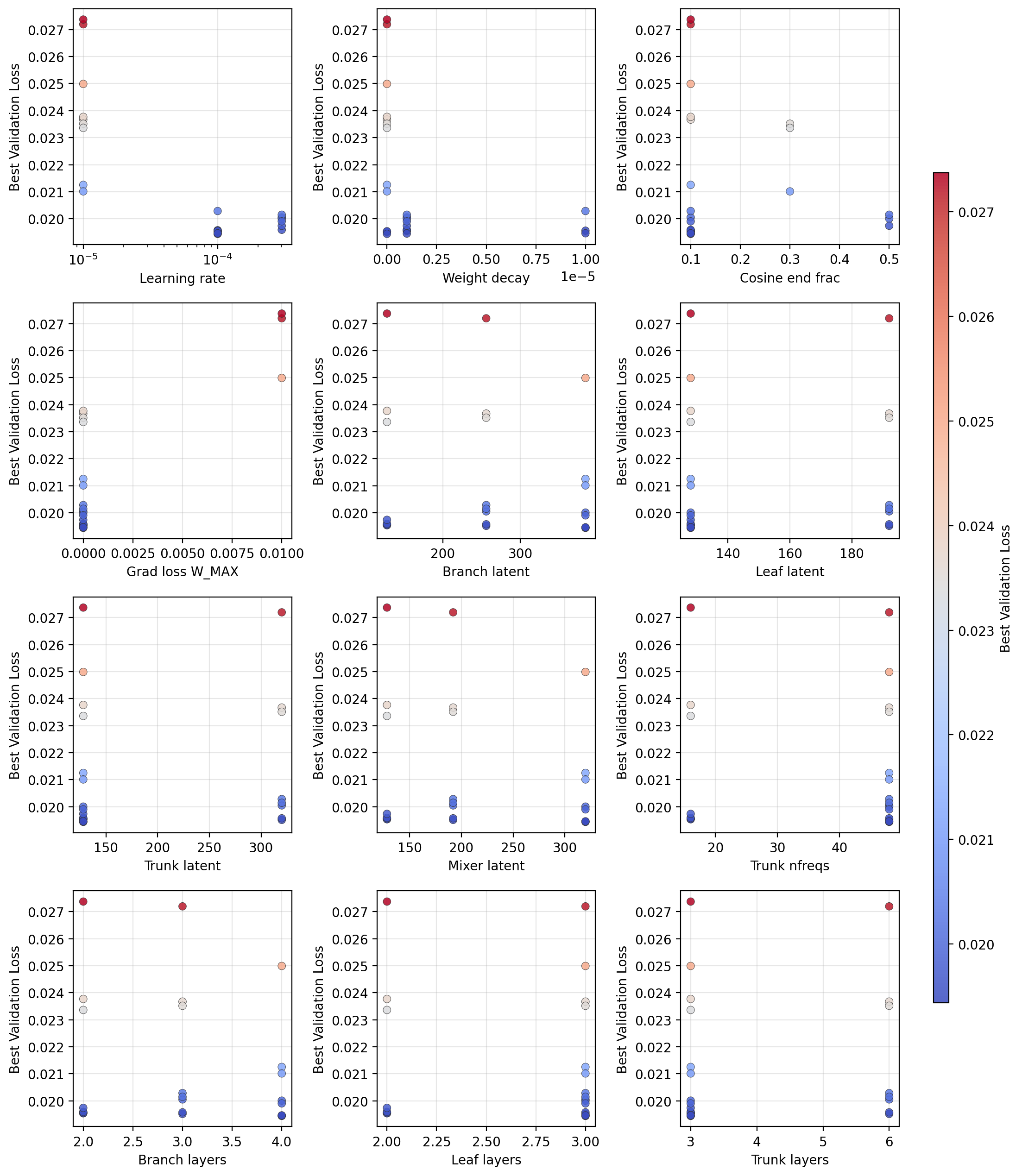}
    \caption{\textbf{Parameter-wise variation in Stage~II GeONet validation
    performance ($N=5000$).}
    Best validation loss is plotted against the optimization and architecture
    settings evaluated during transferred refinement; marker color denotes the
    same loss. Most configurations form a low-loss band near
    $0.0195$--$0.0203$. The principal degradations occur for
    $\eta=10^{-5}$ and $w_{\mathrm{grad}}=0.01$, while multiple latent
    widths and depths remain competitive within the transferred architecture
    pool.}
    \label{fig:geonet_param_5000}
\end{figure}

The Stage~II sensitivity summary reinforces this interpretation
(Fig.~\ref{fig:geonet_sensitivity_5000}). Gradient-loss weight and learning
rate remain the leading variables, with correlations of $+0.66$ and $-0.63$,
respectively (Fig.~\ref{fig:geonet_sensitivity_5000}a). Correlations of
$+0.56$ for leaf width, trunk width, and trunk depth should be interpreted
within the restricted, transferred pool rather than as independent global
effects, because only 24 locally refined trials were evaluated and several
settings co-vary by construction. Among the transferred families, the
free-form configurations exhibit the lowest and tightest distribution
(Fig.~\ref{fig:geonet_sensitivity_5000}b). The trial trajectory
(Fig.~\ref{fig:geonet_sensitivity_5000}c) contains a broad cluster of
near-optimal solutions, and the two-parameter projection
(Fig.~\ref{fig:geonet_sensitivity_5000}d) places the leading trials within
the same stable optimizer region identified during Stage~I.

\begin{figure}[width=\linewidth,pos=htbp]
    \centering
    \includegraphics[width=0.99\textwidth]{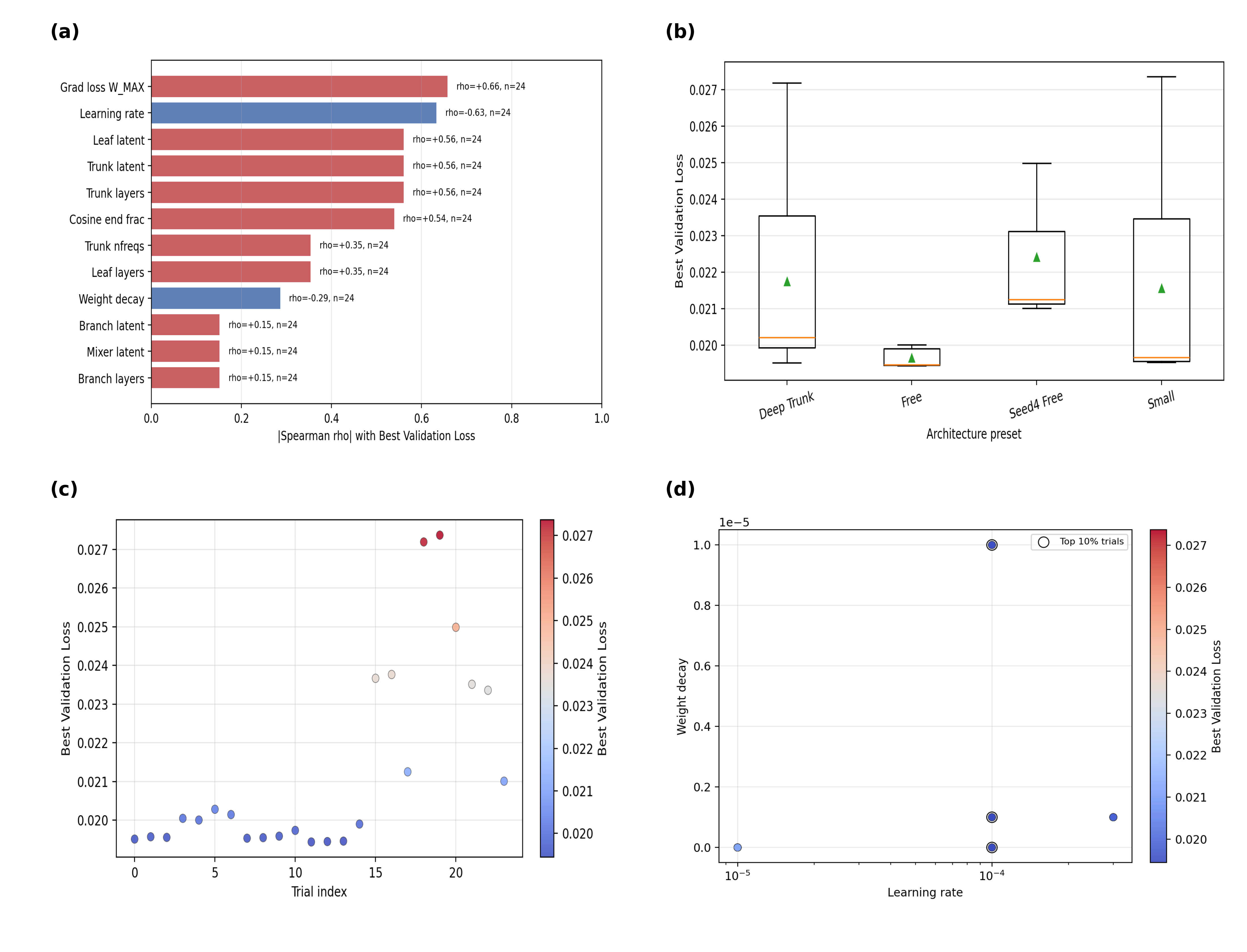}
    \caption{\textbf{Stage~II sensitivity and architecture analysis for
    GeONet ($N=5000$).}
    (A) Absolute Spearman rank correlations between the locally varied
    hyperparameters and best validation loss; bar color indicates correlation
    sign. (B) Validation-loss distributions for the transferred architecture
    families, with medians and means shown by orange lines and green triangles,
    respectively. (C) Best validation loss across the 24 refinement trials.
    (D) Learning-rate--weight-decay projection with the top 10\% of trials
    outlined. The results identify a broad near-optimal regime and confirm
    that learning rate and gradient-loss weight remain the dominant
    optimization controls after transfer. Correlations are descriptive of the
    restricted Stage~II pool and are not interpreted as independent causal
    effects.}
    \label{fig:geonet_sensitivity_5000}
\end{figure}

DeepONet also benefited from transfer to the full set. Its best
validation loss decreased from $0.61911$ in Stage~I to $0.56718$ in Stage~II,
an $8.4\%$ relative reduction. The Stage~II champion, trial~001, retained a
wide trunk (latent dimension 384), mild random Fourier features
($n_f=4$, $\sigma_f=1$), and a learning rate of $3\times10^{-4}$.
Thus, unlike the classical high-RFF DeepONet baseline considered previously,
the BC-focused family continued to benefit from additional realizations,
although its absolute validation loss remained substantially higher than that
of GeONet. The complete diagnostics are provided in
Appendix~\ref{app:deeponet_hpo}. Unless otherwise stated, subsequent
full-data experiments use GeONet Stage~II trial~011 and DeepONet Stage~II trial~001. The computational cost further
supports the two-stage design: 100 Stage~I trials required approximately
30 GPU-hours for GeONet and 18 GPU-hours for DeepONet, whereas 24 Stage~II
trials required approximately 88 and 63 GPU-hours, respectively.

\begin{table}[t]
\centering
\caption{\textbf{Transfer of the selected model families from Stage~I
($N=500$) to Stage~II ($N=5000$).}
The relative change is calculated from the best validation loss, with negative
values indicating improvement. Both models benefit from full-data refinement,
while GeONet retains a substantial absolute advantage.}
\label{tab:hpo_transfer}
\begin{tabular}{@{}lrrr@{}}
\toprule
Model & Best val ($N=500$) & Best val ($N=5000$) & Relative change \\
\midrule
GeONet      & $0.02267$ & $0.01944$ & $-14.2\%$ \\
DeepONet & $0.61911$ & $0.56718$ & $-8.4\%$ \\
\bottomrule
\end{tabular}
\end{table}

\subsection{Training Dynamics of the Best Models}
\label{Sec:results_training}

Figure~\ref{fig:training_curves} compares the training and validation loss trajectories of the HPO-selected GeONet and DeepONet models on the full-dataset. The curves are displayed over the same 50-epoch window and on a common logarithmic loss scale, allowing their convergence behavior and loss magnitudes to be compared directly. These models correspond to the full-data configurations selected through the two-stage HPO procedure summarized in Tables~\ref{tab:hpo_stage_summary} and~\ref{tab:hpo_best_configs}.
GeONet exhibits a rapid reduction in training loss during the first few epochs, followed by a slower refinement regime in which the loss continues to decrease toward approximately $2\times10^{-2}$ (Fig.~\ref{fig:training_curves}a). Its validation trajectory follows a similar pattern: after the initial decrease, the loss remains close to the training curve and fluctuates within a narrow low-loss range (Fig.~\ref{fig:training_curves}b). This behavior is consistent with the Stage~II HPO results, for which the selected GeONet configuration attained a minimum validation loss of $0.01944$ (Table~\ref{tab:hpo_best_configs}). The close proximity of the training and validation trajectories after the initial transient indicates stable optimization without a pronounced divergence between fitting and held-out performance over the displayed training interval.

DeepONet exhibits a markedly different optimization regime. Its training and validation objectives decrease only modestly after the initial epochs and subsequently remain near a substantially higher plateau of approximately $0.55$--$0.60$ over the displayed interval (Fig.~\ref{fig:training_curves}). This behavior is consistent with the Stage~II HPO results, where the selected DeepONet configuration attained a minimum validation loss of $0.56718$ (Table~\ref{tab:hpo_best_configs}). The Stage~I and Stage~II searches both showed this separation: increasing the training cohort from $N=500$ to $N=5000$ reduced the best validation loss from $0.02267$ to $0.01944$ for GeONet and from $0.61911$ to $0.56718$ for DeepONet (Table~\ref{tab:hpo_transfer}). Thus, the different loss regimes observed in Fig.~\ref{fig:training_curves} are consistent with the broader HPO results rather than being specific to a single training trajectory.
The absolute loss magnitudes should nevertheless be interpreted in the context of the model-specific training objectives. The selected GeONet configuration uses $w_{\mathrm{grad}}=0$, such that its objective reduces to the pressure MSE defined in Eq.~\ref{eq:pressure_mse}. DeepONet, in contrast, is trained with the pressure MSE together with the Dirichlet boundary-condition penalty defined in Eq.~\ref{eq:deeponet_loss}. Consequently, Fig.~\ref{fig:training_curves} is most informative as a comparison of optimization behavior and the objective values attained by the HPO-selected configurations, rather than as a direct decomposition of pressure-prediction error alone. Prediction accuracy in physical units is evaluated separately using the test-set metrics and spatial error analyses presented in the subsequent sections.

\begin{figure}[width=\linewidth,pos=htbp]
    \centering
    \includegraphics[width=0.99\textwidth]{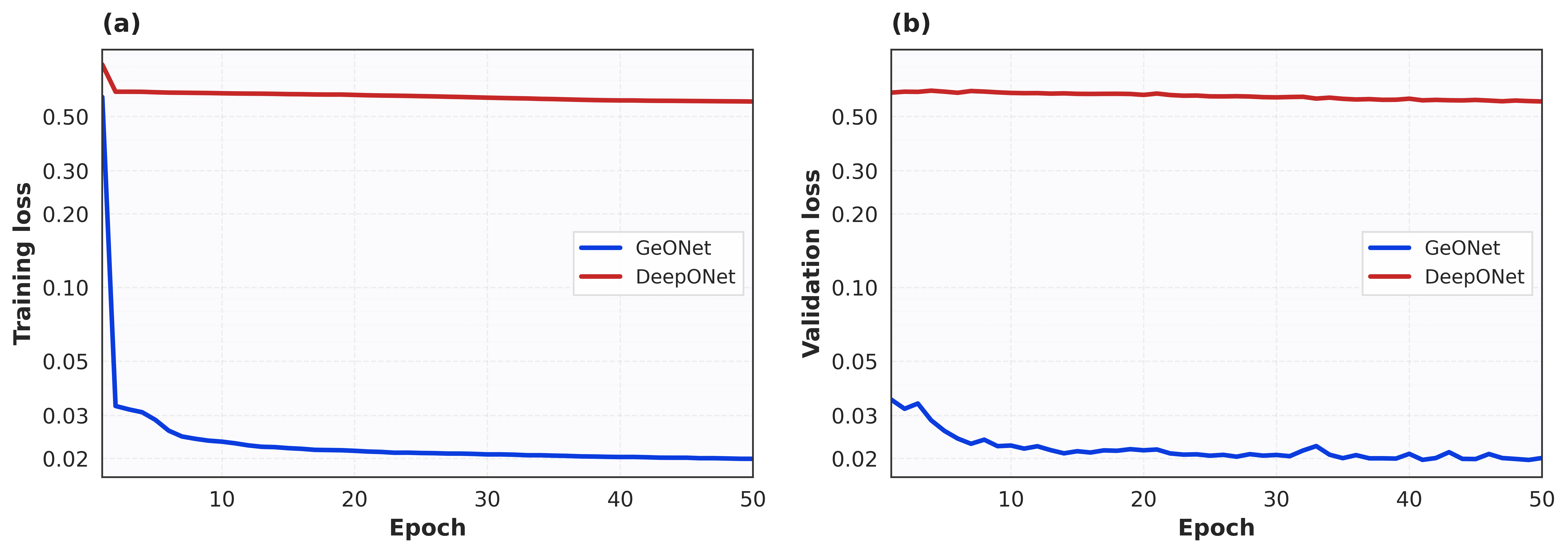}
    \caption{\textbf{Training dynamics of the HPO-selected GeONet and DeepONet models.} (a) Training loss and (b) validation loss over the first 50 epochs of full-data training. Both panels use the same logarithmic loss scale to facilitate comparison of convergence behavior and objective magnitude. GeONet undergoes a rapid initial decrease followed by gradual refinement in a low-loss regime, whereas DeepONet reaches a substantially higher loss plateau. Because the selected models use different training objectives, the absolute loss values should be interpreted as model-specific optimization objectives rather than as directly equivalent pressure-error metrics.}
    \label{fig:training_curves}
\end{figure}

\subsection{GeONet preserves geometry-induced pressure structure}
\label{Sec:results_pressure_prediction}
\begin{figure}[width=\linewidth,pos=htbp]
    \centering
    \includegraphics[width=0.99\textwidth]
    {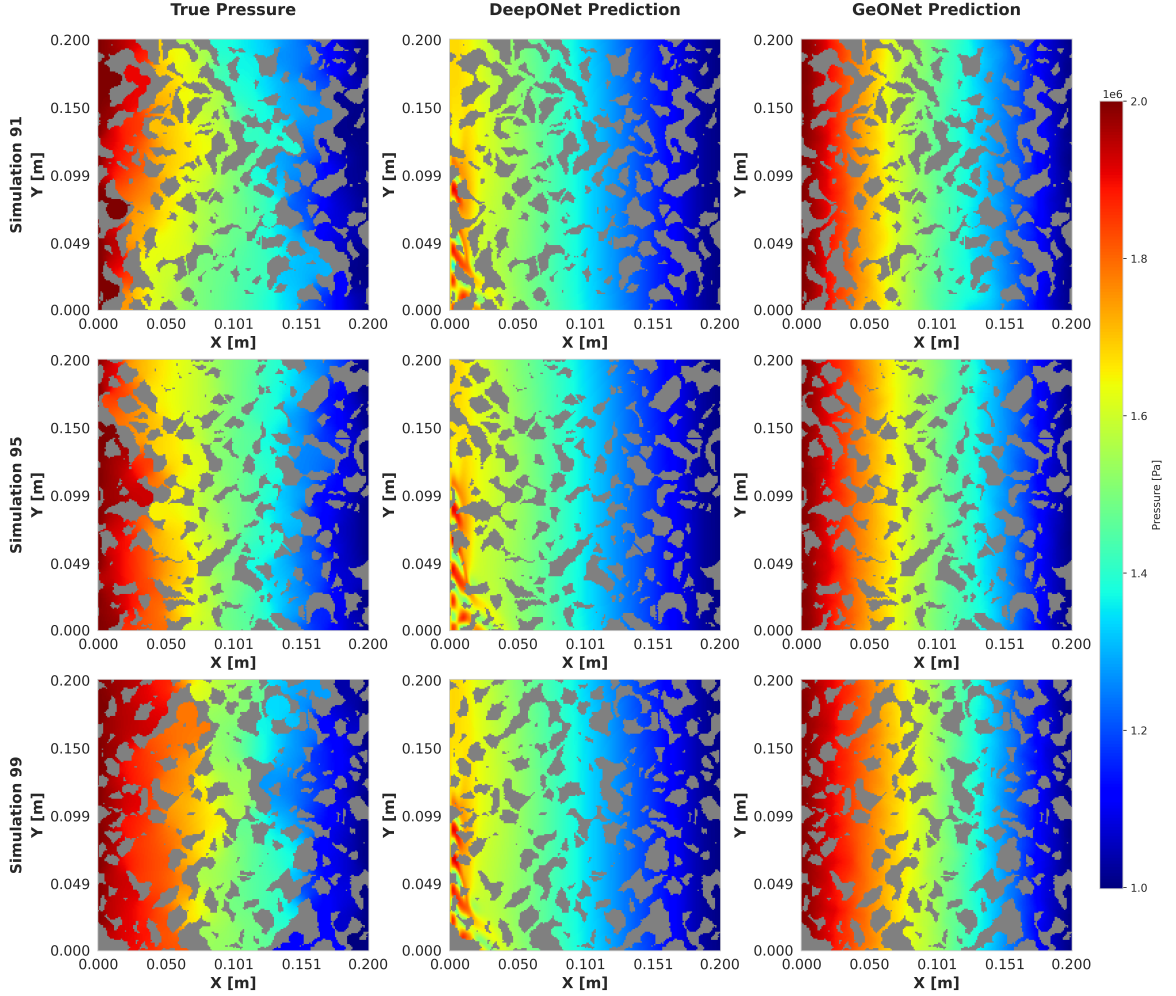}
    \caption{\textbf{Pressure-field reconstruction for three unseen heterogeneous RCP domains.}
    PFLOTRAN reference solutions (left column) are compared with DeepONet (center column) and GeONet (right column) predictions for test simulations 91, 95, and 99, shown from top to bottom. Gray regions denote the solid phase and are excluded from the pressure comparison. All panels use the same pressure scale, ranging from 1.0 to 2.0~MPa. DeepONet reproduces the overall inlet-to-outlet pressure decline but underestimates the spatial extent and magnitude of the high-pressure region and introduces localized artifacts near the inlet boundary. GeONet more accurately reconstructs the pressure magnitude, irregular contours, and geometry-dependent spatial variations observed in the PFLOTRAN solutions across all three test domains.}
    \label{fig:pressure_comparison}
\end{figure}

Figure~\ref{fig:pressure_comparison} compares the predicted pressure fields with the corresponding PFLOTRAN solutions for three unseen RCP geometries. Each realization exhibits a general pressure decline from the high-pressure boundary on the left to the low-pressure boundary on the right. However, the detailed shape of this decline varies among the domains because the connected pore space, solid boundaries, and local constrictions redistribute the pressure gradient.

DeepONet captures the dominant left-to-right pressure trend but does not reproduce the full spatial extent of the high-pressure region. Across the three simulations, its interior predictions are shifted toward lower pressure values, producing broader green and cyan regions where the PFLOTRAN solutions retain red, orange, or yellow pressure levels. The DeepONet fields also contain localized high-pressure artifacts along portions of the inlet boundary, most visibly in simulations 95 and 99. Although the predicted contours respond partly to the domain mask, their pressure magnitudes and spatial transitions differ appreciably from the reference solutions.
In contrast, GeONet closely reproduces both the magnitude and spatial organization of the PFLOTRAN pressure fields. The high-pressure regions near the inlet extend through the connected pore space with shapes comparable to the reference solutions, while the intermediate-pressure contours follow the irregular solid--void boundaries across the domain. GeONet also preserves localized variations associated with constricted and tortuous flow pathways without introducing the pronounced inlet artifacts observed in the DeepONet predictions. The consistency of this behavior across three geometrically distinct test domains demonstrates that explicit geometry encoding improves the reconstruction of pressure fields governed by pore-network structure.

\subsection{Distribution and profile analyses show reduced bias}
\label{Sec:results_distribution_error}
Figure~\ref{fig:pressure_distribution} compares histograms from the reference PFLOTRAN simulations, the GeONet predictions, and the DeepONet predictions.
The GeONet histograms preserve the broad pressure range from approximately $1.0\times10^6$ to $2.0\times10^6$~Pa, including the low- and high-pressure tails.
The DeepONet histograms are shifted toward lower pressures and contain substantially fewer values above approximately $1.65\times10^6$~Pa, indicating that the baseline model underrepresents the high-pressure region and compresses the solution variability.

\begin{figure}[width=\linewidth,pos=htbp]
    \centering
    \includegraphics[width=0.99\textwidth]
    {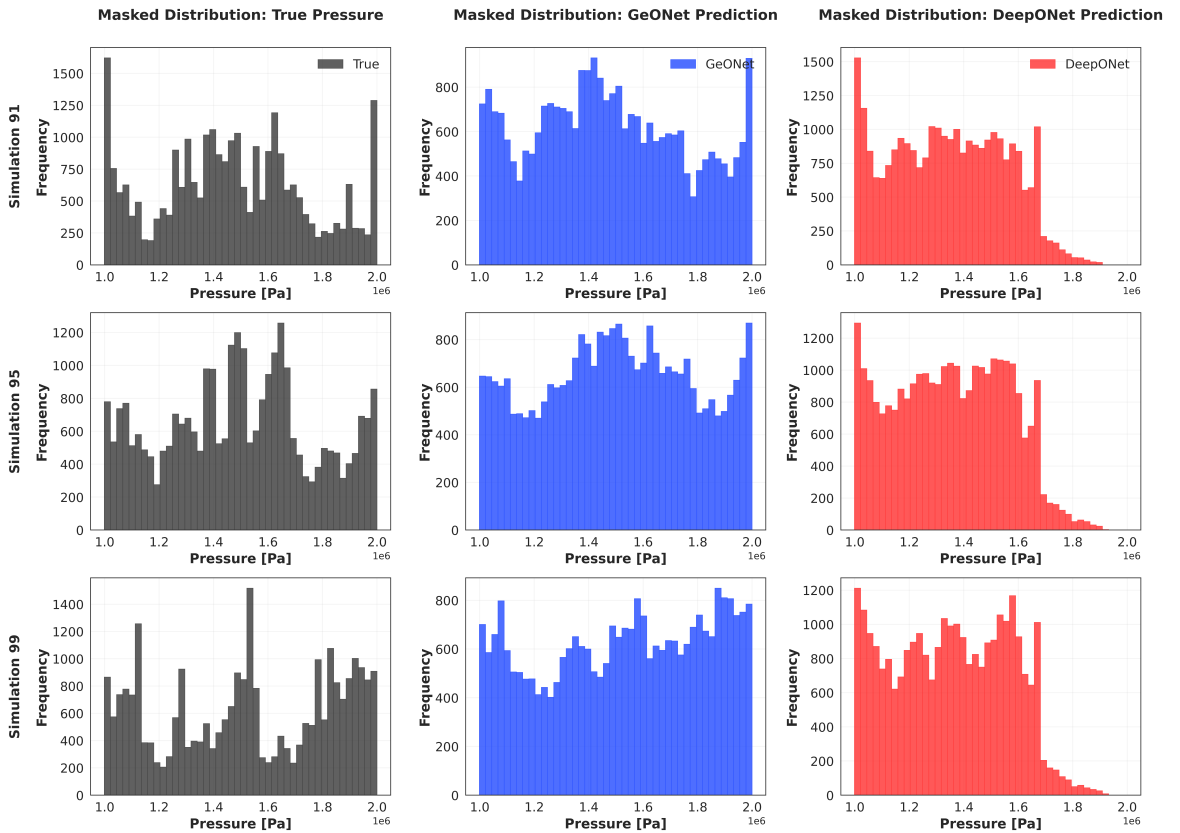}
    \caption{\textbf{Pressure-value distributions for PFLOTRAN, GeONet, and DeepONet.}
    The top, middle, and bottom rows correspond to simulations 91, 95, and 99, respectively, while the columns show the pore-space pressure distributions for the PFLOTRAN reference solution, GeONet prediction, and DeepONet prediction. Soil-matrix cells are masked and excluded from all distributions.
    Across all three simulations, GeONet preserves the broad reference pressure range and more closely reproduces the heterogeneous distribution structure.
    In contrast, DeepONet shifts the distributions toward lower pressures and substantially underrepresents the high-pressure portion of the PFLOTRAN solutions.}
    \label{fig:pressure_distribution}
\end{figure}

The top, middle, and bottom rows correspond to simulations 91, 95, and 99, respectively, while the columns show the pore-space pressure distributions for the PFLOTRAN reference solution, GeONet prediction, and DeepONet prediction. Solid cells are masked and excluded from all distributions.

\begin{table}[t]
    \centering
    \caption{Pressure-distribution summary statistics calculated using values within the fluid region. Mean pressure and standard deviation are normalized by the corresponding mean reference pressure, such that the normalized reference mean equals 1.}
    \label{tab:fig3c_distribution_stats}
    \begin{tabular}{lcccccc}
        \toprule
        & \multicolumn{2}{c}{True}
        & \multicolumn{2}{c}{GeONet}
        & \multicolumn{2}{c}{DeepONet} \\
        \cmidrule(lr){2-3}
        \cmidrule(lr){4-5}
        \cmidrule(lr){6-7}
        Simulation & Mean & Std. dev. & Mean & Std. dev. & Mean & Std. dev. \\
        \midrule
        91 & 1.0000 & 0.1887 & 1.0101 & 0.1912 & 0.9155 & 0.1376 \\
        95 & 1.0000 & 0.1834 & 0.9992 & 0.1923 & 0.9032 & 0.1337 \\
        99 & 1.0000 & 0.2155 & 1.0232 & 0.2236 & 0.9089 & 0.1417 \\
        \bottomrule
    \end{tabular}
\end{table}

Figure~\ref{fig:error_analysis} further compares the prediction-error distributions and one-dimensional pressure profiles.
Across all three simulations, the GeONet error distributions are more tightly concentrated near zero, whereas the DeepONet errors are broader and predominantly shifted toward negative values.
This shift confirms that DeepONet systematically underestimates the pressure magnitude.
Along the horizontal midline at $Y=0.10$~m, both models reproduce the overall inlet-to-outlet pressure decline.
However, the DeepONet profiles remain below the GeONet profiles across most connected fluid segments, consistent with the negative bias observed in the error distributions.
Along the vertical midline at $X=0.10$~m, GeONet retains the higher pressure levels associated with the reference fields, whereas DeepONet substantially underestimates the pressure across all three simulations.
Gaps in the profiles correspond to intersections of the midlines with masked solid regions.

\begin{figure}[width=\linewidth,pos=htbp]
    \centering
    \includegraphics[width=0.99\textwidth]
    {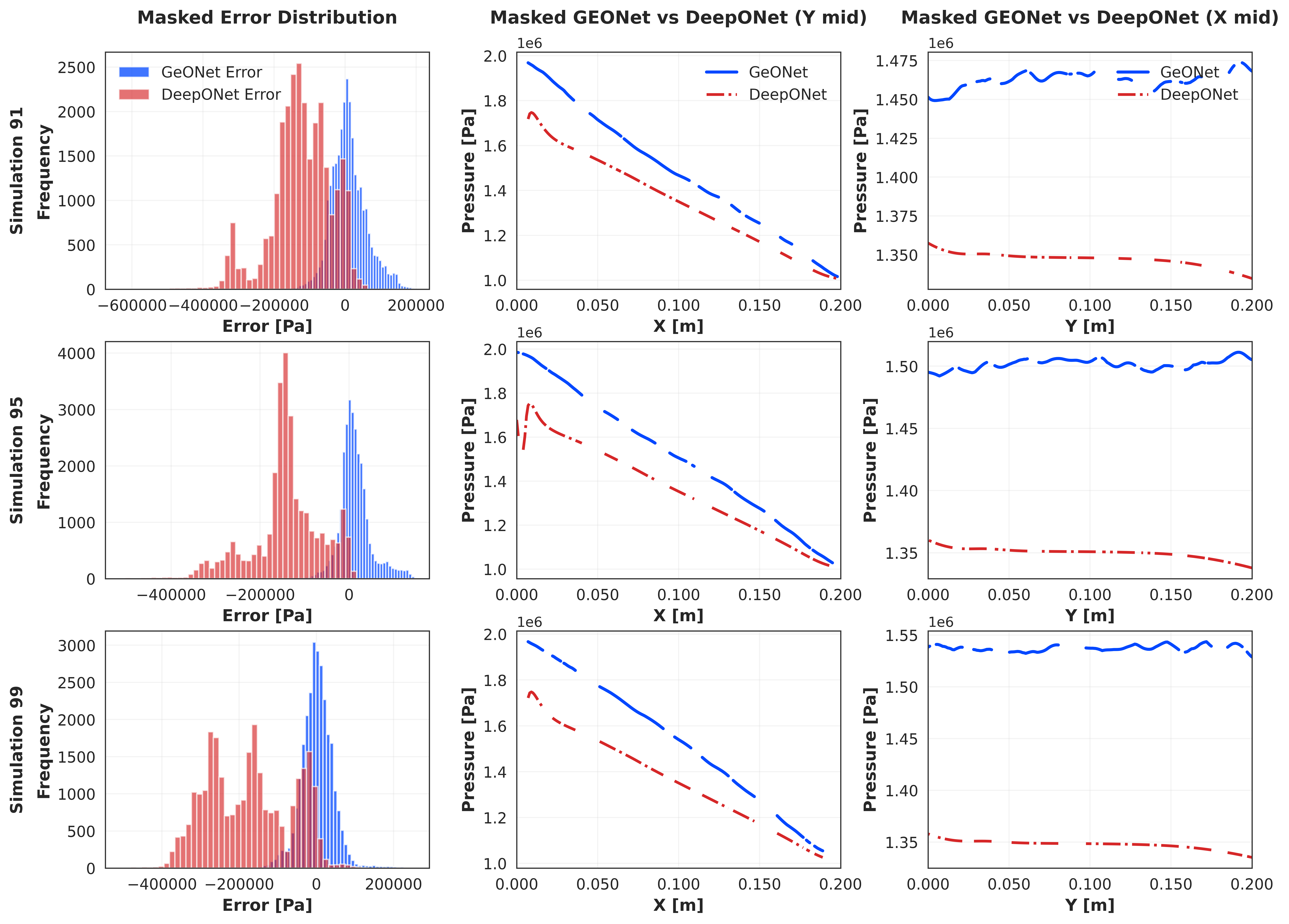}
    \caption{\textbf{Error distributions and pressure-profile comparisons.}
    The top, middle, and bottom rows correspond to simulations 91, 95, and 99, respectively.
    The left column shows histograms of prediction errors relative to the PFLOTRAN reference solutions.
    GeONet errors remain more tightly concentrated near zero, while DeepONet errors span a broader range and exhibit a predominantly negative bias.
    The middle column compares GeONet and DeepONet pressure profiles along the horizontal domain midline, $Y=0.10$~m, and the right column presents the corresponding profiles along the vertical midline, $X=0.10$~m.
    Gaps in the profiles denote masked solid regions.
    Across all three test simulations, DeepONet systematically underestimates pressure relative to GeONet, particularly along the vertical midline.}
    \label{fig:error_analysis}
\end{figure}

Together, the distribution and profile results show that GeONet does more than reduce the average error.
It preserves the physically important portions of the solution that are most strongly affected by geometry, including the broad pressure range, localized variations, and heterogeneous pressure distributions.

\begin{figure}[width=\linewidth,pos=htbp]
    \centering
    \includegraphics[width=0.75\textwidth]
    {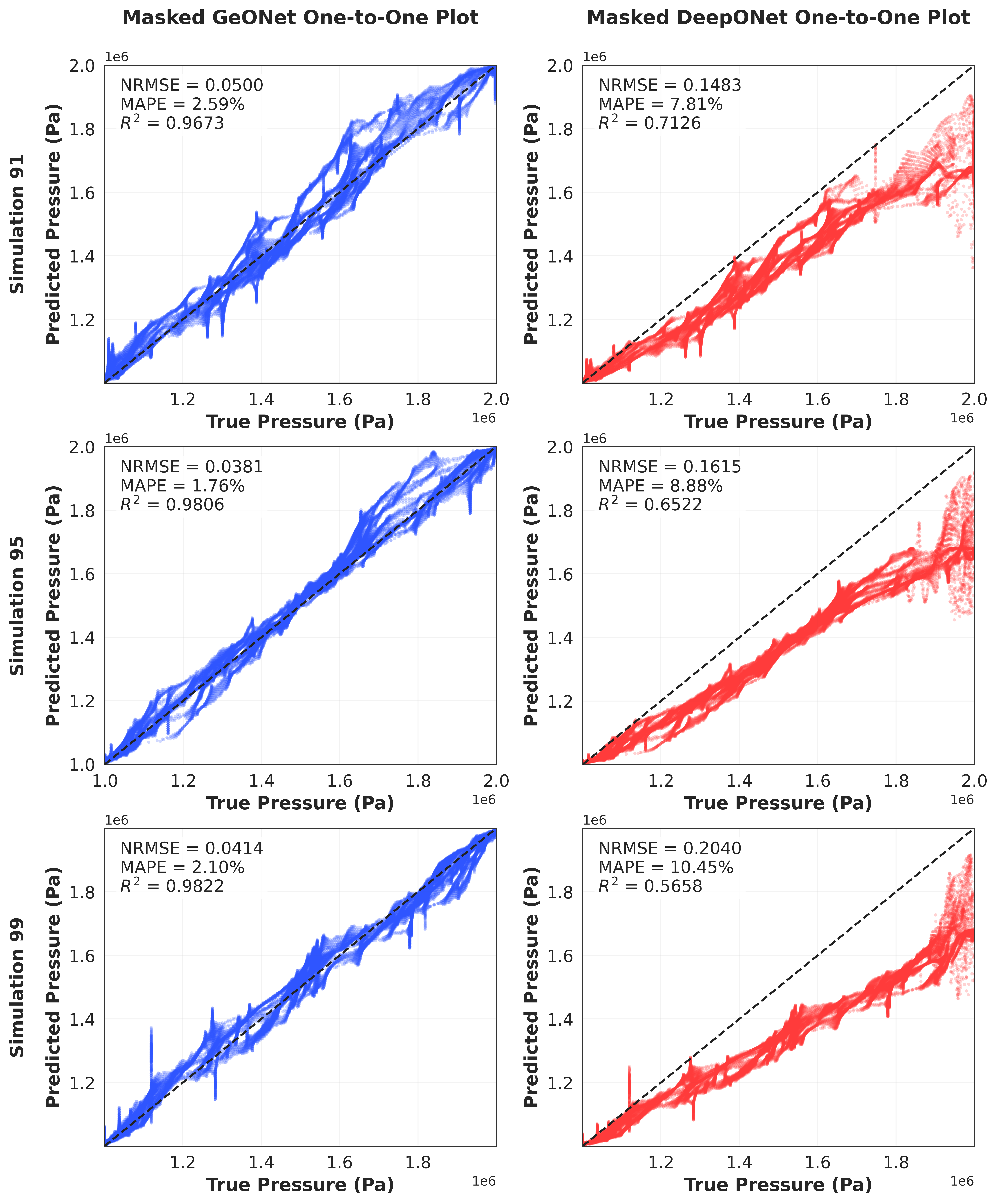}
    \caption{\textbf{One-to-one comparison of predicted and reference pressures for GeONet and DeepONet.}
    The top, middle, and bottom rows correspond to simulations 91, 95, and 99, respectively, while the left and right columns show GeONet and DeepONet predictions.
    The black dashed line represents perfect agreement between the predicted and PFLOTRAN reference pressures.
    GeONet predictions remain closely distributed around the one-to-one line across all three simulations, resulting in lower NRMSE and MAPE values and higher $R^2$.
    In contrast, DeepONet increasingly underpredicts pressure at higher reference values and exhibits weaker overall agreement with the PFLOTRAN solutions.}
    \label{fig:one-to-one_analysis}
\end{figure}

\begin{table}[t]
    \centering
    \caption{One-to-one agreement metrics for Figure~\ref{fig:one-to-one_analysis}, evaluated within the fluid region. MAPE is reported as a percentage.}
    \label{tab:fig4c_one_to_one_metrics}
    \begin{tabular}{lcccccc}
        \toprule
        & \multicolumn{3}{c}{GeONet}
        & \multicolumn{3}{c}{DeepONet} \\
        \cmidrule(lr){2-4}
        \cmidrule(lr){5-7}
        Simulation & NRMSE & MAPE (\%) & $R^2$
                   & NRMSE & MAPE (\%) & $R^2$ \\
        \midrule
        91 & 0.0500 & 2.59 & 0.9673 & 0.1483 & 7.81  & 0.7126 \\
        95 & 0.0381 & 1.76 & 0.9806 & 0.1615 & 8.88  & 0.6522 \\
        99 & 0.0414 & 2.10 & 0.9822 & 0.2040 & 10.45 & 0.5658 \\
        \bottomrule
    \end{tabular}
\end{table}

The one-to-one comparisons reinforce the distribution and profile analyses.
GeONet remains closely aligned with the reference pressure across the full range, with NRMSE values of 0.0381--0.0500, MAPE values of 1.76--2.59\%, and $R^2$ values of 0.9673--0.9822.
DeepONet increasingly deviates below the identity line as the reference pressure increases, confirming that its largest errors arise in the high-pressure portion of the solution.

\subsection{GeONet generalizes to unseen RCP geometries}
\label{Sec:results_generalization}
A useful geometry-encoded operator must generalize to chip layouts that were not observed during training.
Figure~\ref{fig:cross_domain_1} evaluates GeONet on three higher-index RCP geometries held out from model training and selection.
In all three cases, the predicted pressure fields reproduce the global inlet-to-outlet pressure drop and the local pressure variations induced by the pore structure.

\begin{figure}[width=\linewidth,pos=htbp]
    \centering
    \includegraphics[width=0.99\textwidth]
    {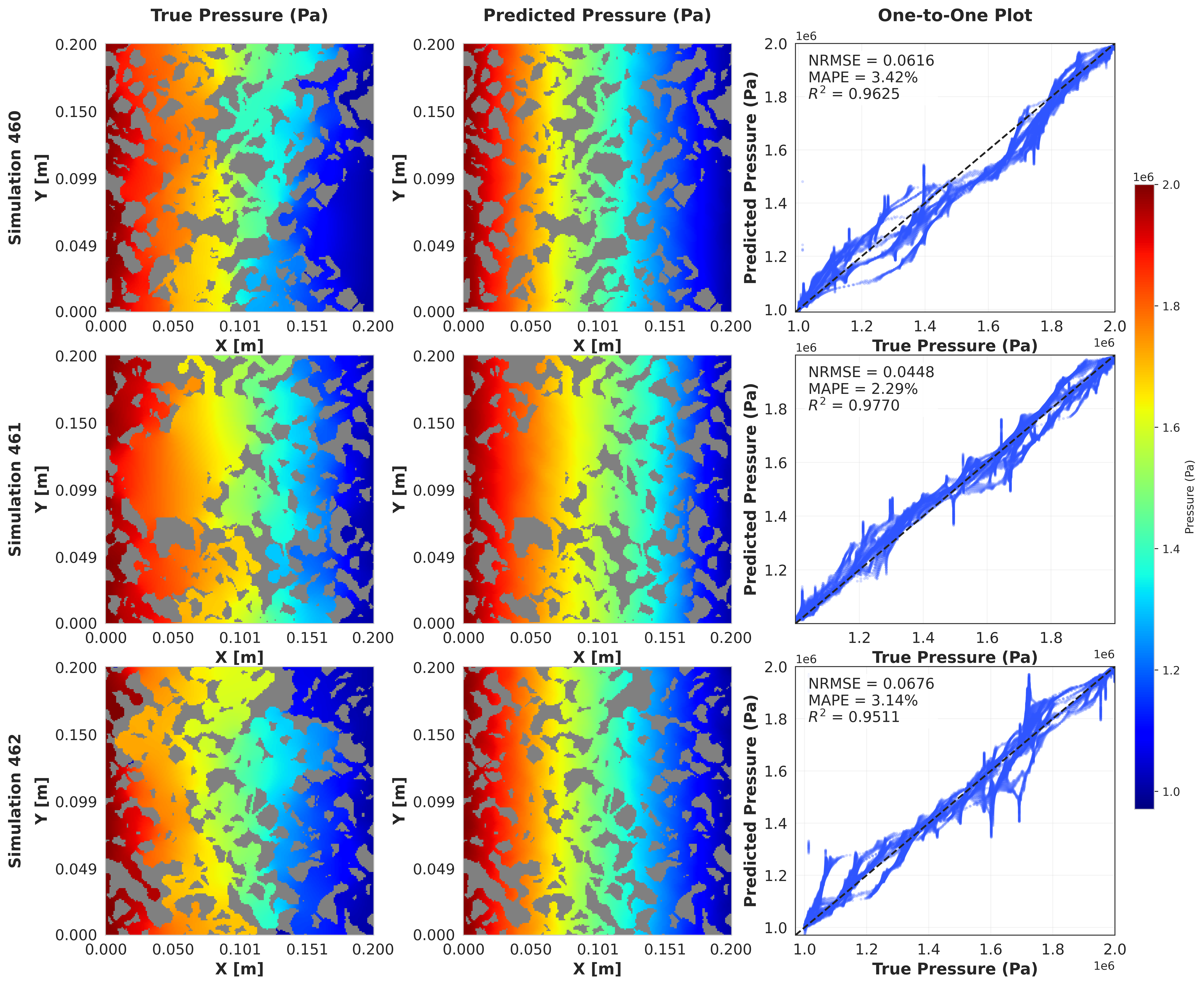}
    \caption{\textbf{Pressure reconstruction for higher-index held-out RCP geometries.}
    The top, middle, and bottom rows correspond to simulations 460, 461, and 462, respectively.
    The left and middle columns show the PFLOTRAN reference pressure fields and GeONet predictions, while the right column presents the corresponding one-to-one comparisons evaluated within the fluid region.
    Gray regions denote the solid phase.
    GeONet accurately reconstructs the dominant pressure gradients, irregular contours, and geometry-induced spatial variations across the three held-out RCP domains.
    The predicted pressures remain closely distributed around the one-to-one line, with NRMSE values of 0.0448--0.0676, MAPE values of 2.29--3.42\%, and $R^2$ values of 0.9511--0.9770.}
    \label{fig:cross_domain_1}
\end{figure}

\begin{table}[t]
    \centering
    \caption{Prediction-agreement metrics for the higher-index held-out RCP geometries shown in Figure~\ref{fig:cross_domain_1}. MAPE is reported as a percentage.}
    \label{tab:fig5c_other_test_metrics}
    \begin{tabular}{lccc}
        \toprule
        Simulation & NRMSE & MAPE (\%) & $R^2$ \\
        \midrule
        460 & 0.0616 & 3.42 & 0.9625 \\
        461 & 0.0448 & 2.29 & 0.9770 \\
        462 & 0.0676 & 3.14 & 0.9511 \\
        \bottomrule
    \end{tabular}
\end{table}

For simulation 460, GeONet achieves NRMSE $=0.0616$, MAPE $=3.42\%$, and $R^2=0.9625$.
The predicted field reproduces the overall pressure range and the geometry-dependent progression of the pressure contours, although localized deviations are visible in the one-to-one comparison.
GeONet produces its strongest performance for simulation 461, with NRMSE $=0.0448$, MAPE $=2.29\%$, and $R^2=0.9770$.
The one-to-one distribution remains closely aligned with the diagonal across most of the pressure range, and the predicted field accurately preserves the broad high-pressure region and its transition toward the outlet.
For simulation 462, GeONet achieves NRMSE $=0.0676$, MAPE $=3.14\%$, and $R^2=0.9511$.
Despite the moderately larger scatter around the one-to-one line, the model preserves the dominant inlet-to-outlet pressure decline, overall pressure range, and geometry-dependent spatial organization of the reference solution.
The consistent performance across these higher-index, structurally distinct geometries indicates that the model has learned a transferable geometry-conditioned operator rather than memorizing individual training domains.
These generalization tests are important because practical LoC and porous-media applications frequently involve new geometries, permeability arrangements, and experimental designs.
The results indicate that GeONet can transfer learned relationships among geometry, material heterogeneity, and pressure response across previously unseen RCP domains.

\subsection{Computational Analysis and Data Efficiency}
\label{sec:computational_cost}
The practical utility of a neural-operator surrogate depends not only on its
predictive accuracy but also on the computational resources required to
generate reference data, train the model, and subsequently deploy it. For
GeONet, the computational workflow involves a substantial upfront investment
in high-fidelity PFLOTRAN simulations followed by the considerably smaller
cost of training the surrogate. We therefore evaluate how GeONet training
scales with dataset size and hardware configuration and examine the resulting
trade-off between computational cost and predictive accuracy. The analysis is
performed independently on the EMSL Tahoma (e.g., using NVIDIA Tesla V100) and PNNL Deception (e.g., using NVIDIA A100) high-performance computing (HPC) platforms using both GPU-accelerated and CPU-only training, providing a cross-platform assessment of computational
scaling, convergence behavior, and data efficiency.

\subsubsection{Data Generation Cost}
\label{subsubsec:data_generation}

The dataset used in this study comprises approximately 5,000 high-resolution
PFLOTRAN simulations of heterogeneous RCP domains and occupies approximately
30~TB of storage. Generating the reference dataset required approximately
13 million CPU core-hours on the EMSL Tahoma HPC system. A representative
PFLOTRAN realization required approximately two days on two compute nodes,
with 36 CPU cores per node. Construction of the supervised dataset therefore
represents a substantial one-time computational investment.
This cost provides an important motivation for developing a surrogate model.
The high-fidelity simulations encode the relationships among pore geometry,
heterogeneous material properties, and the resulting pressure response. Once
these relationships have been learned by GeONet, the trained operator can be
reused for subsequent predictions without repeatedly solving the
high-fidelity flow problem. An important practical consideration is therefore
how much of this computationally expensive training dataset is required to
obtain useful predictive accuracy. The following experiments address this
question by jointly examining training cost and test-set performance as the
training-data fraction is varied.

\subsubsection{Training Computational Cost}
\label{subsubsec:training_cost}

GeONet was trained on the EMSL Tahoma and PNNL Deception HPC platforms using
both GPU-accelerated and CPU-only configurations. The implementation uses JAX
with Flax and Optax, together with just-in-time (JIT) compilation and XLA
optimization. To isolate the effects of dataset size and computational
hardware, all experiments used the same model configuration and optimization
protocol, including AdamW optimization, cosine-decay learning-rate scheduling
with warmup, gradient clipping, coordinate sampling, and a maximum training
budget of 100 epochs. GeONet was trained using 20\%, 40\%, 60\%, 80\%, and
100\% of the available training data, while the validation and test sets were
held fixed across experiments.
The measured wall-clock training times are summarized in
Table~\ref{tab:training_time} and visualized in
Figure~\ref{fig:comp_cost}. Training time increases approximately
proportionally with the amount of training data for both platforms and
hardware configurations. On Tahoma, increasing the training fraction from
20\% to 100\% increases GPU training time from 39~min~19~s to
3~h~33~min~35~s and CPU training time from 6~h~19~min~46~s to
38~h~23~min~30~s. On Deception, the corresponding GPU training time
increases from 42~min~42~s to 3~h~16~min~2~s, whereas CPU training
increases from 12~h~25~min~14~s to 61~h~7~min~57~s
(Table~\ref{tab:training_time}).

\begin{table}[htbp]
    \centering
    \caption{Wall-clock training time for GeONet over 100 epochs as a
    function of training-data fraction on the Tahoma and Deception HPC
    platforms. GPU-accelerated and CPU-only configurations use the same
    training protocol and model configuration. All values are measured
    wall-clock times from completed training runs.}
    \label{tab:training_time}
    \renewcommand{\arraystretch}{1.25}
    \begin{tabular}{ccccc}
        \toprule
        \textbf{Training data (\%)} &
        \textbf{Tahoma GPU} &
        \textbf{Tahoma CPU} &
        \textbf{Deception GPU} &
        \textbf{Deception CPU} \\
        \midrule
        20  & 39m 19s    & 6h 19m 46s  & 42m 42s    & 12h 25m 14s \\
        40  & 1h 14m 9s  & 14h 10m 42s & 1h 18m 37s & 24h 15m 51s \\
        60  & 2h 4m 10s  & 21h 11m 3s  & 2h 8m 43s  & 34h 40m 8s \\
        80  & 2h 46m 54s & 27h 17m 0s  & 2h 47m 25s & 46h 36m 14s \\
        100 & 3h 33m 35s & 38h 23m 30s & 3h 16m 2s  & 61h 7m 57s \\
        \bottomrule
    \end{tabular}
\end{table}

GPU acceleration substantially reduces the wall-clock cost throughout the
investigated data range. Based on the measured times in
Table~\ref{tab:training_time}, the CPU-to-GPU speedup is approximately
9.7--11.5$\times$ on Tahoma and 16.2--18.7$\times$ on Deception
(Figure~\ref{fig:comp_cost}). At the full data scale, GPU acceleration
reduces the 100-epoch training time from approximately 38.4~h to 3.6~h on
Tahoma and from approximately 61.1~h to 3.3~h on Deception. The larger
relative speedup observed on Deception results primarily from its higher
CPU-only wall-clock cost rather than from uniformly faster GPU execution.
Indeed, GPU training times on the two systems are comparable: Tahoma is
slightly faster for the 20--80\% data fractions, whereas Deception is faster
for the full-data experiment. These results demonstrate that the magnitude
of GPU acceleration depends on the underlying HPC configuration, while the
approximately proportional scaling of training time with dataset size is
observed on both systems.

\begin{figure}[width=\linewidth,pos=htbp]
    \centering
    \includegraphics[width=0.95\textwidth]
    {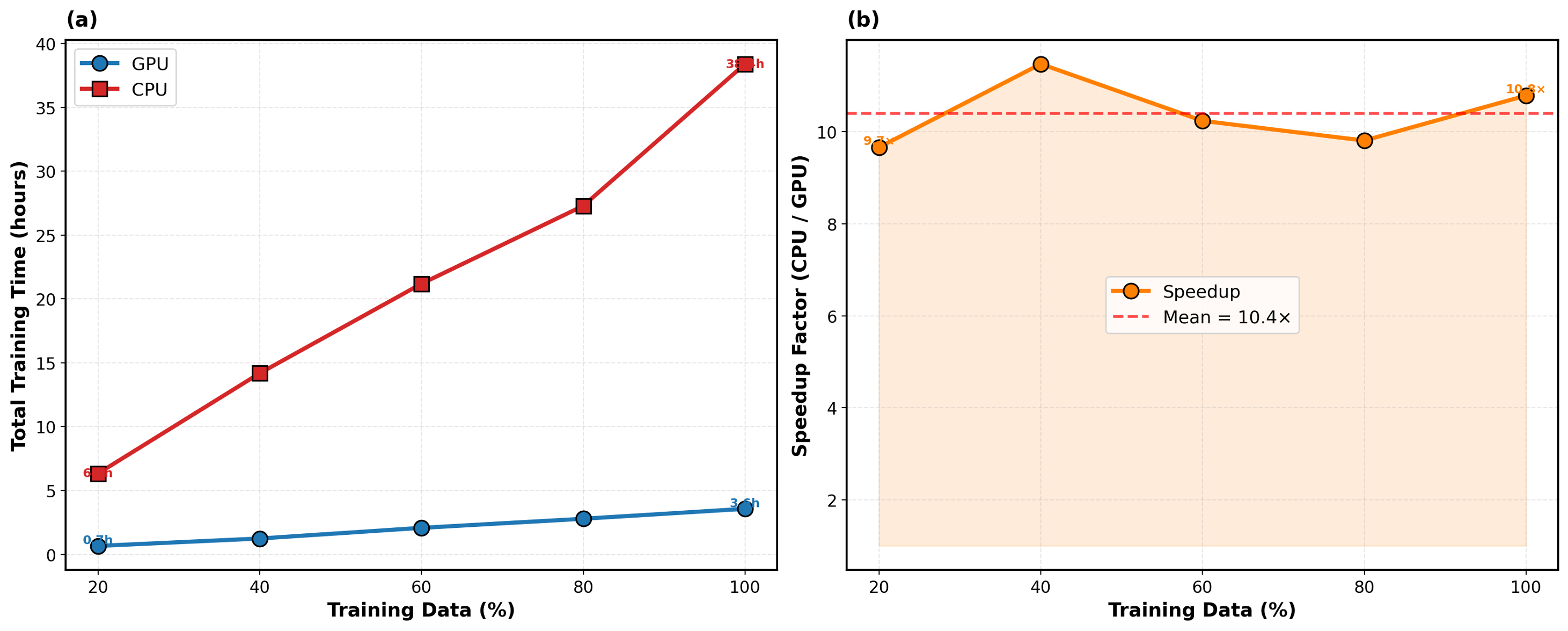}
    \\[6pt]
    \includegraphics[width=0.95\textwidth]
    {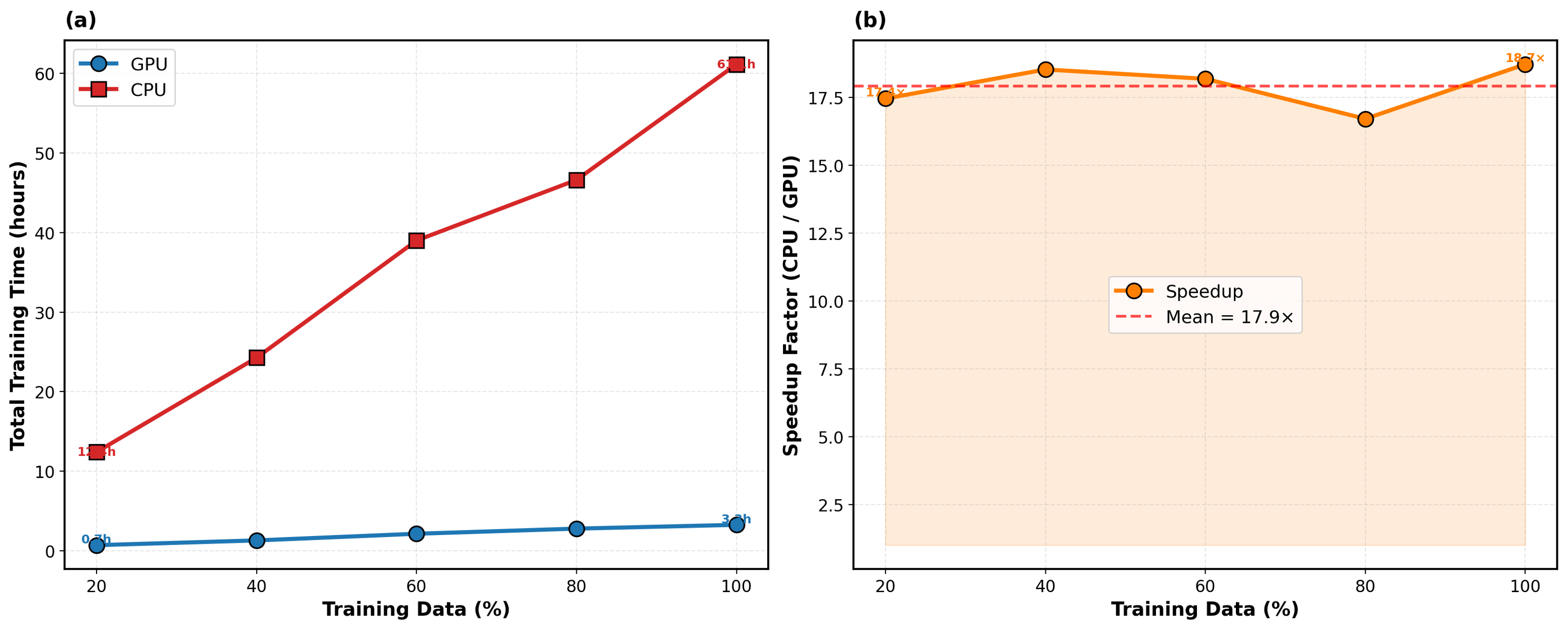}
    \caption{Computational scaling of GeONet training on the Tahoma (top)
    and Deception (bottom) HPC platforms. (a) Total wall-clock time for
    100 training epochs as a function of training-data fraction for GPU
    and CPU execution. (b) CPU-to-GPU wall-clock speedup, with the dashed
    line indicating the mean speedup across the evaluated data fractions.
    Training time increases approximately proportionally with dataset size,
    while GPU acceleration substantially reduces the computational cost on
    both platforms.}
    \label{fig:comp_cost}
\end{figure}

Figure~\ref{fig:val_loss_conv} compares validation-loss convergence across
training-data fractions, platforms, and computational backends. For all
configurations, validation loss decreases sharply during approximately the
first 5--10 epochs and subsequently enters a slower refinement regime. The
GPU and CPU trajectories exhibit similar qualitative behavior on both
Tahoma and Deception, indicating that changing the computational backend
primarily affects wall-clock time rather than the overall optimization
trajectory. Models trained with larger data fractions generally attain lower
validation losses, with the 100\% models reaching the lowest values on both HPC
platforms. The 60\% and 80\% trajectories are considerably closer,
indicating diminishing separation among intermediate-to-large training
fractions. The similarity of the convergence behavior across Tahoma and
Deception further demonstrates that the observed optimization characteristics
are reproducible across the two investigated HPC environments.

\begin{figure}[width=\linewidth,pos=htbp]
    \centering
    \includegraphics[width=0.95\textwidth]
    {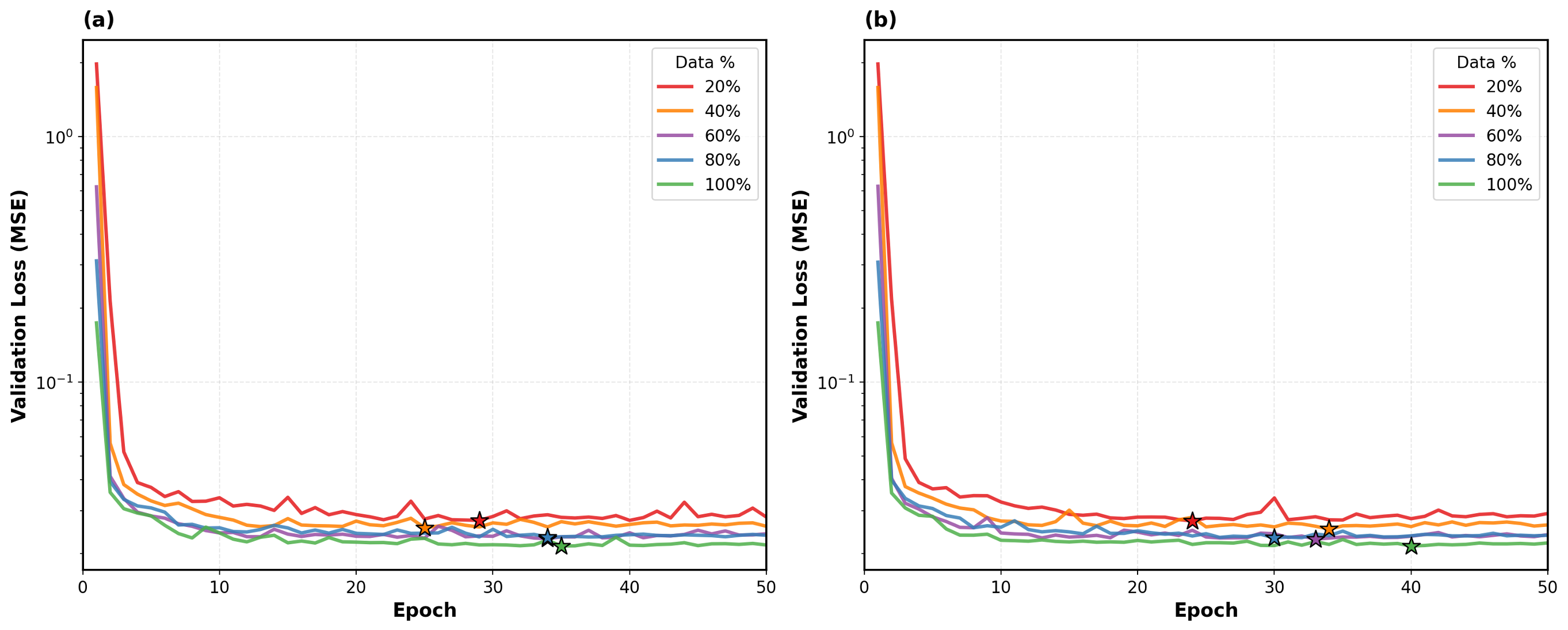}
    \\[6pt]
    \includegraphics[width=0.95\textwidth]
    {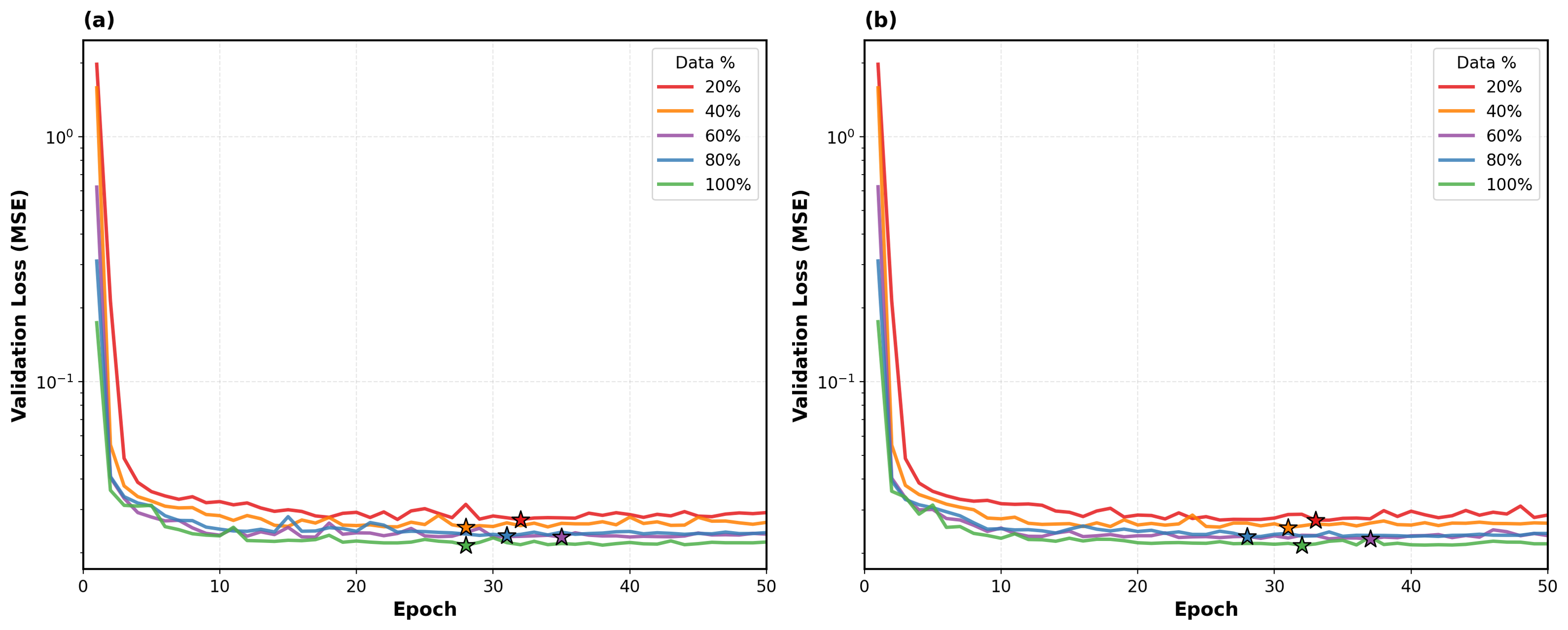}
    \caption{Validation-loss convergence for GeONet across training-data
    fractions. The top and bottom rows correspond to Tahoma and Deception,
    respectively, while the left and right columns correspond to GPU and
    CPU training. Stars denote the epochs attaining the minimum validation
    loss. Across configurations, validation loss decreases rapidly during
    the initial 5--10 epochs and subsequently enters a slower refinement
    regime.}
    \label{fig:val_loss_conv}
\end{figure}

\subsubsection{Data-Scaling Effects on Model Accuracy}
\label{subsubsec:data_scaling}

The relationship between training-data volume and predictive accuracy is
summarized in Tables~\ref{tab:tahoma_metrics} and
\ref{tab:deception_metrics} and visualized in
Figures~\ref{fig:data_scaling} and~\ref{fig:pred_accuracy}. The tables report
NRMSE, MAPE, and $R^2$ as mean $\pm$ standard deviation over 20 held-out
test samples. Figure~\ref{fig:data_scaling} additionally compares the best
validation loss and test-set $R^2$ as the training-data fraction increases.
GeONet exhibits substantial data efficiency. Even when trained using only
20\% of the available data, the model achieves mean $R^2$ values between
0.9761 and 0.9781 across the four platform--device combinations, while mean
MAPE remains between 2.11\% and 2.23\%
(Tables~\ref{tab:tahoma_metrics} and~\ref{tab:deception_metrics}). Thus,
models trained with only one fifth of the available corpus already capture
most of the spatial variability in the held-out pressure fields while
maintaining a mean percentage error close to 2\%.

\begin{figure}[width=\linewidth,pos=htbp]
    \centering
    \includegraphics[width=0.95\textwidth]
    {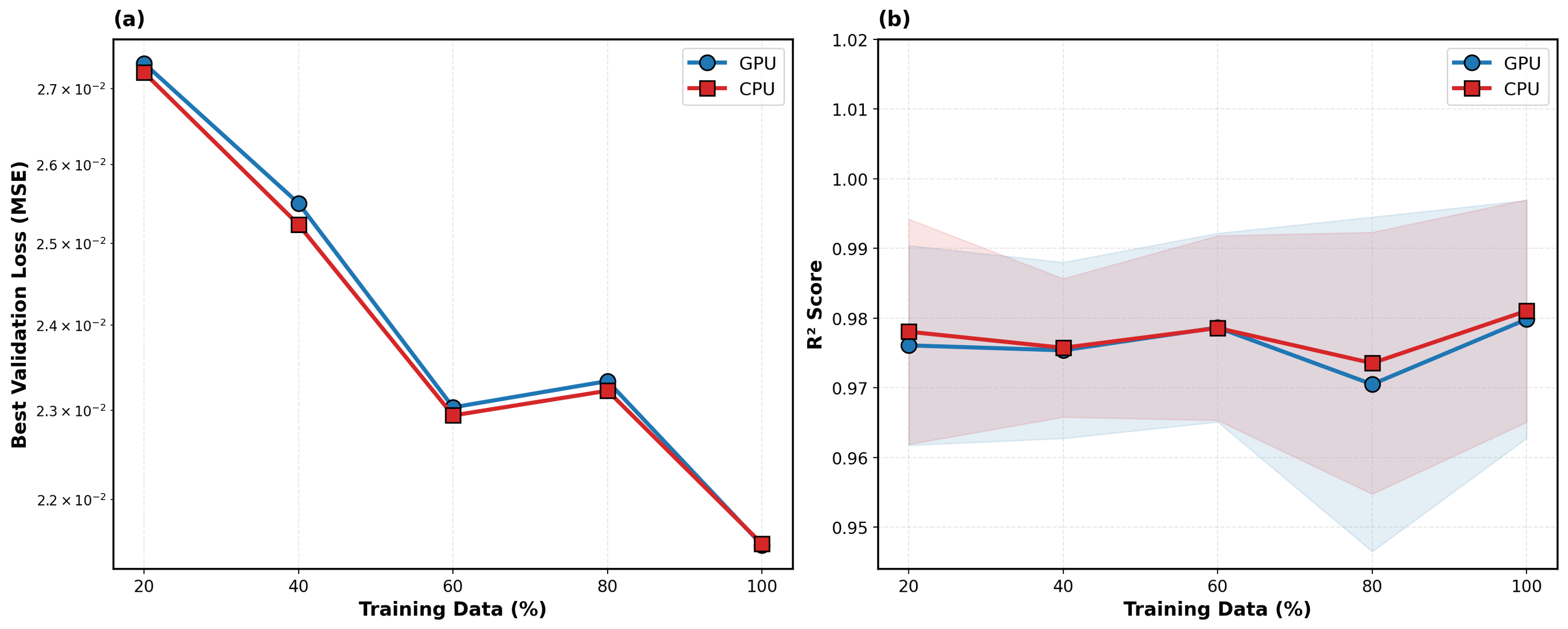}
    \\[6pt]
    \includegraphics[width=0.95\textwidth]
    {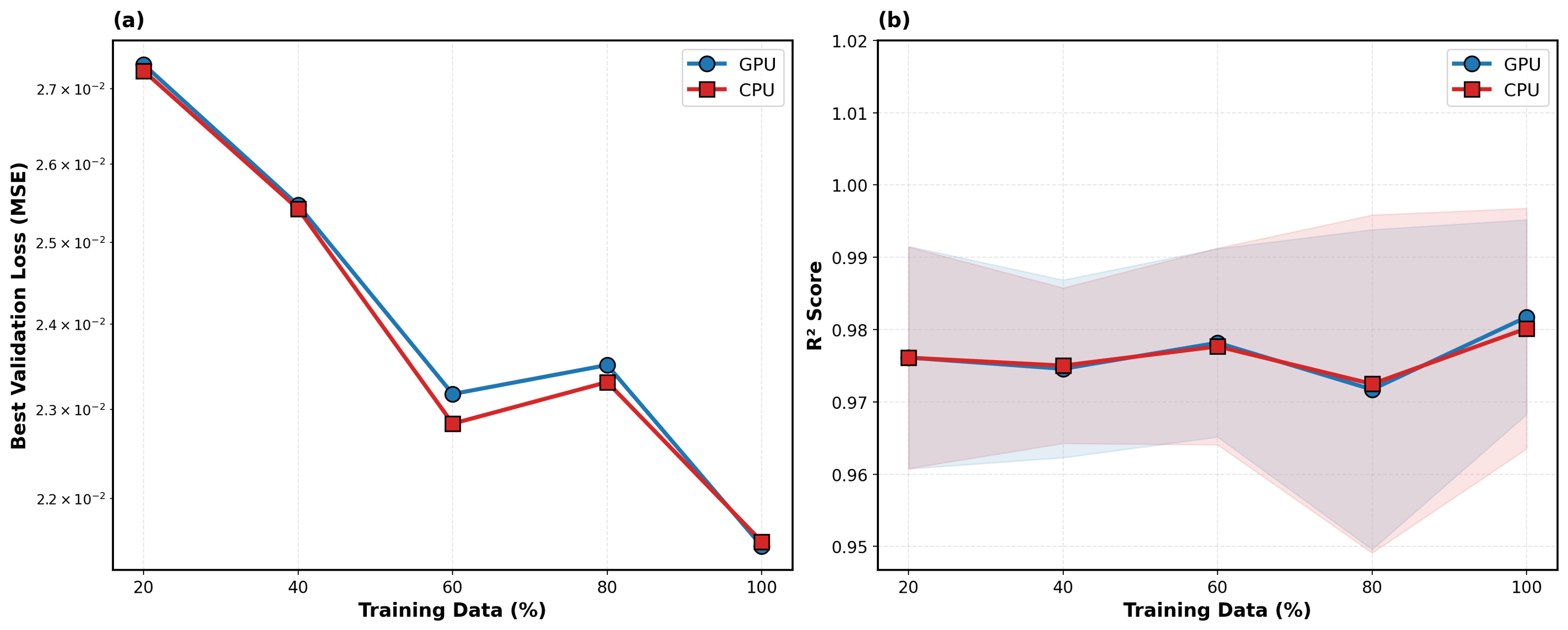}
    \caption{Data-scaling behavior of GeONet on the Tahoma (top) and
    Deception (bottom) HPC platforms. (a) Best validation loss as a function
    of training-data fraction. (b) Test-set $R^2$ as a function of
    training-data fraction for GPU- and CPU-trained models; shaded regions
    represent variability across the held-out test realizations. Validation
    loss generally decreases as the training-data fraction increases,
    whereas test-set $R^2$ remains above 0.97 across all configurations and
    reaches its highest mean values with the complete training set.}
    \label{fig:data_scaling}
\end{figure}

\begin{table}[htbp]
\centering
\caption{Test-set prediction accuracy for GeONet trained on Tahoma as a
function of training-data fraction. Values represent mean $\pm$ standard
deviation over 20 held-out test samples. GPU- and CPU-trained models exhibit
closely comparable predictive accuracy across the evaluated data fractions.}
\label{tab:tahoma_metrics}
\renewcommand{\arraystretch}{1.3}
\begin{tabular}{lccccc}
\toprule
\textbf{Metric} &
\textbf{20\%} &
\textbf{40\%} &
\textbf{60\%} &
\textbf{80\%} &
\textbf{100\%} \\
\midrule
NRMSE (GPU) &
$0.0506 \pm 0.0154$ &
$0.0502 \pm 0.0126$ &
$0.0464 \pm 0.0127$ &
$0.0540 \pm 0.0194$ &
$0.0453 \pm 0.0172$ \\

NRMSE (CPU) &
$0.0480 \pm 0.0166$ &
$0.0504 \pm 0.0100$ &
$0.0466 \pm 0.0124$ &
$0.0518 \pm 0.0178$ &
$0.0436 \pm 0.0177$ \\
\midrule
MAPE (\%) (GPU) &
$2.23 \pm 0.52$ &
$2.12 \pm 0.51$ &
$1.93 \pm 0.48$ &
$2.24 \pm 0.86$ &
$1.86 \pm 0.70$ \\

MAPE (\%) (CPU) &
$2.11 \pm 0.50$ &
$2.10 \pm 0.39$ &
$1.94 \pm 0.47$ &
$2.12 \pm 0.74$ &
$1.77 \pm 0.69$ \\
\midrule
$R^2$ (GPU) &
$0.9761 \pm 0.0143$ &
$0.9754 \pm 0.0126$ &
$0.9786 \pm 0.0135$ &
$0.9705 \pm 0.0240$ &
$0.9798 \pm 0.0171$ \\

$R^2$ (CPU) &
$0.9781 \pm 0.0162$ &
$0.9757 \pm 0.0099$ &
$0.9786 \pm 0.0132$ &
$0.9735 \pm 0.0188$ &
$0.9810 \pm 0.0160$ \\
\bottomrule
\end{tabular}
\end{table}

\begin{table}[htbp]
\centering
\caption{Test-set prediction accuracy for GeONet trained on Deception as a
function of training-data fraction. Values represent mean $\pm$ standard
deviation over 20 held-out test samples. GPU- and CPU-trained models exhibit
closely comparable predictive accuracy across the evaluated data fractions.}
\label{tab:deception_metrics}
\renewcommand{\arraystretch}{1.3}
\begin{tabular}{lccccc}
\toprule
\textbf{Metric} &
\textbf{20\%} &
\textbf{40\%} &
\textbf{60\%} &
\textbf{80\%} &
\textbf{100\%} \\
\midrule
NRMSE (GPU) &
$0.0503 \pm 0.0162$ &
$0.0512 \pm 0.0117$ &
$0.0471 \pm 0.0121$ &
$0.0530 \pm 0.0193$ &
$0.0433 \pm 0.0158$ \\

NRMSE (CPU) &
$0.0503 \pm 0.0158$ &
$0.0511 \pm 0.0108$ &
$0.0474 \pm 0.0129$ &
$0.0522 \pm 0.0189$ &
$0.0448 \pm 0.0170$ \\
\midrule
MAPE (\%) (GPU) &
$2.21 \pm 0.56$ &
$2.16 \pm 0.46$ &
$1.96 \pm 0.46$ &
$2.18 \pm 0.84$ &
$1.78 \pm 0.64$ \\

MAPE (\%) (CPU) &
$2.17 \pm 0.52$ &
$2.13 \pm 0.42$ &
$1.96 \pm 0.51$ &
$2.17 \pm 0.82$ &
$1.86 \pm 0.69$ \\
\midrule
$R^2$ (GPU) &
$0.9761 \pm 0.0154$ &
$0.9746 \pm 0.0123$ &
$0.9782 \pm 0.0130$ &
$0.9717 \pm 0.0221$ &
$0.9817 \pm 0.0135$ \\

$R^2$ (CPU) &
$0.9761 \pm 0.0153$ &
$0.9750 \pm 0.0107$ &
$0.9777 \pm 0.0136$ &
$0.9725 \pm 0.0234$ &
$0.9802 \pm 0.0166$ \\
\bottomrule
\end{tabular}
\end{table}

Increasing the training-data fraction produces an overall improvement in
predictive performance, although the relationship is not strictly monotonic.
The 60\% models generally improve upon the 20\% and 40\% configurations,
whereas the 80\% models exhibit a temporary degradation in NRMSE, MAPE, and
$R^2$ before performance improves again at 100\%
(Tables~\ref{tab:tahoma_metrics} and~\ref{tab:deception_metrics};
Figures~\ref{fig:data_scaling} and~\ref{fig:pred_accuracy}). For example,
on Tahoma GPU, mean NRMSE decreases from 0.0506 at 20\% data to 0.0464 at
60\%, increases to 0.0540 at 80\%, and then decreases to 0.0453 with the
complete training set. A similar response is observed for the Tahoma CPU and
Deception configurations. The consistency of this intermediate fluctuation
across the experiments indicates that predictive performance should be
assessed using held-out validation and test statistics rather than assuming
that every increment in training-data volume necessarily produces a lower
test error.
The complete training corpus nevertheless provides the best overall
performance. On Tahoma, the 100\% models attain NRMSE values of
$0.0453 \pm 0.0172$ and $0.0436 \pm 0.0177$ for GPU and CPU training,
respectively, with corresponding $R^2$ values of
$0.9798 \pm 0.0171$ and $0.9810 \pm 0.0160$
(Table~\ref{tab:tahoma_metrics}). On Deception, the GPU- and CPU-trained
models attain NRMSE values of $0.0433 \pm 0.0158$ and
$0.0448 \pm 0.0170$, respectively, while $R^2$ reaches
$0.9817 \pm 0.0135$ and $0.9802 \pm 0.0166$
(Table~\ref{tab:deception_metrics}). Mean MAPE for the full-data models
ranges from 1.77\% to 1.86\% across the four configurations.

The GPU- and CPU-trained models exhibit close agreement at a fixed
training-data fraction. At 100\% data, for example, the difference between
GPU and CPU mean $R^2$ is approximately 0.0012 on Tahoma and 0.0015 on
Deception. Similar agreement is observed for NRMSE and MAPE
(Tables~\ref{tab:tahoma_metrics} and~\ref{tab:deception_metrics}). These
differences are small relative to the corresponding test-set standard
deviations. Figures~\ref{fig:data_scaling} and~\ref{fig:pred_accuracy}
further show that the GPU and CPU accuracy trends closely track one another
as the training fraction changes. Consequently, the principal advantage of
GPU execution in these experiments is the substantial reduction in
wall-clock training time rather than a systematic improvement in the
predictive accuracy of the converged model.
The similarity between the Tahoma and Deception results provides an
additional measure of computational reproducibility. Both platforms exhibit
comparable validation-loss scaling (Figure~\ref{fig:data_scaling}), similar
GPU--CPU agreement in test accuracy (Tables~\ref{tab:tahoma_metrics}
and~\ref{tab:deception_metrics}), and the same qualitative response to
increasing training-data fraction. Within the hardware and training
configurations investigated here, the data-efficiency behavior is therefore
reproducible across the two HPC environments.

\begin{figure}[width=\linewidth,pos=htbp]
    \centering
    \includegraphics[width=0.95\textwidth]
    {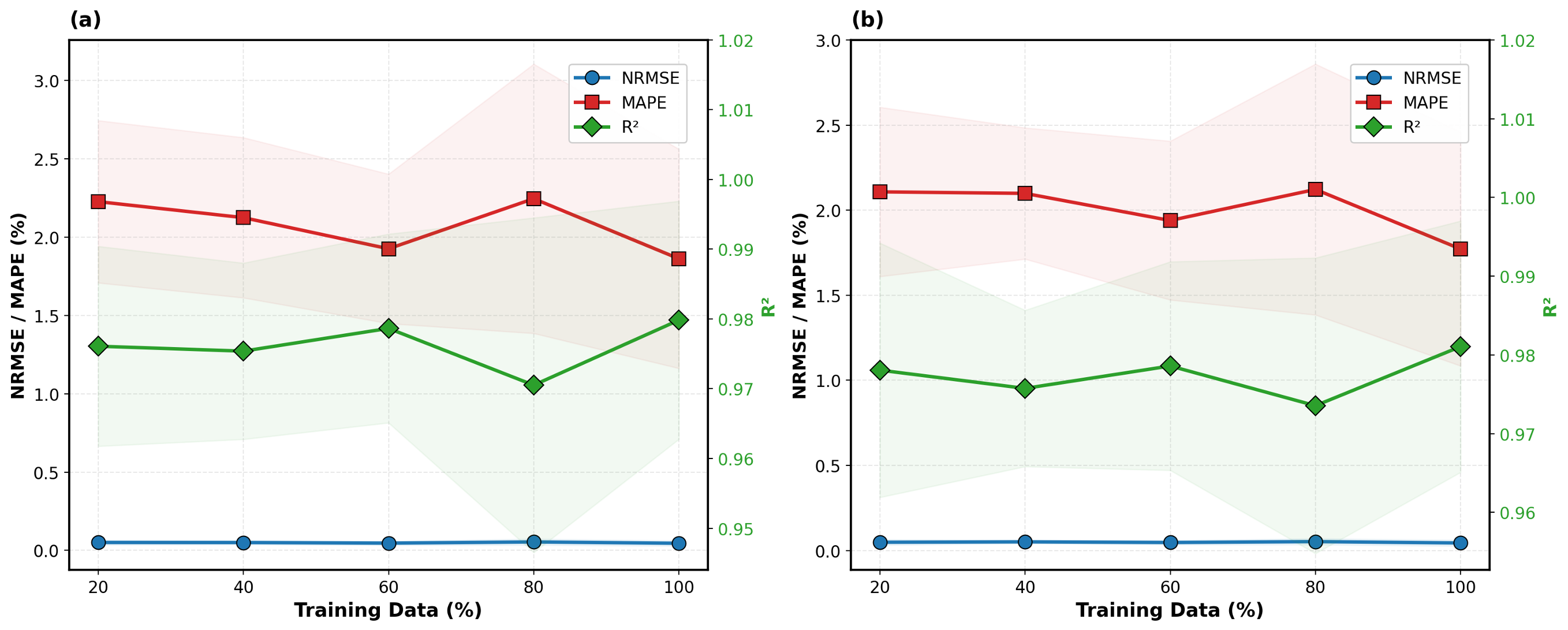}
    \\[6pt]
    \includegraphics[width=0.95\textwidth]
    {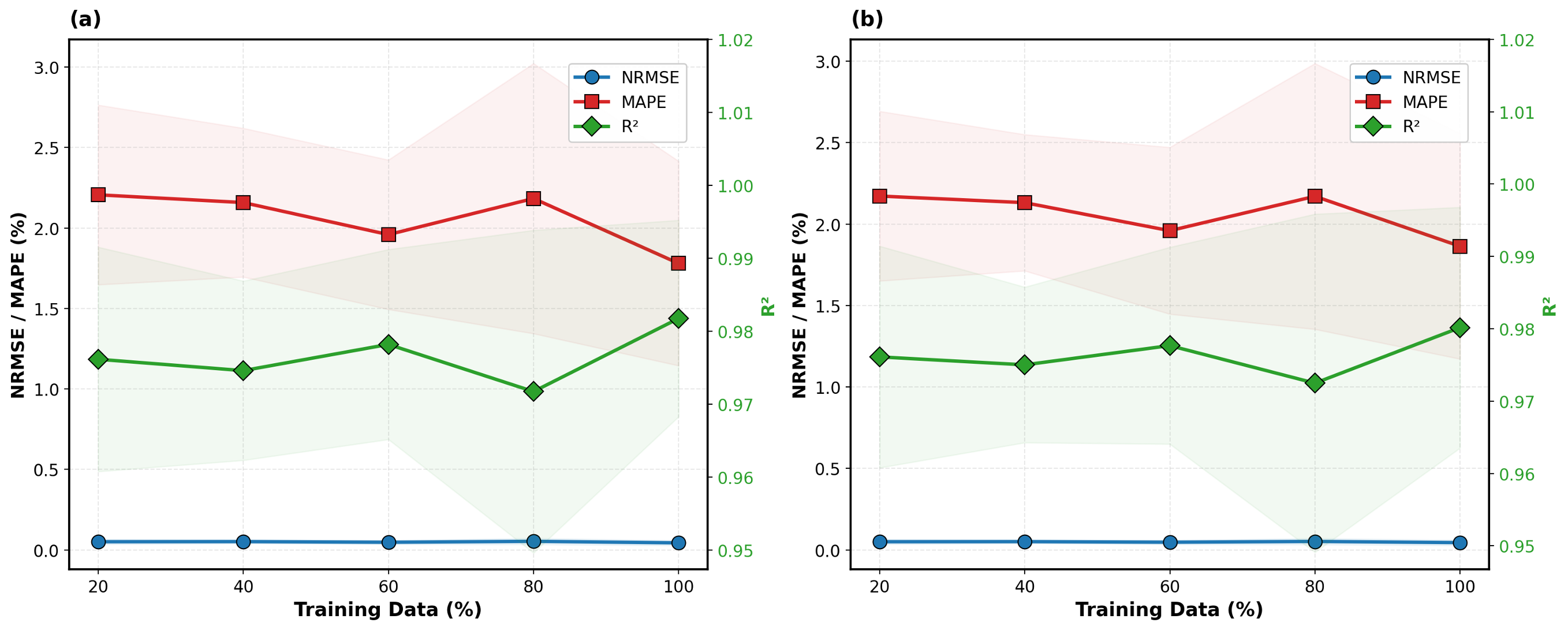}
    \caption{Test-set prediction accuracy versus training-data fraction for GeoNet models trained on Tahoma (top) and Deception (bottom). Columns show GPU- and CPU-trained models, respectively. Prediction quality is quantified using normalized root-mean-square error (NRMSE), mean absolute percentage error (MAPE), and coefficient of determination ($R^2$). Shaded bands indicate variability across held-out test realizations. The GPU- and CPU-trained models exhibit comparable data-scaling trends, indicating that the training device has negligible influence on predictive accuracy. Increasing the training-data fraction yields modest improvements, with the full training set achieving the best overall performance.}
    \label{fig:pred_accuracy}
\end{figure}

Figure~\ref{fig:error_dist} complements the aggregate accuracy metrics by
showing the distribution of per-sample mean relative prediction error over
the held-out test realizations. The error distributions remain concentrated
at relatively low values for all training fractions, including the 20\%
models. The full-data models generally exhibit lower central errors,
particularly for the GPU runs, although the widths of the distributions do
not decrease monotonically with increasing data fraction. The 80\%
configurations exhibit broader variability and several larger-error
realizations, consistent with the temporary degradation observed in
Tables~\ref{tab:tahoma_metrics} and~\ref{tab:deception_metrics} and
Figure~\ref{fig:pred_accuracy}. The subsequent improvement at 100\% data
shows that the complete training corpus provides lower mean errors and strong
overall predictive performance, while the intermediate fluctuations
emphasize the importance of evaluating data efficiency across multiple
held-out realizations rather than relying on a single aggregate statistic.

\begin{figure}[width=\linewidth,pos=htbp]
    \centering
    \begin{minipage}{0.48\textwidth}
        \centering
        \includegraphics[width=\textwidth]
        {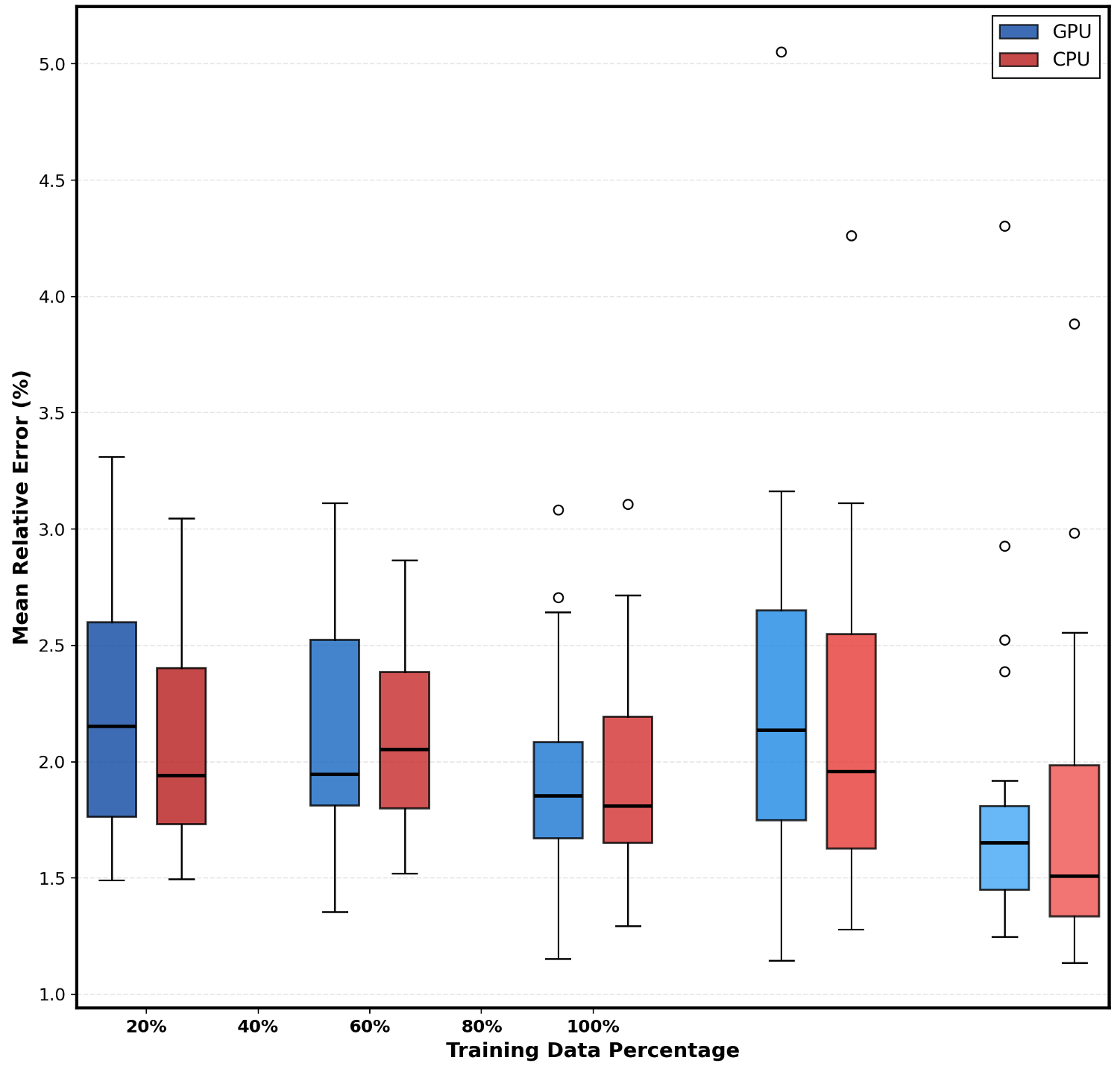}
    \end{minipage}
    \hfill
    \begin{minipage}{0.48\textwidth}
        \centering
        \includegraphics[width=\textwidth]
        {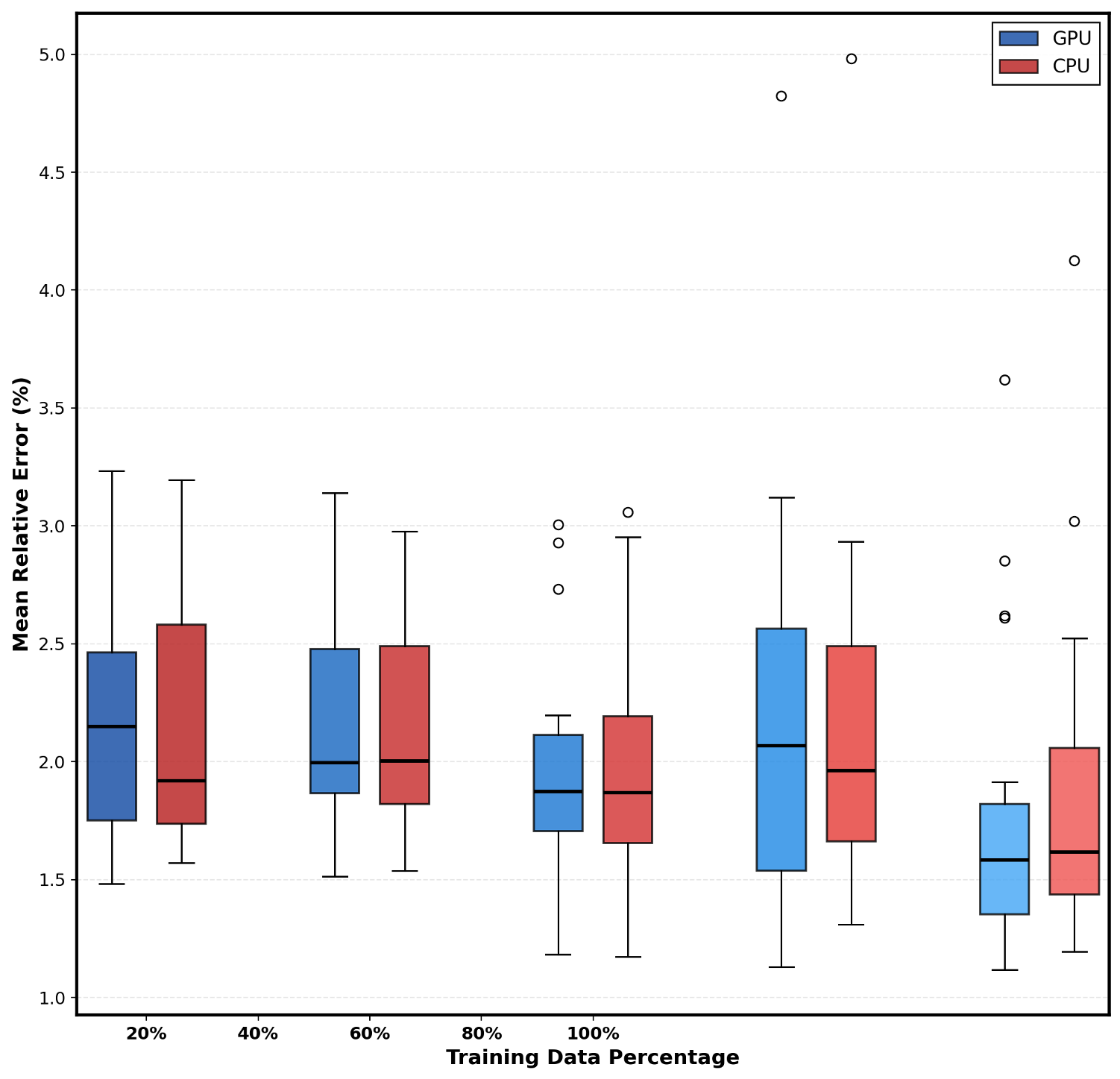}
    \end{minipage}
    \caption{Distribution of per-sample mean relative prediction error
    across held-out test realizations for Tahoma (left) and Deception
    (right). Box plots show the distributions obtained with GPU- and
    CPU-trained models at each training-data fraction. The full-data models
    generally exhibit lower central errors, whereas the intermediate
    fractions show non-monotonic variation in error spread and several
    higher-error test realizations.}
    \label{fig:error_dist}
\end{figure}

Taken together, the computational and predictive results reveal a favorable
accuracy--cost trade-off. Training with 20\% of the available data requires
only 39~min~19~s on Tahoma GPU and 42~min~42~s on Deception GPU
(Table~\ref{tab:training_time}), while retaining mean $R^2$ values of
approximately 0.976--0.978 and MAPE values of approximately 2.1--2.2\%
(Tables~\ref{tab:tahoma_metrics} and~\ref{tab:deception_metrics}). This
reduced-data regime therefore provides an efficient setting for
computationally intensive model-development activities such as architecture
screening and hyperparameter exploration.
Increasing the training fraction to 60\% requires approximately 2.1~h of
GPU training on either platform and yields mean $R^2$ values near 0.978 and
MAPE values of approximately 1.9--2.0\%. At the full data scale, training
requires approximately 3.6~h on Tahoma GPU and 3.3~h on Deception GPU and
provides the best overall test performance, with mean $R^2$ near 0.98 and
MAPE near 1.8\%
(Tables~\ref{tab:training_time}--\ref{tab:deception_metrics}). These results
show that the training-data fraction can be selected according to the
available computational budget and required predictive fidelity: reduced
subsets provide an efficient environment for model development, whereas the
complete dataset is preferable when minimizing predictive error is the
primary objective.

\subsubsection{Inference Cost and Prediction Efficiency}
\label{subsubsec:inference_cost}

The computational advantage of GeONet becomes particularly relevant after
the one-time costs of high-fidelity data generation and surrogate training
have been incurred. As discussed in
Section~\ref{subsubsec:data_generation}, a representative PFLOTRAN
realization requires approximately two days on two 36-core compute nodes.
By comparison, the complete 100-epoch GeONet training procedure requires
approximately 3.3--3.6~h on the GPU configurations considered here
(Table~\ref{tab:training_time} and Figure~\ref{fig:comp_cost}). Once this
training stage is complete, predictions for additional input realizations
are obtained through evaluation of the learned neural operator rather than
through repeated numerical solution of the governing flow equations.
This distinction is particularly relevant for applications that require
repeated forward evaluations, including uncertainty quantification,
sensitivity analysis, inverse modeling, and design-space exploration. In
such applications, the substantial cost of constructing the PFLOTRAN
training corpus can be amortized over repeated use of the trained surrogate.
Moreover, because GeONet represents the solution through spatial query
coordinates, the trained model provides a coordinate-based mechanism for
evaluating the predicted pressure field for new input realizations.

The present computational experiments quantify high-fidelity data-generation
and model-training costs but do not include a controlled inference-latency
benchmark. Consequently, a numerical inference speedup relative to PFLOTRAN
is not reported. Establishing such a value would require direct measurements
of GeONet inference latency and throughput under specified hardware,
batch-size, and spatial-resolution conditions. The results reported here
nevertheless demonstrate that the dominant computational investment is
concentrated in generation of the high-fidelity reference data, after which
GeONet provides a reusable learned representation for subsequent
pressure-field evaluation.

\subsubsection{Practical Computational Implications}
\label{subsubsec:practical_considerations}

The combined results in Table~\ref{tab:training_time},
Tables~\ref{tab:tahoma_metrics}--\ref{tab:deception_metrics}, and
Figures~\ref{fig:data_scaling}--\ref{fig:error_dist} provide practical
guidance for selecting computational resources and training-data volume.
GPU acceleration is particularly advantageous for full-scale model
development, reducing wall-clock training time by approximately one order of
magnitude on Tahoma and by more than an order of magnitude on Deception.
CPU-only training remains feasible for reduced-data experiments, but the
difference becomes substantial at full scale, where CPU training requires
approximately 38.4~h on Tahoma and 61.1~h on Deception compared with
approximately 3.6~h and 3.3~h, respectively, on GPU
(Table~\ref{tab:training_time} and Figure~\ref{fig:comp_cost}).

For exploratory model development, the 20\% data configuration provides a
particularly efficient operating regime because it completes in less than
45 minutes on either GPU platform while retaining mean $R^2>0.976$ and MAPE
below 2.3\% (Tables~\ref{tab:tahoma_metrics} and
\ref{tab:deception_metrics}). The 60\% configuration provides an
intermediate regime, reducing MAPE to approximately 2\% while requiring
only about 2.1~h of GPU training. The complete training set provides the
highest overall predictive accuracy at a GPU cost of approximately
3--4~h. This separation between inexpensive reduced-data experiments and
full-data training is useful for workflows in which multiple architectures
or optimization configurations must be evaluated before selecting a final
model.

Finally, the approximately proportional increase in wall-clock time with
training-data fraction in Figure~\ref{fig:comp_cost} provides a practical
basis for estimating the computational resources required for larger
training campaigns. The convergence curves in
Figure~\ref{fig:val_loss_conv} further show that most of the reduction in
validation loss occurs during the early stage of optimization, followed by
slower refinement. Validation-based early stopping may therefore provide an
additional mechanism for reducing training cost, particularly during
exploratory experiments. For final models, however, the stopping criterion
should be determined from validation performance rather than from a fixed
epoch threshold because the epoch attaining the minimum validation loss
varies with training-data fraction, platform, and computational backend.

\subsection{Limitations}
\label{Sec:limitations}
The present study focuses on steady-state pressure prediction for single-phase flow in RCP geometries generated using PFLOTRAN. 
While the held-out tests demonstrate strong generalization across previously unseen RCP configurations, the current evaluation remains within the broader class of geometries and physical conditions represented in the training data.
Future work will extend GeONet to experimentally derived porous-media geometries.
This includes X-ray computed tomography and other chemical imaging modalities (e.g., Scanning Electron Microscopy coupled with Energy-Dispersive X-ray Spectroscopy, Matrix-Assisted Laser Desorption/Ionization), as well as domains with greater morphological and material complexity (e.g., critical minerals embedded in the chip). 
The existing dataset also contains tracer-transport information, providing a natural pathway to extend the framework toward transient transport, coupled flow--transport processes, and ultimately more complex multiphysics applications.

The current implementation is primarily supervised and deterministic, with the reported model optimized for pressure-field reconstruction. 
Future developments will investigate physics- and conservation-informed objectives, uncertainty quantification, and out-of-distribution diagnostics to further improve model reliability when applied to new experimental conditions. 
Future studies will further evaluate the contribution of individual GeONet components, including the Leaf Network, local geometric descriptors, Fourier features, and DeepMixer fusion architecture.
Finally, systematic benchmarking of inference time, memory requirements, and scalability across multiscale problems will help quantify GeONet's computational advantages and establish its practical applicability to experimental design, uncertainty analysis, inverse modeling, and real-time digital-twin applications.

\subsection{Discussion related to experimental data toward multi-physics model calibration}
\label{Sec:results_experimental_data}
The simulation-based analyses presented above demonstrate that GeONet can learn geometry-conditioned pressure fields across heterogeneous RCP domains. 
An important step toward experimentally constrained prediction is to connect this framework with pore structures and flow responses measured directly from physical samples. 
To explore this connection, future work will use complementary X-ray computed tomography (XCT), pore-network analysis, and Pore2Chip experiments \cite{truong2025pore2chip} on soil aggregates before and after freeze--thaw disturbance. 
The experimental workflow translates three-dimensional changes in pore architecture into reduced-complexity micromodels that retain key structural characteristics of the original pore network. 
Rather than selecting an arbitrary two-dimensional slice from the XCT volume, which may not adequately represent disturbance-induced structural changes, Pore2Chip extracts pore-network characteristics such as pore-size distributions, pore-throat-size distributions, and pore coordination numbers, and uses these descriptors to construct experimentally realizable micromodels. In this way, the platform provides a controlled bridge between the complex three-dimensional pore structure observed by XCT and laboratory-scale flow and transport experiments.

%----------------------------------------------------------
% Undisturbed pore network
%----------------------------------------------------------
\begin{figure}[width=\linewidth,pos=htbp]
    \centering
    \includegraphics[
        page=1,
        width=\textwidth
    ]{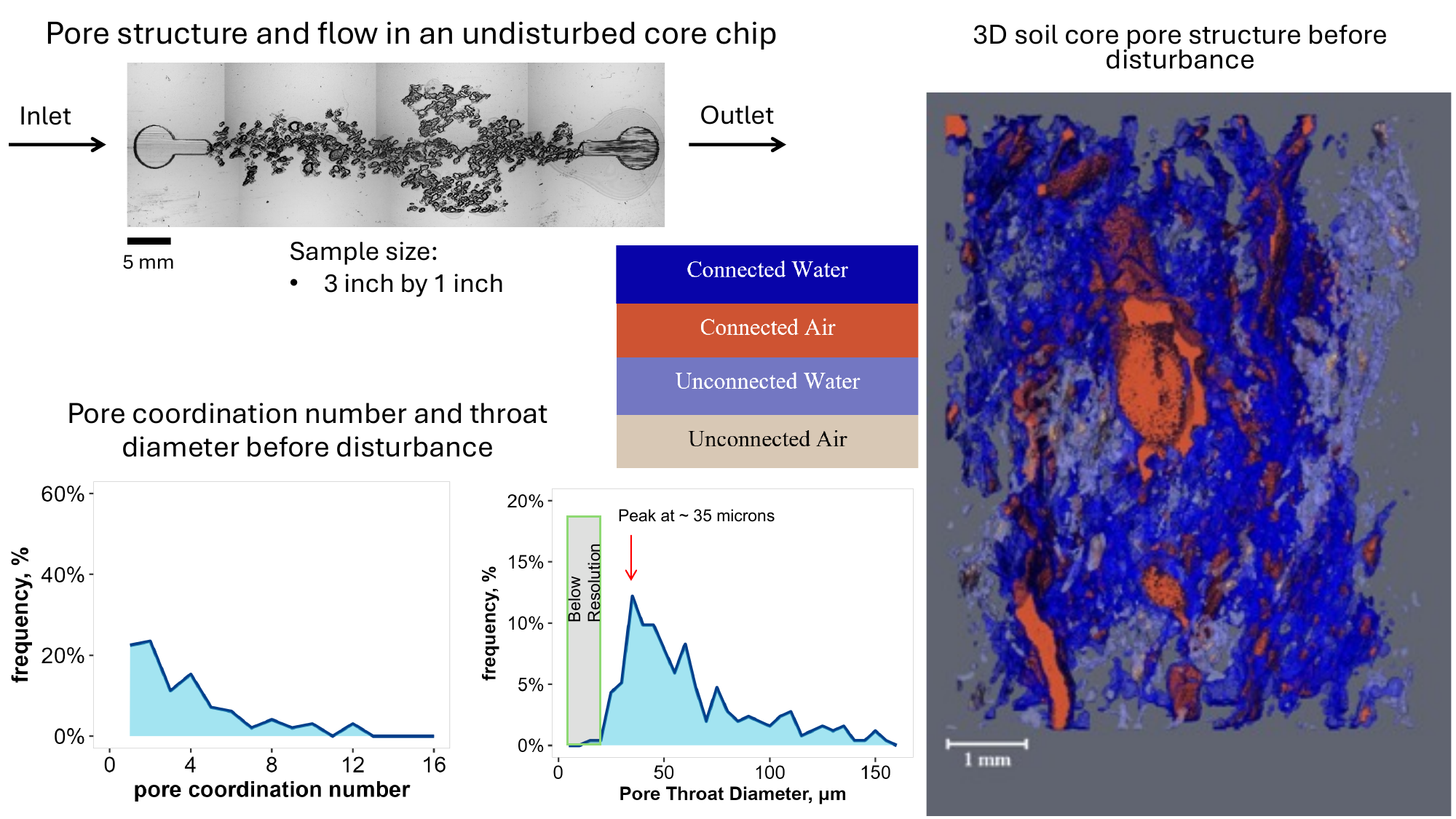}
    \caption{
    Experimental characterization of the undisturbed soil pore network and its corresponding Pore2Chip micromodel. 
    The micromodel (upper left) represents the pore structure used for controlled flow experiments, with flow imposed from the inlet to the outlet. 
    We will use these experimental data in subsequent modeling work in PFLOTRAN. 
    The lower panels show the pore coordination number distribution and pore throat diameter distribution derived from pore-network analysis. 
    The resolved pore throat distribution exhibits a dominant peak near approximately $35~\mu\mathrm{m}$.
    The three-dimensional XCT-based visualization of a soil core (right) illustrates the spatial organization and connectivity of the pore structure before disturbance.
    }
    \label{fig:exp_undisturbed}
\end{figure}

The before--after comparison reveals clear changes in pore-network organization following disturbance (see Figure~\ref{fig:exp_undisturbed}). 
In the undisturbed sample, the pore network exhibits a broader distribution of coordination numbers, reflecting a larger fraction of multiply connected pores and a greater diversity of potential flow pathways. 
Following freeze--thaw disturbance, the distribution shifts toward lower coordination numbers, with a larger fraction of singly connected pores. 
This structural change matters because coordination number directly reflects network connectivity.
That is, reduced pore coordination decreases the number of alternative pathways for fluid redistribution and can increase bulk-flow sensitivity to individual pore throats or local bottlenecks. 
The XCT reconstructions qualitatively support this interpretation by showing substantial reorganization of the pore structure following disturbance. 
These observations are consistent with the underlying Pore2Chip motivation that disturbance modifies not only pore volume but also the connectivity and topology of the pore network, which we will use to calibrate the PFLOTRAN multiphysics model in future work.

%----------------------------------------------------------
% Disturbed pore network
%----------------------------------------------------------
\begin{figure}[width=\linewidth,pos=htbp]
    \centering
    \includegraphics[
        page=2,
        width=\textwidth
    ]{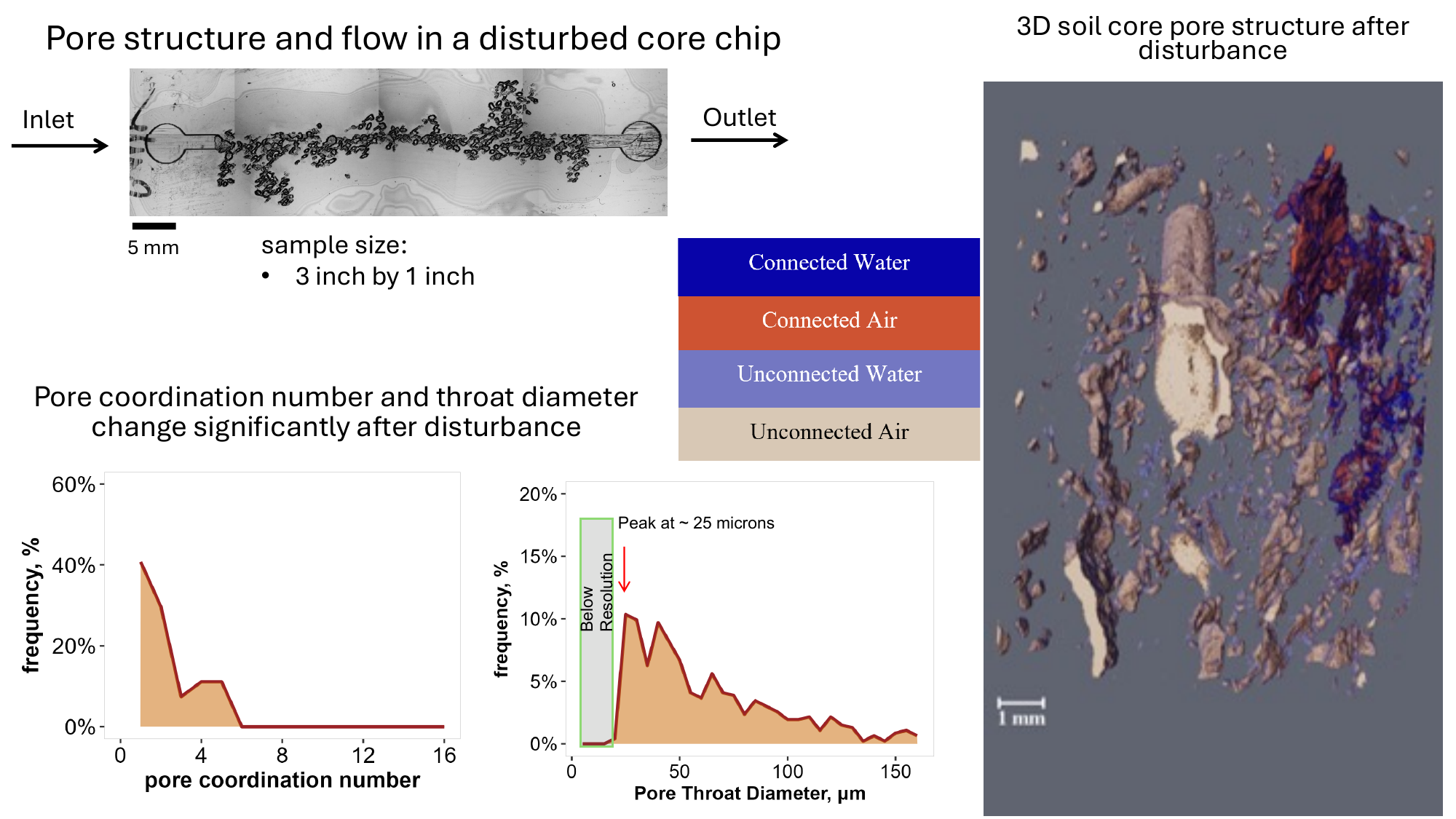}
    \caption{Experimental characterization of the disturbed soil pore network and its corresponding Pore2Chip micromodel. 
    The micromodel (upper left) represents the altered pore structure following disturbance and provides a controlled platform for examining the resulting flow response. 
    Different colors in the right figure show the connected pore network with (or including) water- and air-filled pores. 
    Pore network analysis (bottom figures) shows a shift toward lower pore coordination numbers and a change in the pore-throat diameter distribution, with the dominant resolved peak near $25~\mu\mathrm{m}$.
    The three-dimensional XCT-based rendering (right) illustrates the corresponding reorganization of the pore structure following disturbance on the same soil core.}
    \label{fig:exp_disturbed}
\end{figure}

Changes in pore-throat geometry provide a complementary description of this structural response. 
The resolved pore-throat-size distribution before disturbance shows a dominant peak near $35~\mu$m, whereas the corresponding distribution after disturbance shifts toward $25~\mu$m. 
A shift toward smaller characteristic pore throats can increase hydraulic resistance and alter flow partitioning among connected pathways. 
When combined with the observed increase in low-coordination pores, these changes provide a physically plausible explanation for the slower flow behavior observed in the disturbed Pore2Chip experiments. 
The experimental results therefore indicate that disturbance can affect flow by reducing network connectivity and modifying constrictions connecting neighboring pores. 
At the same time, pore-throat statistics close to the XCT resolution limit should be interpreted with care because changes in feature detectability can influence the extracted pore-network distributions. 
The experimental analysis therefore motivates using multiple complementary descriptors rather than relying on a single geometric metric to characterize disturbance.

These experimental observations reinforce the central hypothesis underlying GeONet: a porous system's response depends not only on material properties such as porosity and permeability, but also on the spatial arrangement and connectivity of the pore domain. 
In GeONet, these complementary sources of information are represented explicitly through separate modeling pathways. 
The Leaf Network represents the solid--void geometry, the Branch Network represents hydraulic material properties, and the Trunk Network provides the spatial query location together with local geometric descriptors. 
The predicted pressure at each query location is determined by the combined representation of pore geometry, material properties such as porosity and permeability, and local spatial information.
The disturbance experiments illustrate why this separation matters: two systems may have comparable bulk properties but differ substantially in pore coordination number, pore throat-size distribution, and pore network connectivity, leading to different flow responses.

The combined XCT--Pore2Chip dataset also provides a pathway toward experimentally constrained model calibration. 
Segmented XCT or Pore2Chip geometries can provide direct geometric inputs to the Leaf Network, while experimentally estimated porosity and effective hydraulic properties can constrain the Branch inputs. 
Measured inlet and outlet conditions define the hydraulic forcing, and micromodel observations provide independent quantities against which PFLOTRAN and GeONet predictions can be evaluated. 
In this workflow, uncertain parameters such as effective permeability can be adjusted so simulated and observed hydraulic responses are mutually consistent. 
Pore-network descriptors, including coordination number and pore-throat-size distributions, can additionally be used to characterize how closely a new experimental geometry resembles the structural domain represented by the RCP training ensemble. 
This creates an opportunity to distinguish uncertainty in hydraulic parameters from uncertainty in geometric representation.

The same experimental--computational framework can be progressively extended to multi-physics model calibration. 
Flow experiments provide an initial basis for constraining hydraulic parameters and geometry-dependent pressure responses. 
Tracer experiments can then provide time-dependent information to constrain reactive-transport behavior, including effective dispersion and accessible flow pathways. 
The Pore2Chip platform can also support experiments involving solute transport, mineral embedding, oxygen-rich flow, drying, or other perturbations designed to probe coupled physical, chemical, and biogeochemical processes. 
Such experiments would introduce additional observables that could be incorporated into PFLOTRAN-based calibration and, subsequently, into geometry-encoded neural-operator surrogates. 
The existing simulation dataset already retains tracer-concentration fields, providing a natural pathway to extend the present pressure-focused GeONet formulation toward coupled flow--transport prediction.

This integration suggests a broader experimental--computational workflow in which (i) XCT characterizes three-dimensional structural change, (ii) pore-network analysis identifies the geometric descriptors most affected by disturbance, (iii) Pore2Chip reproduces those characteristics in controlled laboratory platforms, (iv) PFLOTRAN provides mechanistic simulations of the associated flow and transport processes, and (v) GeONet provides a computationally efficient surrogate for repeated model evaluation. 
This workflow would enable experimental observations to constrain PFLOTRAN process model parameters, while simulations provide physically interpretable connections between pore-scale structural changes and system-scale responses. 
In the longer term, repeated interaction among experiments, high-fidelity simulations, and neural-operator predictions could support uncertainty analysis, inverse modeling, experimental design, and digital-twin applications for dynamically evolving porous-media systems.

%% file: sections/s5_conclusions.tex
\section{Conclusions}
\label{Sec:S5_conclusions}

This work introduced GeONet, a geometry-encoded neural operator for predicting pressure-driven flow in heterogeneous microfluidic porous-media domains. The central architectural contribution is the explicit separation of domain-scale geometry, hydraulic material properties, and query-local spatial information. A dedicated Leaf Network uses a multiscale spatial encoder to represent the complete binary solid--void geometry, while the Branch Network represents porosity and permeability, and the Trunk Network represents spatial coordinates together with local geometric descriptors and random Fourier features. These complementary representations are integrated through the DeepMixer, allowing pressure predictions to depend jointly on the hydraulic state, pore-space organization, and spatial location.
The numerical results establish the consistency of the reference data and the predictive performance of the proposed framework. Direct comparison of absolute, unscaled PFLOTRAN and pore-scale Stokes pressure fields for three representative RCP geometries yielded Pearson correlation coefficients of 0.9870--0.9911 and coefficients of determination of $R^2=0.9743$--$0.9823$. These comparisons were performed on common pore-space grid cells without pressure-offset subtraction, normalization, or pressure-drop scaling, supporting the use of the PFLOTRAN pressure fields as reference data for surrogate development.

Under matched data partitions and the two-stage HPO protocol, GeONet consistently attained lower validation objectives and substantially higher test-set pressure accuracy than the boundary-condition-aware DeepONet baseline. For test simulations 91, 95, and 99, GeONet achieved NRMSE values of 0.0381--0.0500, MAPE values of 1.76--2.59\%, and $R^2$ values of 0.9673--0.9822, compared with NRMSE values of 0.1483--0.2040, MAPE values of 7.81--10.45\%, and $R^2$ values of 0.5658--0.7126 for DeepONet. The spatial, distributional, and profile analyses further show that GeONet more accurately reproduces the pressure magnitude, irregular pressure contours, broad pressure distribution, and localized geometry-dependent variations observed in the PFLOTRAN reference solutions.
GeONet also maintained strong predictive performance on three additional higher-index RCP geometries excluded from model training and selection. For simulations 460--462, the model achieved NRMSE values of 0.0448--0.0676, MAPE values of 2.29--3.42\%, and $R^2$ values of 0.9511--0.9770. These results indicate that the learned operator generalizes to previously unseen RCP configurations within the geometry class represented in the training distribution. The results therefore support treating solid--void geometry as an explicit model input rather than relying solely on globally aggregated material properties and spatial coordinates to represent geometry-dependent pressure behavior.

The computational analysis further demonstrates favorable data efficiency and training scalability. Across the Tahoma and Deception HPC platforms, models trained with only 20\% of the available training data already achieved mean $R^2$ values of approximately 0.976--0.978 and mean MAPE values of approximately 2.1--2.2\%. Using the complete training set improved overall test performance, yielding mean $R^2$ values near 0.98 and MAPE values near 1.8\%. GPU acceleration reduced full-data, 100-epoch training from approximately 38.4 h to 3.6 h on Tahoma and from 61.1 h to 3.3 h on Deception. The close agreement between GPU- and CPU-trained models also indicates that the primary benefit of GPU execution is reduced wall-clock training time rather than a systematic change in converged predictive accuracy.
The high computational cost of generating the approximately 5000 PFLOTRAN realizations (about 13 million CPU-core-hours and nearly 30 TB of simulation data) provides an important motivation for reusable operator-learning surrogates. Once trained, GeONet evaluates the learned pressure operator without repeatedly solving the high-fidelity numerical flow problem. However, the present study did not conduct a controlled inference-latency benchmark; therefore, a numerical inference speedup relative to PFLOTRAN is claimed. Such a future benchmarking will be important for establishing performance under specified hardware, spatial resolution, and batch size conditions.

The present results are limited to supervised, deterministic prediction of steady-state pressure within the RCP geometry and physical-condition classes represented in the training data. Future work will extend the framework to the tracer-concentration fields already retained in the simulation dataset, investigate physics- and conservation-informed objectives and uncertainty quantification, and directly quantify inference latency and scalability. A further direction is integration with experimentally derived geometries from X-ray computed tomography, chemical imaging, and Pore2Chip platforms, providing a pathway toward experimentally constrained flow--transport prediction and multiphysics model calibration.
Overall, the results demonstrate that explicitly incorporating domain-scale solid--void structure together with material properties and query-local geometric information provides an effective operator-learning framework for heterogeneous porous-media flow. GeONet achieves accurate pressure reconstruction across held-out RCP geometries while retaining favorable data efficiency and computational training characteristics, providing a foundation for repeated simulation, uncertainty analysis, inverse modeling, experimental design, and future digital-twin applications.

%% file: sections/s6_math_appendix.tex
\section{Mathematical formulation and analysis of GeONet}
\label{app:geonet_math}

\setcounter{equation}{0}
\renewcommand{\theequation}{A-\arabic{equation}}

This appendix provides a mathematical formulation of the GeONet model. 
The goal is not to introduce a generic DeepONet analysis, but to formalize the geometry-conditioned operator that is specific to GeONet.
The model receives a binary solid--void geometry through the Leaf Network, heterogeneous porosity and permeability through the Branch Network, and a query coordinate through the Trunk Network. 
These inputs are fused through a gated bilinear module to approximate the pressure or transport state at the queried point.
The mathematics below describes the architecture used in the pressure-dominated experiments reported in the Sec.~\ref{Sec:S3_methods}. 
This formulation clarifies why explicit geometry encoding is required, how the GeONet hypothesis class differs from coordinate-centric operator networks, and how the reported pressure metrics relate to the underlying Darcy operator.

\subsection{Geometry-conditioned flow operator}
\label{app:operator_definition}

Let $\Omega\subset {\mathbb R}^d$, $d=2$ or $3$, denote the fixed chip bounding box. 
A voxelized Reduced Complexity Platform (RCP) geometry is represented by a binary mask
\begin{equation}
    m_G:\Omega_h\rightarrow \{0,1\},
\end{equation}
where $\Omega_h$ is the computational grid, $m_G(\mathbf{x})=1$ indicates a fluid-accessible cell, and $m_G(\mathbf{x})=0$ indicates a solid or impermeable cell. 
The corresponding fluid region is denoted by
\begin{equation}
    \Omega_f(G)=\{\mathbf{x}\in \Omega: m_G(\mathbf{x})=1\},
\end{equation}
with inlet boundary $\Gamma_{\rm in}(G)$, outlet boundary $\Gamma_{\rm out}(G)$, solid--fluid interface $\Gamma_s(G)$, and remaining no-flow boundary $\Gamma_{\rm nf}(G)$.

\begin{definition}[RCP state]
\label{def:rcp_state}
An RCP state is the tuple
\begin{equation}
    a_G = \left(m_G,\phi_G,k_G,\mathcal{B}_G\right),
\end{equation}
where $m_G$ is the binary geometry mask, $\phi_G$ is porosity, $k_G$ is intrinsic permeability, and $\mathcal{B}_G$ denotes boundary information, including inlet, outlet, and no-flow boundaries.
In the simulation dataset described in Sec.~\ref{Sec:S2_data_generation}, $\mathcal{B}_G$ is fixed up to the geometry-induced pore mask, while $m_G$, $\phi_G$, and $k_G$ vary across samples.
\end{definition}

Let $\mathbf{s}=(\mathbf{x},t)$ denote a query point in space--time.
For pressure-only steady-flow evaluation, $\mathbf{s}=\mathbf{x}$ is sufficient.
The time coordinate is retained to keep the notation compatible with the pressure--tracer data representation for our future work.
The physical solution operator is
\begin{equation}
    \mathcal{T}^{\dagger}: a_G \mapsto y_G,
    \qquad
    y_G(\mathbf{s})=\left(p_G(\mathbf{s}),c_G(\mathbf{s})\right),
\end{equation}
where $p_G$ is pressure and $c_G$ is tracer concentration.
The pressure simulations in the Sec.~\ref{Sec:S2_data_generation} evaluate the scalar component $p_G$.

\begin{definition}[GeONet target operator]
\label{def:geonet_target_operator}
The GeONet target is the coordinate-evaluation operator
\begin{equation}
    \mathcal{G}^{\dagger}(m_G,\phi_G,k_G,\mathbf{s}) = y_G(\mathbf{s}),
\end{equation}
where $y_G=\mathcal{T}^{\dagger}(a_G)$.
Thus GeONet is not an image-to-image interpolator only, it is an evaluator of the geometry-conditioned solution field at arbitrary grid or off-grid query coordinates.
\end{definition}

The supervised dataset can therefore be written as
\begin{equation}
    \mathcal{D}_N=\left\{\left(m_{G_j},\phi_{G_j},k_{G_j},\mathbf{s}_{j\ell},y_{G_j}(\mathbf{s}_{j\ell})\right):
    j=1,\ldots,N_G,\;\ell=1,\ldots,N_j\right\}.
\end{equation}
The coordinate-sampling strategy used in training corresponds to drawing a subset of the indices $\ell$ for each geometry $G_j$ rather than presenting the full grid at every optimization step.

\subsection{Assumptions for admissible GeONet inputs}
\label{app:axioms}

The following assumptions define the mathematical class of geometries and material fields for the analysis and training presented in Sec.~\ref{Sec:S3_methods}. 
They are aligned with the RCP pressure-prediction setting rather than with a fully general neural-operator problem.

\begin{assumption}[Inlet--outlet connectivity:]
\label{ax:connectivity}
Each admissible geometry contains at least one connected fluid path from $\Gamma_{\rm in}(G)$ to $\Gamma_{\rm out}(G)$. Solid cells may create isolated dead-end regions and stagnant regions, but the pressure problem is evaluated on the connected pore component that communicates with the inlet and outlet.
\end{assumption}

\begin{assumption}[Bounded hydraulic coefficients:]
\label{ax:coefficients}
On the connected fluid region, porosity and permeability satisfy
\begin{equation}
    0<\phi_{\min}\leq \phi_G(\mathbf{x})\leq \phi_{\max}\leq 1,
    \qquad
    0<k_{\min}\leq k_G(\mathbf{x})\leq k_{\max}<\infty .
\end{equation}
As a fixed-domain formulation is used, solid cells are assigned a residual permeability $k_s$ with $0<k_s\ll k_{\min}$, or are masked out of the pressure loss and residual evaluations.
\end{assumption}

\begin{assumption}[Pressure-drop normalization:]
\label{ax:normalization}
Pressure is compared after removing the arbitrary additive constant and enforcing a common inlet--outlet pressure drop. 
Equivalently, pressure is evaluated modulo constants unless Dirichlet inlet and outlet values are specified. 
The normalized comparison used in our analysis fixes the outlet mean and inlet--outlet pressure drop.
\end{assumption}

\begin{assumption}[Distribution of geometries:]
\label{ax:distribution}
Training and test geometries are drawn from a probability distribution $\rho_G$ supported on admissible RCP masks. Cross-domain generalization means that a test geometry is not seen during training but remains in, or near, the support of $\rho_G$ under the geometry metric defined below.
\end{assumption}

\begin{assumption}[Nonredundant geometry:]
\label{ax:nonredundant_geometry}
The binary geometry mask is not assumed to be recoverable from the scalar porosity and permeability values available to the Branch Network. Therefore, $m_G$ is treated as a distinct physical input rather than as a redundant feature.
\end{assumption}

Assumption~\ref{ax:nonredundant_geometry} is central to GeONet. 
It mathematically separates the role of the Leaf Network from the Branch Network: porosity and permeability encode fluid storage and conductance, while the geometry mask encodes pore connectivity, pore boundaries, and blocked pathways.

\subsection{Weak pressure formulation and well-posedness}
\label{app:weak_form}

For a fixed geometry $G$, define the hydraulic conductivity
\begin{equation}
    \kappa_G(\mathbf{x})=\frac{k_G(\mathbf{x})}{\mu}
\end{equation}
on $\Omega_f(G)$.
Note that each geometry has a different solid–void mask; the fluid domain, pore connectivity, and boundary layout can change.
These geometries are binary solid–void masks, in which void regions are fluid-flow pathways and solid regions constrain the flow.
Pressure is prescribed at the inlet/outlet boundaries, with no flow elsewhere.
For steady incompressible Darcy flow with no sources and negligible gravity, the pressure satisfies
\begin{equation}
    -\nabla\cdot\left(\kappa_G\nabla p_G\right)=0
    \qquad \mathrm{in}\;\Omega_f(G),
\end{equation}
with prescribed pressure on $\Gamma_{\rm in}(G)$ and $\Gamma_{\rm out}(G)$ and no-flow conditions on $\Gamma_{\rm nf}(G)\cup\Gamma_s(G)$.

Let $p_D$ be any function satisfying the prescribed inlet and outlet pressures.
Define
\begin{equation}
    V_{0,G}=\left\{v\in H^1(\Omega_f(G)): v=0\;\mathrm{on}\;\Gamma_{\rm in}(G)\cup\Gamma_{\rm out}(G)\right\}.
\end{equation}
where $H^1(\Omega_G)$ is the Sobolev space of square-integrable functions whose gradients are also square-integrable.
The weak pressure problem is: find $p_G=p_D+w_G$, with $w_G\in V_{0,G}$, such that
\begin{equation}
    \int_{\Omega_f(G)} \kappa_G\nabla w_G\cdot \nabla v\,d\mathbf{x}
    =-
    \int_{\Omega_f(G)} \kappa_G\nabla p_D\cdot \nabla v\,d\mathbf{x}
    \qquad \forall v\in V_{0,G}.
\end{equation}
$(w_G \in V_{0,G})$ means: $(w_G)$ is an admissible test function for geometry $(G)$, and it vanishes on the pressure-prescribed boundaries.

\begin{theorem}[Well-posedness of the geometry-conditioned pressure map:]
\label{thm:wellposedness}
Under Assumptions~\ref{ax:connectivity}--\ref{ax:normalization}, the weak Darcy pressure problem has a unique solution in $p_D+V_{0,G}$.
Moreover, there exists a constant $C_G$ depending on the geometry, boundary partition, and coefficient bounds such that
\begin{equation}
    \|p_G\|_{H^1(\Omega_f(G))}
    \leq C_G\, |\Delta p|,
\end{equation}
where $\Delta p=p_{\mathrm{in}}-p_{\mathrm{out}}$ is the imposed pressure drop.
\end{theorem}

\begin{proof}
Let $g_D$ denote the prescribed pressure data on $\Gamma_{D,G}$:
\begin{equation}
    g_D =
    \begin{cases}
    p_{\mathrm{in}}, & \text{on } \Gamma_{\mathrm{in},G},\\
    p_{\mathrm{out}}, & \text{on } \Gamma_{\mathrm{out},G}.
    \end{cases}
\end{equation}
Because pressure is defined only up to an additive constant for incompressible flow, we may subtract the outlet pressure and write
\begin{subequations}
\begin{align}
    p_{\mathrm{out}}&=0,\\
    p_{\mathrm{in}}&=\Delta p.
\end{align}
\end{subequations}
By the trace and lifting theorem on the Lipschitz domain $\Omega_f(G)$ \cite{mudunuru2016enforcing}, there exists a lift $p_D\in H^1(\Omega_f(G))$ such that
\begin{equation}
    T_G p_D = g_D
    \qquad
    \text{on } \Gamma_{D,G},
\end{equation}
and
\begin{equation}
    \|p_D\|_{H^1(\Omega_f(G))}
    \le
    C_{\mathrm{lift},G}
    \|g_D\|_{H^{1/2}(\Gamma_{D,G})}.
\end{equation}
Since the boundary data are determined by the imposed pressure drop,
\begin{equation}
    \|g_D\|_{H^{1/2}(\Gamma_{D,G})}
    \le
    C_{\Gamma,G} |\Delta p|,
\end{equation}
and therefore
\begin{equation}
    \|p_D\|_{H^1(\Omega_f(G))}
    \le
    C_{\mathrm{lift},G}C_{\Gamma,G}|\Delta p|.
\end{equation}

We write the pressure solution as
\begin{equation}
    p_G = p_D + w_G,
    \qquad
    w_G \in V_{0,G}.
\end{equation}
The weak Darcy problem is then equivalent to finding $w_G\in V_{0,G}$ such that
\begin{equation}
    a_G(w_G,v)
    =
    -a_G(p_D,v)
    \qquad
    \forall v\in V_{0,G},
\end{equation}
where
\begin{equation}
    a_G(u,v)
    =
    \int_{\Omega_f(G)}
    \kappa_G \nabla u\cdot\nabla v\,d\mathbf{x}.
\end{equation}

The bilinear form $a_G$ is bounded on $H^1(\Omega_f(G))$.
Indeed, using $\kappa_G\le \kappa_{\max}$,
\begin{equation}
    |a_G(u,v)|
    \le
    \kappa_{\max}
    \|\nabla u\|_{L^2(\Omega_f(G))}
    \|\nabla v\|_{L^2(\Omega_f(G))}
    \le
    \kappa_{\max}
    \|u\|_{H^1(\Omega_f(G))}
    \|v\|_{H^1(\Omega_f(G))}.
\end{equation}

The bilinear form is also coercive on $V_{0,G}$.
Since $\kappa_G\ge \kappa_{\min}$,
\begin{equation}
    a_G(v,v)
    =
    \int_{\Omega_f(G)}
    \kappa_G |\nabla v|^2\,d\mathbf{x}
    \ge
    \kappa_{\min}
    \|\nabla v\|^2_{L^2(\Omega_f(G))}.
\end{equation}
Because functions in $V_{0,G}$ vanish on the inlet and outlet Dirichlet boundaries, the Poincar\'e inequality \cite{mudunuru2016enforcing} gives
\begin{equation}
    \|v\|_{L^2(\Omega_f(G))}
    \le
    C_{P,G}
    \|\nabla v\|_{L^2(\Omega_f(G))}.
\end{equation}
Hence
\begin{equation}
    \|v\|^2_{H^1(\Omega_f(G))}
    =
    \|v\|^2_{L^2(\Omega_f(G))}
    +
    \|\nabla v\|^2_{L^2(\Omega_f(G))}
    \le
    (1+C_{P,G}^2)
    \|\nabla v\|^2_{L^2(\Omega_f(G))},
\end{equation}
or equivalently,
\begin{equation}
    \|\nabla v\|^2_{L^2(\Omega_f(G))}
    \ge
    \frac{1}{1+C_{P,G}^2}
    \|v\|^2_{H^1(\Omega_f(G))}.
\end{equation}
Therefore,
\begin{equation}
    a_G(v,v)
    \ge
    \frac{\kappa_{\min}}{1+C_{P,G}^2}
    \|v\|^2_{H^1(\Omega_f(G))}.
\end{equation}
Thus $a_G$ is coercive on $V_{0,G}$.

The linear functional
\begin{equation}
    \ell_G(v) = -a_G(p_D,v)
\end{equation}
is bounded because
\begin{equation}
    |\ell_G(v)|
    =
    |a_G(p_D,v)|
    \le
    \kappa_{\max}
    \|\nabla p_D\|_{L^2(\Omega_f(G))}
    \|\nabla v\|_{L^2(\Omega_f(G))}
    \le
    \kappa_{\max}
    \|p_D\|_{H^1(\Omega_f(G))}
    \|v\|_{H^1(\Omega_f(G))}.
\end{equation}
The Lax--Milgram theorem therefore guarantees the existence and uniqueness of $w_G\in V_{0,G}$.
Consequently,
\begin{equation}
    p_G = p_D+w_G
\end{equation}
is the unique solution in $p_D+V_{0,G}$.
To derive the energy estimate, taking $v=w_G$ in the weak problem gives
\begin{equation}
    a_G(w_G,w_G)
    =
    -a_G(p_D,w_G).
\end{equation}
Using coercivity on the left and boundedness on the right,
\begin{equation}
    \kappa_{\min}
    \|\nabla w_G\|^2_{L^2(\Omega_f(G))}
    \le
    \kappa_{\max}
    \|\nabla p_D\|_{L^2(\Omega_f(G))}
    \|\nabla w_G\|_{L^2(\Omega_f(G))}.
\end{equation}
If $w_G\neq 0$, division by $\|\nabla w_G\|_{L^2(\Omega_f(G))}$ yields
\begin{equation}
    \|\nabla w_G\|_{L^2(\Omega_f(G))}
    \le
    \frac{\kappa_{\max}}{\kappa_{\min}}
    \|\nabla p_D\|_{L^2(\Omega_f(G))}.
\end{equation}
The Poincar\'e inequality then implies
\begin{equation}
    \|w_G\|_{H^1(\Omega_f(G))}
    \le
    \sqrt{1+C_{P,G}^2}
    \frac{\kappa_{\max}}{\kappa_{\min}}
    \|p_D\|_{H^1(\Omega_f(G))}.
\end{equation}
Therefore,
\begin{equation}
    \|p_G\|_{H^1(\Omega_f(G))}
    \le
    \|p_D\|_{H^1(\Omega_f(G))}
    +
    \|w_G\|_{H^1(\Omega_f(G))}.
\end{equation}
Combining the previous estimates gives
\begin{equation}
    \|p_G\|_{H^1(\Omega_f(G))}
    \le
    \left[
    1+
    \sqrt{1+C_{P,G}^2}
    \frac{\kappa_{\max}}{\kappa_{\min}}
    \right]
    C_{\mathrm{lift},G}C_{\Gamma,G}
    |\Delta p|.
\end{equation}
Thus,
\begin{equation}
    \|p_G\|_{H^1(\Omega_f(G))}
    \le
    C_G |\Delta p|,
\end{equation}
where
\begin{equation}
    C_G =
    \left[
    1+
    \sqrt{1+C_{P,G}^2}
    \frac{\kappa_{\max}}{\kappa_{\min}}
    \right]
    C_{\mathrm{lift},G}C_{\Gamma,G}.
\end{equation}
This completes the proof.
\end{proof}

Theorem~\ref{thm:wellposedness} justifies treating
\begin{equation}
    (m_G,\phi_G,k_G)\mapsto p_G
\end{equation}
as a deterministic operator on admissible geometries.
The role of GeONet is to learn the coordinate-evaluation form of this operator.

\subsection{Finite-volume graph formulation of RCP pressure flow: pore-space connectivity and edge transmissibility}
\label{app:graph_formulation}

For the pressure component, this two-point finite-volume discretization on a structured grid can be interpreted as a weighted graph problem \cite{karra2018modeling}.
Let $V_G$ be the set of fluid cells and let $E_G$ be the set of neighboring cell pairs connected through fluid-accessible faces.
Each edge $e=(i,j)$ has a transmissibility $\tau_e(G,k)$ determined by geometry, face area, distance, and permeability.
For a fixed geometry $G$, let $V_G$ denote the set of active pore-space cells or graph vertices in the discretized fluid domain.
The discrete pressure field is represented by the vector
\begin{equation}
    \mathbf{p}_G = (p_i)_{i\in V_G} \in \mathbb{R}^{|V_G|},
\end{equation}
where $p_i$ approximates the continuum pressure $p_G(\mathbf{x}_i)$ at the cell center or graph vertex $\mathbf{x}_i$.
The discrete pressure vector $\mathbf{p}_G\in {\mathbb R}^{|V_G|}$ solves
\begin{subequations}
\begin{align}
    A_G\mathbf{p}_G&=\mathbf{f}_G,\\
    A_G&=B_G^{T}C_GB_G + A_{D,G},
\end{align}
\end{subequations}
where $B_G$ is a signed incidence matrix, $C_G=\operatorname{diag}(\tau_e(G,k))$ is the edge-conductance matrix, and $A_{D,G}$ contains Dirichlet boundary contributions.

\begin{proposition}[The solid--void mask controls the discrete pressure operator:]
\label{prop:graph_geometry}
Let $m_G$ denote the binary solid--void mask associated with geometry $G$, where $m_G(i)=1$ indicates an active fluid cell and $m_G(i)=0$ indicates a solid cell.
In a cell-centered finite-volume discretization, the pressure field satisfies a graph-based linear system
\begin{equation}
    A_G \mathbf{p}_G = \mathbf{b}_G,
\end{equation}
where the graph $H_G=(V_G,E_G)$ is induced by the pore space.
The node set
\begin{equation}
    V_G=\{i:\,m_G(i)=1\}
\end{equation}
contains the active fluid cells, and the edge set $E_G$ contains pairs of neighboring fluid cells that share an unblocked face.
For an interior edge $(i,j)\in E_G$, the flux has the form
\begin{equation}
    q_{ij,G}
    =
    -T_{ij,G}\,(p_{j,G}-p_{i,G}),
\end{equation}
where $T_{ij,G}$ is the face transmissibility.
Hence the entries of the discrete pressure matrix satisfy
\begin{subequations}
\begin{align}
    (A_G)_{ij}&=-T_{ij,G},\qquad i\neq j,\\
    (A_G)_{ii}
    &=
    \sum_{j:(i,j)\in E_G} T_{ij,G}
    +
    T^{D}_{i,G},
\end{align}
\end{subequations}
where $T^{D}_{i,G}$ denotes any transmissibility contribution associated with Dirichlet boundary faces.

Changing the solid--void mask changes $V_G$, $E_G$, and therefore the sparsity pattern and weights of $A_G$.
Equivalently, on a fixed background grid, the mask introduces an edge indicator $\chi_{ij,G}$ such that
\begin{subequations}
\begin{align}
    T_{ij,G}
    &=
    \chi_{ij,G}
    \frac{|F_{ij}|}{d_{ij}}
    \left(\frac{k_G}{\mu}\right)^{H}_{ij},\\
    \chi_{ij,G}&\in\{0,1\}.
\end{align}
\end{subequations}
Here $|F_{ij}|$ is the shared face area, $d_{ij}$ is the distance between cell centers, and $\left(k_G/\mu\right)^{H}_{ij}$ is a harmonic average of the neighboring hydraulic conductivities.
If either cell is solid or if the face is blocked by the geometry, then $\chi_{ij,G}=0$ and the corresponding flux connection is removed.
Therefore, the binary mask affects the pressure solution through graph topology, boundary connectivity, and conductance weights, not merely through pointwise material values.
\end{proposition}

\begin{proof}
For a finite-volume discretization of Darcy flow, conservation of mass in an active fluid cell $i\in V_G$ gives
\begin{equation}
    \sum_{j:(i,j)\in E_G} q_{ij,G} + q^{D}_{i,G} = f_{i,G},
\end{equation}
where $q_{ij,G}$ is the flux between neighboring fluid cells, $q^{D}_{i,G}$ is the contribution from pressure-prescribed boundary faces, and $f_{i,G}$ contains sources or sinks.
The two-point flux approximation gives
\begin{equation}
    q_{ij,G}
    =
    -T_{ij,G}\,(p_{j,G}-p_{i,G}).
\end{equation}
Thus cell $j$ contributes to the equation for cell $i$ only if the edge $(i,j)$ exists and has positive transmissibility.
The solid--void mask determines whether this edge exists.
If either cell is solid, or if the common face is blocked by a solid barrier, then no physical fluid flux can cross that face.
In that case,
\begin{equation}
    T_{ij,G}=0,
\end{equation}
which is equivalent to removing the edge $(i,j)$ from the graph.
If both cells are fluid and the face is open, then $T_{ij,G}>0$ and its magnitude is controlled by the face geometry, the cell spacing, viscosity, and permeability.
Therefore, the mask determines the active degrees of freedom $V_G$, the neighbor connections $E_G$, and the sparsity pattern of the discrete operator $A_G$.
Permeability controls the magnitude of the retained conductances, but the mask determines which conductances exist in the first place.
Hence the pressure field depends on geometry through the topology and weights of the finite-volume graph.
\end{proof}

This graph interpretation clarifies the role of the Leaf Network in GeONet.
The Branch Network encodes material fields such as porosity and permeability, which influence the magnitudes of local conductances.
However, these fields alone generally do not specify the graph incidence structure: two neighboring cells with similar porosity and permeability values may be either hydraulically connected by an open pore throat or disconnected by a solid barrier.
These two cases produce different transmissibilities, different sparsity patterns in $A_G$, and therefore different pressure solutions.
The Leaf Network supplies the missing structural information by encoding the binary geometry mask, enabling GeONet to condition its pressure prediction on both material properties and pore-space connectivity.

\subsection{GeONet hypothesis class: geometry-conditioned operator representation}
\label{app:hypothesis_class}
This section formalizes the GeONet hypothesis class, defined as the family of geometry-conditioned solution operators representable by the proposed Leaf--Branch--Trunk architecture.
For the RCP pressure problem, the target map is not a conventional coefficient-to-solution operator depending only on porosity, permeability, and spatial coordinates.
The solid--void mask $m_G$ defines the fluid subdomain $\Omega_f(G)$, the pressure and no-flow boundary partitions, and the finite-volume connectivity graph through which pressure equilibrates.
As a result, two samples with similar local values of $\phi_G$ and $k_G$ may produce different pressure responses if their pore-space connectivity or nearby solid barriers differ.
The appropriate target map is therefore a geometry-conditioned operator,
\begin{equation}
    \mathcal{G}^{\dagger}
    :
    (m_G,\phi_G,k_G,\mathbf{s})
    \longmapsto
    y_G(\mathbf{s}),
\end{equation}
where $y_G$ denotes the pressure or transport response evaluated at the query coordinate $\mathbf{s}$.
This viewpoint motivates the defining feature of GeONet: the geometry mask is encoded by a dedicated Leaf Network before being fused with material and coordinate features.
The resulting hypothesis class can therefore represent operators whose coordinate-wise predictions depend explicitly on pore topology, connectivity, and boundary-induced flow constraints.

Let $G$ denote an admissible RCP geometry, and let
\begin{equation}
    m_G:\Omega\rightarrow\{0,1\}
\end{equation}
be its binary solid--void mask, where $m_G(\mathbf{x})=1$ denotes fluid-accessible pore space and $m_G(\mathbf{x})=0$ denotes solid material.
Let
\begin{subequations}
\begin{align}
    \phi_G&:\Omega\rightarrow\mathbb{R},\\
    k_G&:\Omega\rightarrow\mathbb{R}
\end{align}
\end{subequations}
denote the porosity and permeability fields associated with geometry $G$.
For a query coordinate
\begin{equation}
    \mathbf{s}=(x,y,z,t)\in\mathcal{S},
\end{equation}
the target solution is written as
\begin{equation}
    y_G(\mathbf{s})
    =
    \mathcal{G}^{\dagger}(m_G,\phi_G,k_G)(\mathbf{s}),
\end{equation}
where $y_G$ may represent pressure alone, $y_G=p_G$, or a vector of physical outputs, for example $y_G=(p_G,c_G)$ for pressure and tracer concentration.

GeONet approximates this geometry-conditioned solution operator using three input streams.
The Leaf Network encodes the solid--void geometry, the Branch Network encodes material properties, and the Trunk Network encodes the query coordinate.
Let
\begin{equation}
    L_{\theta_L},\qquad B_{\theta_B},\qquad T_{\theta_T}
\end{equation}
denote the Leaf, Branch, and Trunk networks, respectively.
Let $\gamma(\mathbf{s})$ be a Fourier-feature embedding of the query coordinate.
The latent features are
\begin{subequations}
\begin{align}
    \boldsymbol{\ell}_G(\mathbf{s})
    &=
    L_{\theta_L}(m_G;\mathbf{s}),\\
    \mathbf{b}_G
    &=
    B_{\theta_B}(\phi_G,k_G),\\
    \boldsymbol{\tau}(\mathbf{s})
    &=
    T_{\theta_T}(\gamma(\mathbf{s})).
\end{align}
\end{subequations}
Here $\boldsymbol{\ell}_G(\mathbf{s})$ denotes the geometry feature available at the query coordinate.
This notation covers two common implementations.
In the first, the Leaf Network outputs a global geometry latent vector, in which case $\boldsymbol{\ell}_G(\mathbf{s})=\boldsymbol{\ell}_G$ is constant with respect to $\mathbf{s}$.
In the second, the Leaf Network outputs a spatial feature map, and $\boldsymbol{\ell}_G(\mathbf{s})$ is obtained by sampling or interpolating that feature map at the query coordinate.
The second form is particularly natural for pore-scale pressure prediction because it allows the geometry representation to vary with spatial location.

With these definitions, the GeONet hypothesis class is
\begin{equation}
\mathcal{H}_{\mathrm{GeONet}}
=
\left\{
\widehat{\mathcal{G}}_{\theta}:
(m_G,\phi_G,k_G,\mathbf{s})
\mapsto
D_{\theta_D}
\left[
M_{\theta_M}
\left(
L_{\theta_L}(m_G;\mathbf{s}),
B_{\theta_B}(\phi_G,k_G),
T_{\theta_T}(\gamma(\mathbf{s}))
\right)
\right]
:
\theta\in\Theta
\right\},
\end{equation}
where
\begin{equation}
    \theta =
    (\theta_L,\theta_B,\theta_T,\theta_M,\theta_D)
\end{equation}
collects all trainable parameters, $M_{\theta_M}$ is the feature-fusion module, and $D_{\theta_D}$ is the decoder.
Thus, each element of $\mathcal{H}_{\mathrm{GeONet}}$ is a candidate geometry-conditioned operator that maps a solid--void mask, material fields, and a coordinate to a predicted physical response.

A representative gated bilinear fusion module can be written as follows.
First, the three latent streams are projected to a common dimension $r$:
\begin{subequations}
\begin{align}
    \widetilde{\boldsymbol{\ell}}_G(\mathbf{s})
    &=
    P_{\ell}\boldsymbol{\ell}_G(\mathbf{s}),\\
    \widetilde{\mathbf{b}}_G
    &=
    P_b\mathbf{b}_G,\\
    \widetilde{\boldsymbol{\tau}}(\mathbf{s})
    &=
    P_{\tau}\boldsymbol{\tau}(\mathbf{s}),
\end{align}
\end{subequations}
where
\begin{equation}
    \widetilde{\boldsymbol{\ell}}_G(\mathbf{s}),
    \widetilde{\mathbf{b}}_G,
    \widetilde{\boldsymbol{\tau}}(\mathbf{s})
    \in \mathbb{R}^r.
\end{equation}
The gate is then defined by
\begin{equation}
    \mathbf{g}_G(\mathbf{s})
    =
    \sigma
    \left(
    W_g
    \left[
    \widetilde{\boldsymbol{\ell}}_G(\mathbf{s}),
    \widetilde{\mathbf{b}}_G,
    \widetilde{\boldsymbol{\tau}}(\mathbf{s})
    \right]
    +
    \mathbf{c}_g
    \right),
\end{equation}
where $\sigma$ is applied componentwise.
The fused representation is
\begin{equation}
\begin{aligned}
    \mathbf{h}_G(\mathbf{s})
    =
    &\;
    \mathbf{g}_G(\mathbf{s})
    \odot
    U_{\ell\tau}
    \left(
    \widetilde{\boldsymbol{\ell}}_G(\mathbf{s})
    \odot
    \widetilde{\boldsymbol{\tau}}(\mathbf{s})
    \right)
    \\[1mm]
    &+
    \left(
    \mathbf{1}-\mathbf{g}_G(\mathbf{s})
    \right)
    \odot
    U_{b\tau}
    \left(
    \widetilde{\mathbf{b}}_G
    \odot
    \widetilde{\boldsymbol{\tau}}(\mathbf{s})
    \right)
    \\[1mm]
    &+
    U_{\ell b}
    \left(
    \widetilde{\boldsymbol{\ell}}_G(\mathbf{s})
    \odot
    \widetilde{\mathbf{b}}_G
    \right)
    \\[1mm]
    &+
    U_{\ell b\tau}
    \left(
    \widetilde{\boldsymbol{\ell}}_G(\mathbf{s})
    \odot
    \widetilde{\mathbf{b}}_G
    \odot
    \widetilde{\boldsymbol{\tau}}(\mathbf{s})
    \right).
\end{aligned}
\end{equation}
where \(\odot\) means the element-wise (Hadamard) product. The predicted output is then
\begin{equation}
    \widehat{y}_{\theta}(\mathbf{s};m_G,\phi_G,k_G)
    =
    D_{\theta_D}\left(\mathbf{h}_G(\mathbf{s})\right).
\end{equation}
The exact channel widths, nonlinearities, and linear maps used in an implementation may vary.
Please see the data availability and software section, which provides details on the GeONet code availability.
The defining property is the same: GeONet forms multiplicative interactions among geometry, material state, and query coordinates before decoding the physical response.

\begin{definition}[Geometry-coordinate cross-interaction]
\label{def:cross_interaction}
A predictor
\begin{equation}
    \widehat{y}
    =
    F(\boldsymbol{\ell},\mathbf{b},\boldsymbol{\tau})
\end{equation}
has a geometry-coordinate cross-interaction if the sensitivity of the output with respect to the coordinate latent variable depends on the geometry latent variable.
Equivalently,
\begin{equation}
    D_{\boldsymbol{\ell}}
    \left(
    D_{\boldsymbol{\tau}} F
    \right)
    \neq 0
\end{equation}
for at least some admissible inputs.
For scalar outputs, this condition reduces to the existence of a nonzero mixed derivative
\begin{equation}
    \frac{\partial^2 \widehat{y}}
    {\partial \ell_i\,\partial \tau_j}
    \neq 0
\end{equation}
for some pair of latent coordinates $(i,j)$.
\end{definition}

This condition is important for the RCP pressure problem.
If changing the geometry only shifts a global latent vector but does not alter the coordinate sensitivity of the prediction, then the model has limited ability to represent geometry-dependent pressure gradients.
In contrast, a geometry-coordinate cross-interaction allows the same coordinate $\mathbf{s}$ to be interpreted differently depending on whether it lies near an open channel, a pore throat, a dead-end region, or a solid barrier.

\begin{proposition}[Gated bilinear fusion induces geometry-coordinate cross-interactions]
\label{prop:geonet_cross_interaction}
Assume that the decoder and fusion weights are not identically zero.
Then the GeONet fusion module defined above contains geometry-coordinate cross-interactions.
\end{proposition}

\begin{proof}
Consider the bilinear term
\begin{equation}
    U_{\ell\tau}
    \left(
    \widetilde{\boldsymbol{\ell}}_G(\mathbf{s})
    \odot
    \widetilde{\boldsymbol{\tau}}(\mathbf{s})
    \right).
\end{equation}
For simplicity, we write
\begin{equation}
    \mathbf{z}_{\ell\tau}
    =
    \widetilde{\boldsymbol{\ell}}
    \odot
    \widetilde{\boldsymbol{\tau}}.
\end{equation}
The derivative of $\mathbf{z}_{\ell\tau}$ with respect to $\widetilde{\boldsymbol{\tau}}$ is
\begin{equation}
    D_{\widetilde{\boldsymbol{\tau}}}\mathbf{z}_{\ell\tau}
    =
    \operatorname{diag}
    \left(
    \widetilde{\boldsymbol{\ell}}
    \right).
\end{equation}
Therefore, the coordinate sensitivity of this term depends directly on the geometry latent variable.
Differentiating once more with respect to $\widetilde{\boldsymbol{\ell}}$ gives a nonzero mixed derivative whenever the corresponding weights in $U_{\ell\tau}$ and the decoder are nonzero.
The gate $\mathbf{g}_G(\mathbf{s})$ also depends jointly on $\widetilde{\boldsymbol{\ell}}_G(\mathbf{s})$, $\widetilde{\mathbf{b}}_G$, and $\widetilde{\boldsymbol{\tau}}(\mathbf{s})$, which introduces additional geometry-coordinate interactions.
Hence GeONet contains terms for which the coordinate response changes with the geometry representation.
\end{proof}

\begin{proposition}[DeepONet is a limiting case of GeONet]
\label{prop:deeponet_subset}
A standard Branch--Trunk DeepONet of finite rank $r$ is contained in the GeONet hypothesis class as a special case.
\end{proposition}

\begin{proof}
A scalar DeepONet approximation has the form
\begin{equation}
    \widehat{y}_{\mathrm{DON}}(\mathbf{s})
    =
    \sum_{i=1}^{r}
    b_i(\phi_G,k_G)\,
    t_i(\mathbf{s})
    +
    c,
\end{equation}
where $b_i$ are Branch-network features and $t_i$ are Trunk-network features.
To recover this form from GeONet, choose the Leaf output to be a constant vector independent of $m_G$, set the gate to select only the Branch--Trunk interaction, set
\begin{equation}
    \mathbf{h}_G(\mathbf{s})
    =
    \mathbf{b}_G
    \odot
    \boldsymbol{\tau}(\mathbf{s}),
\end{equation}
and choose a linear decoder
\begin{equation}
    D_{\theta_D}(\mathbf{h})
    =
    \mathbf{1}^{\top}\mathbf{h}+c.
\end{equation}
Then, we have
\begin{equation}
    \widehat{y}_{\theta}(\mathbf{s})
    =
    \sum_{i=1}^{r}
    b_i(\phi_G,k_G)\,
    t_i(\mathbf{s})
    +
    c,
\end{equation}
which is exactly the canonical DeepONet representation.
Therefore, DeepONet is a special limiting case of GeONet obtained when the geometry stream is ignored or made constant.
\end{proof}

\begin{remark}\textbf{Why the hypothesis class is geometry-aware?}~The preceding proposition shows that GeONet can reproduce the Branch--Trunk structure of DeepONet.
However, GeONet is not limited to that structure.
Through the Leaf stream and the gated bilinear fusion, GeONet can represent terms such as
\begin{equation}
    \boldsymbol{\ell}_G(\mathbf{s})
    \odot
    \boldsymbol{\tau}(\mathbf{s}),
    \qquad
    \boldsymbol{\ell}_G(\mathbf{s})
    \odot
    \mathbf{b}_G,
    \qquad
    \boldsymbol{\ell}_G(\mathbf{s})
    \odot
    \mathbf{b}_G
    \odot
    \boldsymbol{\tau}(\mathbf{s}).
\end{equation}
These terms allow the predicted pressure at a coordinate to depend jointly on local geometry, material state, and spatial position.
This is the mathematical mechanism by which GeONet represents pore-boundary effects, connectivity constraints, bottlenecks, and geometry-induced pressure gradients.
\end{remark}

\begin{remark}\textbf{Connection to the finite-volume pressure operator}:~In the finite-volume graph view, the solid--void mask determines the active pore-cell set $V_G$, the edge set $E_G$, and the transmissibility pattern of the discrete pressure matrix $A_G$.
The Branch Network can encode hydraulic properties that affect the magnitudes of transmissibilities, while the Leaf Network encodes the pore-space structure that determines which transmissibility connections exist.
The GeONet hypothesis class is therefore aligned with the discrete physics: it provides separate latent streams for material coefficients, coordinate evaluation, and geometry-induced connectivity, and then fuses those streams before predicting the pressure response.
\end{remark}

\subsection{Why explicit geometry is mathematically identifiable}
\label{app:identifiability}
Here, we formalize why the geometry mask is not an optional auxiliary input for the RCP pressure problem.
In a geometry-conditioned flow problem, the material fields $(\phi_G,k_G)$ and the query coordinate $\mathbf{s}$ do not, in general, uniquely determine the pressure response.
The binary solid--void mask $m_G$ defines the fluid domain, the admissible flow paths, and the local pore-space connectivity.
Therefore, unless the mask is already encoded redundantly in the Branch inputs, a model that ignores $m_G$ may be unable to distinguish two physically different pressure operators.
Let the target pressure operator be
\begin{equation}
    p_G(\mathbf{s})
    =
    \mathcal{G}^{\dagger}(m_G,\phi_G,k_G)(\mathbf{s}),
\end{equation}
A predictor that ignores geometry has the form
\begin{equation}
    f(\phi_G,k_G,\mathbf{s}),
\end{equation}
and must return the same value for any two geometries that share the same Branch inputs and query coordinate.

\begin{theorem}[Irreducible error without explicit geometry]
\label{thm:irreducible_geometry_error}
Suppose there exist two admissible RCP geometries $G_a$ and $G_b$ whose Branch-network inputs are identical:
\begin{subequations}
\begin{align}
    \mathcal{B}_{\rm in}(G_a)
    &=
    \mathcal{B}_{\rm in}(G_b),\\
    \mathcal{B}_{\rm in}(G)
    &:=
    (\phi_G,k_G).
\end{align}
\end{subequations}
Equivalently,
\begin{subequations}
\begin{align}
    \phi_{G_a}&=\phi_{G_b},\\
    k_{G_a}&=k_{G_b}.
\end{align}
\end{subequations}
However, assume their geometry masks differ, $m_{G_a}\neq m_{G_b}$, and that this difference produces distinct pressure fields on a set of positive measure.
That is, their pressure fields differ on a measurable set $S\subset\Omega$ with positive measure:
\begin{equation}
    |p_{G_a}(\mathbf{s})-p_{G_b}(\mathbf{s})|
    \geq
    \delta
    \qquad
    \text{for } \mathbf{s}\in S,
\end{equation}
for some $\delta>0$.
Assume that, conditional on the observed Branch and Trunk inputs, the two geometries occur with probabilities
\begin{subequations}
\begin{align}
    \mathbb{P}(G=G_a)&=\pi,\\
    \mathbb{P}(G=G_b)&=1-\pi,
\end{align}
\end{subequations}
where $0<\pi<1$.
Then every predictor $f(\phi,k,\mathbf{s})$ that does not observe $m_G$ has pointwise conditional mean-squared error at least
\begin{equation}
    \pi(1-\pi)\delta^2
    \qquad
    \text{for } \mathbf{s}\in S.
\end{equation}
In the equal-probability case $\pi=1/2$, the lower bound is
\begin{equation}
    \frac{\delta^2}{4}.
\end{equation}
If the risk is averaged over coordinates and $\nu(S)$ denotes the probability of sampling a coordinate in $S$, then the integrated conditional risk is at least
\begin{equation}
    \nu(S)\,\pi(1-\pi)\delta^2 .
\end{equation}
\end{theorem}

\begin{proof}
Fix a query coordinate $\mathbf{s}\in S$ and condition on the observed Branch--Trunk inputs $(\phi,k,\mathbf{s})$.
By assumption, the two geometries $G_a$ and $G_b$ have the same Branch inputs and the same query coordinate.
Therefore, a predictor that does not observe the geometry mask $m_G$ cannot distinguish between the two cases.
It must output a single value
\begin{equation}
    f(\phi,k,\mathbf{s})
\end{equation}
whether the hidden geometry is $G_a$ or $G_b$.
For notational simplicity, we define
\begin{subequations}
\begin{align}
    a &= p_{G_a}(\mathbf{s}),\\
    b &= p_{G_b}(\mathbf{s}).
\end{align}
\end{subequations}
Conditional on the observed inputs, the target pressure is a two-point random variable:
\begin{equation}
    p_G(\mathbf{s})
    =
    \begin{cases}
    a, & \text{with probability } \pi,\\
    b, & \text{with probability } 1-\pi.
    \end{cases}
\end{equation}
Thus, for any geometry-blind prediction $f$, the conditional squared-error risk is
\begin{equation}
    R(f)
    =
    \pi(a-f)^2
    +
    (1-\pi)(b-f)^2.
\end{equation}
Completing the square gives
\begin{equation}
    R(f)
    =
    \left[
    f-
    \left(
    \pi a+(1-\pi)b
    \right)
    \right]^2
    +
    \pi(1-\pi)(a-b)^2.
\end{equation}
Therefore, the risk is minimized when
\begin{equation}
    f^*(\phi,k,\mathbf{s})
    =
    \pi a+(1-\pi)b
    =
    \pi p_{G_a}(\mathbf{s})
    +
    (1-\pi)p_{G_b}(\mathbf{s}).
\end{equation}
This is the conditional mean of the unobserved pressure value.
At this optimum, the first term vanishes, so the minimum achievable conditional mean-squared error is
\begin{equation}
    R(f^*)
    =
    \pi(1-\pi)(a-b)^2.
\end{equation}
Returning to the pressure notation,
\begin{equation}
    R(f^*)
    =
    \pi(1-\pi)
    \left(
    p_{G_a}(\mathbf{s})-p_{G_b}(\mathbf{s})
    \right)^2.
\end{equation}
Because the two pressure fields differ by at least $\delta$ on $S$,
\begin{equation}
    \left|
    p_{G_a}(\mathbf{s})-p_{G_b}(\mathbf{s})
    \right|
    \geq
    \delta
    \qquad
    \text{for } \mathbf{s}\in S,
\end{equation}
we obtain
\begin{equation}
    R(f^*)
    \geq
    \pi(1-\pi)\delta^2.
\end{equation}
Since $f^*$ is the best possible geometry-blind predictor under squared-error loss, no predictor that ignores $m_G$ can achieve a smaller conditional mean-squared error.
Finally, averaging this pointwise lower bound over the coordinate distribution gives the integrated lower bound.
If $\nu(S)$ is the probability of sampling a coordinate in $S$, then
\begin{equation}
    \mathbb{E}_{\mathbf{s}}
    \left[
    R(f^*)
    \right]
    \geq
    \nu(S)\,\pi(1-\pi)\delta^2.
\end{equation}
\end{proof}

Theorem~\ref{thm:irreducible_geometry_error} is specific to the GeONet setting.
It does not simply state that a larger neural network is more expressive.
Rather, it states that when pore topology is not determined by material inputs, the geometry mask is statistically identifiable.
That is, without observing $m_G$, the pressure operator may be non-unique.
The Leaf Network removes this ambiguity by making the geometry an explicit input to the learned operator.
The key idea is the
\begin{equation}
R(f) = \left[
f-
\left(
\pi a+(1-\pi)b
\right)
\right]^2
+
\pi(1-\pi)(a-b)^2.
\end{equation}
It separates the error into two parts: avoidable error + unavoidable geometry ambiguity.
The avoidable error disappears when the model predicts the conditional mean.
The unavoidable term,
\begin{equation}
\pi(1-\pi)(a-b)^2,
\end{equation}
remains because the model does not know which geometry produced the pressure field.
This is why the theorem supports the Leaf Network: observing $(m_G)$ removes this ambiguity.

\subsection{Geometry stability under an RCP metric}
\label{app:geometry_stability}
Section~\ref{app:identifiability} showed that the geometry mask can be necessary for identifiability: two samples may have identical material-property inputs but different pressure fields because their solid--void masks define different pore-space connectivity.
We now consider the complementary question: once geometry is included as an input, when should two similar RCP geometries produce similar pressure responses?
This question is central to cross-domain generalization.
A GeONet model trained on an ensemble of RCP geometries can be expected to interpolate successfully to a new, unseen geometry when the target pressure operator varies continuously within the admissible geometry class.
Informally, nearby geometries should induce nearby pressure fields, provided that the perturbation does not create a major topological transition, such as opening or closing a critical inlet--outlet flow path (e.g., only possible when mineral dissolution or precipitation occurs).

A direct pointwise distance between binary masks is not well suited to this purpose.
Since
\begin{equation}
    m_G(\mathbf{x})\in\{0,1\},
\end{equation}
changing a single grid cell from solid to fluid produces
\begin{equation}
    \|m_G-m_{G'}\|_{L^\infty(\Omega)}=1,
\end{equation}
even if that isolated change has negligible influence on the pressure field.
Thus, the $L^\infty$ distance treats all cell flips as equally large, regardless of their hydraulic consequence.
This is too rigid for comparing RCP geometries.

Moreover, GeONet does not process the geometry mask as an abstract scalar descriptor.
The Leaf Network operates on the complete binary mask using a multiscale spatial encoding procedure.
At successive spatial resolutions, coarse representations of the mask and associated summary statistics are retained before the resulting multiscale features are mapped to the geometry latent representation.
Therefore, the geometry comparison should reflect spatial patterns over multiple length scales rather than isolated pixelwise differences.
This motivates a smoothed geometry-aware distance.
Therefore, the geometry comparison should reflect local spatial patterns rather than isolated pixelwise differences.
This motivates a smoothed geometry-aware distance.

Let $S_{\sigma}$ denote a local smoothing, or mollification, operator with length scale $\sigma$ comparable to a few grid cells.
For example, one may define
\begin{equation}
    S_{\sigma}m_G = \eta_{\sigma} * m_G,
\end{equation}
where $\eta_{\sigma}$ is a nonnegative smoothing kernel with characteristic width $\sigma$ and unit integral.
Equivalently,
\begin{equation}
    (S_{\sigma}m_G)(\mathbf{x})
    =
    \int_{\Omega}
    \eta_{\sigma}(\mathbf{x}-\mathbf{y})
    m_G(\mathbf{y})
    \,d\mathbf{y}.
\end{equation}
Thus, although $m_G$ is binary, $S_{\sigma}m_G$ is a locally averaged field.
It is close to one in regions surrounded by pore space, close to zero in regions surrounded by solid material, and intermediate near solid--void interfaces.

This smoothing is used only to define a comparison measure between geometries. It does not modify the physical geometry used in the finite-volume or PFLOTRAN simulation.
This smoothing is used only to define a comparison measure between geometries; it does not modify the physical geometry used in the finite-volume or PFLOTRAN simulation.
The goal is to compare masks over spatial scales relevant to the multiscale Leaf representation and to the local flow response, rather than at the scale of isolated single-pixel perturbations.

We define the RCP geometry-material distance
\begin{equation}
    d_{\rm Geo}(G,G')
    =
    \alpha_m
    \|S_{\sigma}(m_G-m_{G'})\|_{L^2(\Omega)}
    +
    \alpha_{\phi}
    \|\phi_G-\phi_{G'}\|_{L^2(\Omega)}
    +
    \alpha_k
    \|\log(k_G+k_\epsilon)-\log(k_{G'}+k_\epsilon)\|_{L^2(\Omega)}.
\end{equation}
Here $\alpha_m$, $\alpha_{\phi}$, and $\alpha_k$ are scaling weights that make the three terms comparable, and $k_\epsilon>0$ is a small regularization constant that avoids taking the logarithm of zero or near-zero permeability values.
The first term compares the solid--void masks after local smoothing and measures whether the two geometries exhibit similar pore-space patterns over the spatial scales represented by the multiscale Leaf encoder.
The second term compares porosity fields, and the third term compares permeability fields on a logarithmic scale.
The logarithmic permeability term is appropriate because permeability commonly spans several orders of magnitude and is normalized before training.

The quantity $d_{\rm Geo}$ should be interpreted as a local RCP similarity measure rather than as a universal topological metric.
It is designed to compare geometries within the same admissible RCP class, where inlet--outlet connectivity, boundary access, and coefficient bounds remain controlled.
A global Lipschitz theorem cannot hold for arbitrary binary masks because a small local change can sometimes cause a large global effect.
For example, opening or closing a single critical pore throat can create or destroy an inlet--outlet flow path.
Such a perturbation may be small in a smoothed image norm but large in hydraulic consequence.
Accordingly, the stability result below should be interpreted as local: it applies to perturbations that do not leave the admissible geometry class or trigger a connectivity transition.
Within this regime, the pressure operator varies continuously with respect to changes in geometry and material properties.
This provides a mathematical explanation for why cross-domain generalization is plausible when unseen RCP geometries remain close to the training distribution under $d_{\rm Geo}$.

\begin{theorem}[Local Lipschitz stability of the geometry-conditioned pressure operator]
\label{thm:lipschitz_operator}
Let $G$ and $G'$ be two admissible RCP geometries that can be compared on a common fluid domain $\Omega_c$ with the same pressure-prescribed boundary $\Gamma_D=\Gamma_{\mathrm{in}}\cup\Gamma_{\mathrm{out}}$.
Let
\begin{equation}
    V_0
    =
    \{v\in H^1(\Omega_c): v=0 \text{ on } \Gamma_D\}.
\end{equation}
Assume that the normalized pressures $p_G$ and $p_{G'}$ satisfy the weak Darcy problems
\begin{subequations}
\begin{align}
    \int_{\Omega_c}
    \kappa_G \nabla p_G\cdot\nabla v\,d\mathbf{x}
    &=
    0,
    \qquad
    \forall v\in V_0,
    \\
    \int_{\Omega_c}
    \kappa_{G'} \nabla p_{G'}\cdot\nabla v\,d\mathbf{x}
    &=
    0,
    \qquad
    \forall v\in V_0,
\end{align}
\end{subequations}
with the same normalized inlet--outlet pressure drop.
Here
\begin{subequations}
\begin{align}
    \kappa_G&=\frac{k_G}{\mu},
    \\
    \kappa_{G'}&=\frac{k_{G'}}{\mu}
\end{align}
\end{subequations}
are the effective hydraulic conductivities.
Suppose the conductivities are uniformly bounded:
\begin{equation}
    0<\kappa_{\min}
    \leq
    \kappa_G,\kappa_{G'}
    \leq
    \kappa_{\max}<\infty
    \qquad
    \text{a.e. in } \Omega_c.
\end{equation}
Assume further that the conductivity perturbation is controlled by the RCP geometry-material distance:
\begin{equation}
    \|\kappa_G-\kappa_{G'}\|_{L^\infty(\Omega_c)}
    \leq
    C_{\kappa}\,d_{\rm Geo}(G,G').
\end{equation}
Then there exists a constant $C>0$ such that
\begin{equation}
    \|p_G-p_{G'}\|_{H^1(\Omega_c)}
    \leq
    C\,d_{\rm Geo}(G,G').
\end{equation}
The constant $C$ depends on the ellipticity bounds, the Poincar\'e constant of $\Omega_c$, the boundary partition, and the imposed pressure-drop normalization.
\end{theorem}

\begin{proof}
Let
\begin{equation}
    e=p_G-p_{G'}.
\end{equation}
Because $p_G$ and $p_{G'}$ satisfy the same normalized Dirichlet pressure values on $\Gamma_D$, their difference satisfies
\begin{equation}
    e\in V_0.
\end{equation}
Subtracting the weak equation for $p_{G'}$ from the weak equation for $p_G$ gives
\begin{equation}
    \int_{\Omega_c}
    \kappa_G \nabla e\cdot \nabla v
    \,d\mathbf{x}
    =
    \int_{\Omega_c}
    (\kappa_{G'}-\kappa_G)
    \nabla p_{G'}\cdot \nabla v
    \,d\mathbf{x},
    \qquad
    \forall v\in V_0.
\end{equation}
Choosing $v=e$ gives
\begin{equation}
    \int_{\Omega_c}
    \kappa_G |\nabla e|^2
    \,d\mathbf{x}
    =
    \int_{\Omega_c}
    (\kappa_{G'}-\kappa_G)
    \nabla p_{G'}\cdot \nabla e
    \,d\mathbf{x}.
\end{equation}
Using the lower bound $\kappa_G\geq \kappa_{\min}$, we obtain
\begin{equation}
    \kappa_{\min}
    \|\nabla e\|_{L^2(\Omega_c)}^2
    \leq
    \left|
    \int_{\Omega_c}
    (\kappa_{G'}-\kappa_G)
    \nabla p_{G'}\cdot \nabla e
    \,d\mathbf{x}
    \right|.
\end{equation}
By H\"older's inequality,
\begin{equation}
    \kappa_{\min}
    \|\nabla e\|_{L^2(\Omega_c)}^2
    \leq
    \|\kappa_G-\kappa_{G'}\|_{L^\infty(\Omega_c)}
    \|\nabla p_{G'}\|_{L^2(\Omega_c)}
    \|\nabla e\|_{L^2(\Omega_c)}.
\end{equation}
If $e\neq 0$, division by $\|\nabla e\|_{L^2(\Omega_c)}$ yields
\begin{equation}
    \|\nabla e\|_{L^2(\Omega_c)}
    \leq
    \frac{1}{\kappa_{\min}}
    \|\kappa_G-\kappa_{G'}\|_{L^\infty(\Omega_c)}
    \|\nabla p_{G'}\|_{L^2(\Omega_c)}.
\end{equation}
The energy estimate from Theorem~\ref{thm:wellposedness} gives
\begin{equation}
    \|\nabla p_{G'}\|_{L^2(\Omega_c)}
    \leq
    C_E,
\end{equation}
where $C_E$ is proportional to the imposed pressure drop.
For normalized pressures, the pressure drop is fixed, so $C_E$ is a geometry-class constant.
Using the assumed relation
\begin{equation}
    \|\kappa_G-\kappa_{G'}\|_{L^\infty(\Omega_c)}
    \leq
    C_{\kappa}d_{\rm Geo}(G,G'),
\end{equation}
we obtain
\begin{equation}
    \|\nabla e\|_{L^2(\Omega_c)}
    \leq
    \frac{C_E C_{\kappa}}{\kappa_{\min}}
    d_{\rm Geo}(G,G').
\end{equation}
Finally, since $e\in V_0$, the Poincar\'e inequality gives
\begin{equation}
    \|e\|_{L^2(\Omega_c)}
    \leq
    C_P\|\nabla e\|_{L^2(\Omega_c)}.
\end{equation}
Therefore,
\begin{equation}
    \|e\|_{H^1(\Omega_c)}
    =
    \left(
    \|e\|_{L^2(\Omega_c)}^2
    +
    \|\nabla e\|_{L^2(\Omega_c)}^2
    \right)^{1/2}
    \leq
    \sqrt{1+C_P^2}\,
    \|\nabla e\|_{L^2(\Omega_c)}.
\end{equation}
Combining the previous estimates gives
\begin{equation}
    \|p_G-p_{G'}\|_{H^1(\Omega_c)}
    \leq
    C\,d_{\rm Geo}(G,G'),
\end{equation}
with
\begin{equation}
    C=
    \sqrt{1+C_P^2}
    \frac{C_E C_{\kappa}}{\kappa_{\min}}.
\end{equation}
This completes the proof.
\end{proof}

The lower-bound assumption on $\kappa_G$ and $\kappa_{G'}$ should be interpreted locally.
It excludes geometry perturbations that create or destroy an inlet--outlet connection or introduce truly disconnected zero-conductivity regions in the comparison domain.
Thus, the result is not a global stability theorem for arbitrary binary masks.
Rather, it states that within an admissible RCP class with controlled connectivity and uniformly bounded conductivities, small changes in the geometry-material metric lead to small changes in the pressure field.

\begin{proposition}[Discrete stability of the finite-volume pressure operator:]
\label{prop:discrete_stability}
Let $G$ and $G'$ be two admissible RCP geometries whose finite-volume pressure systems can be compared on a common set of free pressure degrees of freedom.
Equivalently, assume that the pressure vectors have been restricted to a common active pore-cell set or embedded into a fixed background grid with inactive cells masked out.
After applying Dirichlet pressure boundary conditions, let
\begin{subequations}
\begin{align}
    A_G\mathbf{p}_G&=\mathbf{b}_G,
    \\
    A_{G'}\mathbf{p}_{G'}&=\mathbf{b}_{G'}
\end{align}
\end{subequations}
denote the corresponding finite-volume systems.
Assume that the matrices are uniformly positive definite on the free pressure degrees of freedom:
\begin{subequations}
\begin{align}
    \lambda_{\min}(A_G)&\geq \alpha,
    \\
    \lambda_{\min}(A_{G'})&\geq \alpha,
    \\
    \alpha&>0.
\end{align}
\end{subequations}
Assume further that the geometry-material distance controls the perturbations in the pressure matrix and right-hand side:
\begin{equation}
    \|A_G-A_{G'}\|_2
    \leq
    L_A d_{\rm Geo}(G,G'),
\end{equation}
and
\begin{equation}
    \|\mathbf{b}_G-\mathbf{b}_{G'}\|_2
    \leq
    L_b d_{\rm Geo}(G,G').
\end{equation}
Then the discrete pressure solutions satisfy
\begin{equation}
    \|\mathbf{p}_G-\mathbf{p}_{G'}\|_2
    \leq
    \frac{1}{\alpha}
    \left(
    L_b
    +
    L_A\|\mathbf{p}_{G'}\|_2
    \right)
    d_{\rm Geo}(G,G').
\end{equation}
Thus, within an admissible geometry class with uniformly well-conditioned finite-volume operators, small changes in the RCP geometry-material metric lead to small changes in the discrete pressure vector.
\end{proposition}

\begin{proof}
Let
\begin{equation}
    \mathbf{e}
    =
    \mathbf{p}_G-\mathbf{p}_{G'}.
\end{equation}
Subtract the two finite-volume systems:
\begin{equation}
    A_G\mathbf{p}_G-A_{G'}\mathbf{p}_{G'}
    =
    \mathbf{b}_G-\mathbf{b}_{G'}.
\end{equation}
Add and subtract $A_G\mathbf{p}_{G'}$:
\begin{equation}
    A_G(\mathbf{p}_G-\mathbf{p}_{G'})
    +
    (A_G-A_{G'})\mathbf{p}_{G'}
    =
    \mathbf{b}_G-\mathbf{b}_{G'}.
\end{equation}
Therefore,
\begin{equation}
    A_G\mathbf{e}
    =
    (\mathbf{b}_G-\mathbf{b}_{G'})
    -
    (A_G-A_{G'})\mathbf{p}_{G'}.
\end{equation}
Since $A_G$ is positive definite and $\lambda_{\min}(A_G)\geq \alpha$, we have
\begin{equation}
    \|A_G^{-1}\|_2
    \leq
    \frac{1}{\alpha}.
\end{equation}
Hence,
\begin{equation}
    \|\mathbf{e}\|_2
    \leq
    \|A_G^{-1}\|_2
    \left(
    \|\mathbf{b}_G-\mathbf{b}_{G'}\|_2
    +
    \|A_G-A_{G'}\|_2
    \|\mathbf{p}_{G'}\|_2
    \right).
\end{equation}
Using the assumed perturbation bounds,
\begin{equation}
    \|\mathbf{b}_G-\mathbf{b}_{G'}\|_2
    \leq
    L_b d_{\rm Geo}(G,G'),
\end{equation}
and
\begin{equation}
    \|A_G-A_{G'}\|_2
    \leq
    L_A d_{\rm Geo}(G,G'),
\end{equation}
we obtain
\begin{equation}
    \|\mathbf{p}_G-\mathbf{p}_{G'}\|_2
    \leq
    \frac{1}{\alpha}
    \left(
    L_b
    +
    L_A\|\mathbf{p}_{G'}\|_2
    \right)
    d_{\rm Geo}(G,G').
\end{equation}
This proves the stated bound.
\end{proof}

In this discrete statement, the constant $\alpha$ measures the conditioning of the finite-volume pressure operator.
If $\alpha$ is bounded away from zero, the pressure system is well anchored by the Dirichlet boundaries and does not contain nearly disconnected floating components.
If $\alpha$ becomes small, small changes in geometry or boundary forcing can produce large changes in pressure.
This is the discrete analog of losing coercivity in the weak Darcy formulation.
The perturbation term $\|A_G-A_{G'}\|_2$ captures changes in graph connectivity and edge transmissibilities.
In a finite-volume graph view, one may write schematically
\begin{equation}
    A_G \approx B_G^{\top} W_G B_G + D_G,
\end{equation}
where $B_G$ is an incidence matrix encoding pore-cell connectivity, $W_G$ contains edge transmissibilities, and $D_G$ contains Dirichlet boundary contributions.
The solid--void mask affects both $B_G$ and $W_G$: it determines which neighboring cells are connected and what conductance is assigned to each open face.
The vector $\mathbf{b}_G$ also depends on geometry through the boundary faces connected to prescribed inlet and outlet pressures.

This discrete perturbation estimate is directly aligned with the GeONet training setting.
GeONet observes pressure values on a finite-volume grid, while the Leaf Network encodes the binary mask that determines the graph structure of $A_G$.
Therefore, if an unseen RCP geometry remains close to the training distribution under $d_{\rm Geo}$ and does not undergo a connectivity-changing transition, its matrix $A_G$, right-hand side $\mathbf{b}_G$, and pressure solution $\mathbf{p}_G$ are expected to vary continuously.
This provides a discrete mathematical rationale for cross-domain generalization.

%%%%%% %%%%%% %%%%%% %%%%%% %%%%%% %%%%%% %%%%%% %%%%%% %%%%%% %%%%%% 
%%%%%% Need to update this accordingly.
\subsection{Low-rank smoothing in coordinate-separable operator models}
\label{app:low_rank}
A standard Branch--Trunk DeepONet represents the pressure field as
\begin{equation}
    \widehat{p}(\mathbf{s})
    =
    \sum_{i=1}^{r}
    b_i(\phi,k)\,t_i(\mathbf{s}),
\end{equation}
where $b_i(\phi,k)$ are Branch-network coefficients and $t_i(\mathbf{s})$ are Trunk-network coordinate functions.
After training, the functions $t_1,\ldots,t_r$ define a fixed coordinate basis that is shared across all samples.
The Branch Network changes only the coefficients multiplying these basis functions.

Therefore, for a fixed trained model, every predicted pressure field lies in the same finite-dimensional coordinate space,
\begin{equation}
    \widehat{p}
    \in
    \mathrm{span}\{t_1,\ldots,t_r\}.
\end{equation}
This is a coordinate-separable representation: material-property information enters through the coefficients, while spatial variation is represented through a fixed set of coordinate modes.
This observation does not imply that DeepONet is incapable of approximating complex pressure fields.
With sufficiently large rank, additional inputs, or architecture modifications, a DeepONet-type model can represent richer solution families.
Rather, the point is that a finite-rank coordinate-separable model can smooth localized geometry-induced pressure features when those features occur at geometry-dependent locations.
In that case, a single fixed coordinate basis must represent many possible locations of sharp pressure drops, bottlenecks, and pore-boundary effects.

\begin{proposition}[Projection error for geometry-induced pressure modes]
\label{prop:projection_error}
Let $\{p_{G_j}\}_{j=1}^{J}$ be pressure fields associated with $J$ admissible RCP geometries, evaluated on a common grid or common comparison domain.
Let
\begin{equation}
    \bar{p}
    =
    \frac{1}{J}
    \sum_{m=1}^{J}
    p_{G_m}
\end{equation}
be the empirical mean pressure field, and define the centered fields
\begin{equation}
    \widetilde{p}_{G_j}
    =
    p_{G_j}-\bar{p}.
\end{equation}
Let
\begin{equation}
    \lambda_1\geq \lambda_2\geq\cdots
\end{equation}
be the eigenvalues of the empirical covariance operator associated with the centered pressure fields.
Then any representation whose coordinate dependence lies in a fixed $r$-dimensional subspace has average squared reconstruction error at least
\begin{equation}
    \sum_{i>r}\lambda_i .
\end{equation}
\end{proposition}

\begin{proof}
Let $\mathcal{U}_r$ be any $r$-dimensional coordinate subspace, and let $\Pi_{\mathcal{U}_r}$ denote the orthogonal projection onto that subspace.
The average squared reconstruction error of projecting the centered pressure fields onto $\mathcal{U}_r$ is
\begin{equation}
    \frac{1}{J}
    \sum_{j=1}^{J}
    \left\|
    \widetilde{p}_{G_j}
    -
    \Pi_{\mathcal{U}_r}\widetilde{p}_{G_j}
    \right\|^2 .
\end{equation}
Among all $r$-dimensional subspaces, this quantity is minimized by the span of the first $r$ empirical principal modes of the pressure ensemble.
The minimum achievable error is the sum of the neglected covariance eigenvalues,
\begin{equation}
    \sum_{i>r}\lambda_i .
\end{equation}
A coordinate-separable model with $r$ fixed trunk functions produces pressure fields in one such $r$-dimensional coordinate subspace.
Therefore, its average squared reconstruction error cannot be smaller than the optimal rank-$r$ projection error.
\end{proof}

For RCP geometries, solid barriers, pore constrictions, disconnected regions, and preferential channels generate pressure features whose locations depend on the geometry mask.
For example, a sharp pressure drop may occur at one location in $G_a$ but at a different location in $G_b$ because the constriction or high-resistance pathway has moved.
Such moving localized features typically require many global coordinate modes to represent accurately.
If the effective rank $r$ is too small, the best fixed-basis approximation tends to represent these features by their average effect, producing a smoother pressure field.

GeONet addresses this limitation by allowing geometry to modify the coordinate response before decoding the pressure.
Through Leaf--Trunk and Leaf--Branch--Trunk interactions, GeONet can form geometry-dependent coordinate features of the form
\begin{equation}
    \Psi_i\left(
    \boldsymbol{\ell}_G(\mathbf{s}),
    \mathbf{b}_G,
    \boldsymbol{\tau}(\mathbf{s})
    \right),
\end{equation}
where $\boldsymbol{\ell}_G(\mathbf{s})$ is the geometry feature, $\mathbf{b}_G$ is the material-property feature, and $\boldsymbol{\tau}(\mathbf{s})$ is the coordinate feature.
Thus, the effective coordinate basis is no longer fixed across all geometries.
It can shift, deform, or activate differently depending on the solid--void mask.
In this sense, GeONet replaces a single global coordinate subspace with a geometry-conditioned family of coordinate responses.
This provides a mathematical explanation for why GeONet can preserve sharp pressure transitions and localized pore-scale features that are finite-rank coordinate-separable, which a baseline tends to smooth.

\begin{proposition}[Geometry-gated basis enrichment]
\label{prop:basis_enrichment}
Suppose the Leaf Network and gating module partition the geometry latent space into $Q$ distinguishable regimes.
Suppose further that, within each regime, the decoder can represent a rank-$r$ Branch--Trunk coordinate model.
Then the resulting GeONet class can represent a union of $Q$ rank-$r$ coordinate subspaces.
The dimension of the span of this union is at most $Qr$, and it can reach $Qr$ when the regime-specific coordinate subspaces are linearly independent.
\end{proposition}

\begin{proof}
Let $\mathcal{U}_q$ denote the rank-$r$ coordinate subspace used in regime $q$, where $q=1,\ldots,Q$.
Within a fixed regime, the model represents pressure fields in $\mathcal{U}_q$.
Across all regimes, the possible coordinate responses lie in the union
\begin{equation}
    \bigcup_{q=1}^{Q}\mathcal{U}_q .
\end{equation}
The linear span of this union is
\begin{equation}
    \mathcal{U}_1+\mathcal{U}_2+\cdots+\mathcal{U}_Q,
\end{equation}
whose dimension is at most $Qr$.
If the subspaces are linearly independent, the dimension is exactly $Qr$.
\end{proof}

The model does not need to explicitly label these regimes.
They are learned through the Leaf features and gates.
The practical effect is that GeONet can place pressure transitions, bottleneck-induced gradients, and local extrema at geometry-dependent locations instead of forcing all geometries through a single global coordinate basis.

\subsection{Residual-to-error estimate for GeONet pressure predictions}
\label{app:residual_error}

Although the GeONet models reported in this manuscript are trained using supervised pressure losses, their predictions can also be interpreted through the residual of the Darcy pressure equation.
This provides a bridge between data-driven training and physics-informed error analysis.

Given a predicted pressure field $\widehat{p}_G$, define the interior Darcy residual
\begin{equation}
    \mathcal{R}_G(\widehat{p}_G)
    =
    -\nabla\cdot
    \left(
    \kappa_G\nabla \widehat{p}_G
    \right)
    \qquad
    \text{in } \Omega_f(G),
\end{equation}
where
\begin{equation}
    \kappa_G=\frac{k_G}{\mu}.
\end{equation}
If $\widehat{p}_G$ satisfies the same Dirichlet pressure values as the true solution $p_G$, then the error
\begin{equation}
    e=p_G-\widehat{p}_G
\end{equation}
belongs to $V_{0,G}$ and satisfies
\begin{equation}
    \int_{\Omega_f(G)}
    \kappa_G\nabla e\cdot\nabla v\,d\mathbf{x}
    =
    \langle
    \mathcal{R}_G(\widehat{p}_G),v
    \rangle
    \qquad
    \forall v\in V_{0,G}.
\end{equation}

\begin{theorem}[Energy-norm error controlled by Darcy residual]
\label{thm:residual_bound}
Assume $\widehat{p}_G$ satisfies the same Dirichlet boundary values as $p_G$.
Then
\begin{equation}
    \|\nabla(p_G-\widehat{p}_G)\|_{L^2(\Omega_f(G))}
    \leq
    \frac{1}{\kappa_{\min}}
    \|\mathcal{R}_G(\widehat{p}_G)\|_{V_{0,G}^{*}},
\end{equation}
where $V_{0,G}^{*}$ is the dual space of $V_{0,G}$ with norm
\begin{equation}
    \|\mathcal{R}\|_{V_{0,G}^{*}}
    =
    \sup_{v\in V_{0,G},\,v\neq 0}
    \frac{\langle \mathcal{R},v\rangle}
    {\|\nabla v\|_{L^2(\Omega_f(G))}} .
\end{equation}
\end{theorem}

\begin{proof}
Set $v=e$ in the residual equation:
\begin{equation}
    \int_{\Omega_f(G)}
    \kappa_G|\nabla e|^2\,d\mathbf{x}
    =
    \langle
    \mathcal{R}_G(\widehat{p}_G),e
    \rangle .
\end{equation}
Using $\kappa_G\geq \kappa_{\min}$ gives
\begin{equation}
    \kappa_{\min}
    \|\nabla e\|_{L^2(\Omega_f(G))}^2
    \leq
    \langle
    \mathcal{R}_G(\widehat{p}_G),e
    \rangle .
\end{equation}
By the definition of the dual norm,
\begin{equation}
    \langle
    \mathcal{R}_G(\widehat{p}_G),e
    \rangle
    \leq
    \|\mathcal{R}_G(\widehat{p}_G)\|_{V_{0,G}^{*}}
    \|\nabla e\|_{L^2(\Omega_f(G))}.
\end{equation}
Dividing by $\kappa_{\min}\|\nabla e\|_{L^2(\Omega_f(G))}$ yields the desired bound.
\end{proof}

If the predicted pressure does not exactly satisfy the Dirichlet boundary conditions, the error estimate can be extended by adding a boundary mismatch term.
In that case, one first lifts the boundary discrepancy into the domain and applies the same residual argument to the remaining homogeneous error.
This residual bound is useful for future physics-informed extensions of GeONet.
It also explains why prediction errors often concentrate near sharp pressure gradients, pore constrictions, and solid barriers: these are the regions where local residuals are most sensitive to geometry, edge transmissibility, and finite-volume discretization.

\subsection{Coordinate-sampling loss and unbiased grid-risk estimation}
\label{app:coordinate_sampling}

GeONet is trained by sampling coordinates from high-resolution pressure fields rather than by processing every grid point at every optimization step.
This coordinate batching is memory efficient and is consistent with the operator learning objective.

Let $F_G$ denote the set of active fluid cells for geometry $G$.
The full-grid pressure risk for a model $\widehat{p}_{\theta}$ is
\begin{equation}
    \mathcal{E}_h(\theta;G)
    =
    \frac{1}{|F_G|}
    \sum_{i\in F_G}
    \left|
    \widehat{p}_{\theta}
    (\mathbf{x}_i;m_G,\phi_G,k_G)
    -
    p_G(\mathbf{x}_i)
    \right|^2 .
\end{equation}
During training, a subset of coordinates
\begin{equation}
    \mathcal{S}_G\subset F_G
\end{equation}
is sampled.
The mini-batch estimator is
\begin{equation}
    \widehat{\mathcal{E}}_h(\theta;G)
    =
    \frac{1}{|\mathcal{S}_G|}
    \sum_{i\in \mathcal{S}_G}
    \left|
    \widehat{p}_{\theta}
    (\mathbf{x}_i;m_G,\phi_G,k_G)
    -
    p_G(\mathbf{x}_i)
    \right|^2 .
\end{equation}

\begin{theorem}[Unbiasedness of coordinate-sampled training]
\label{thm:unbiased_sampling}
If $\mathcal{S}_G$ is sampled uniformly from $F_G$, then
\begin{subequations}
\begin{align}
    \mathbb{E}_{\mathcal{S}_G}
    \left[
    \widehat{\mathcal{E}}_h(\theta;G)
    \right]
    &=
    \mathcal{E}_h(\theta;G).
\end{align}
If the sampled coordinates are independent and the pointwise squared error has finite variance, then
\begin{align}
    \mathrm{Var}
    \left[
    \widehat{\mathcal{E}}_h(\theta;G)
    \right]
    &=
    O(|\mathcal{S}_G|^{-1}).
\end{align}
\end{subequations}
\end{theorem}

\begin{proof}
For a uniformly sampled coordinate $i\in F_G$,
\begin{equation}
    \mathbb{E}_{i}
    \left[
    \left|
    \widehat{p}_{\theta}
    (\mathbf{x}_i;m_G,\phi_G,k_G)
    -
    p_G(\mathbf{x}_i)
    \right|^2
    \right]
    =
    \mathcal{E}_h(\theta;G).
\end{equation}
Linearity of expectation gives the unbiasedness result for the sample average.
If coordinates are sampled independently, the variance of the average decreases inversely with the number of samples, giving the stated $O(|\mathcal{S}_G|^{-1})$ scaling.
\end{proof}

If coordinates are sampled nonuniformly, for example to oversample high-gradient regions or pore boundaries, the same estimator can be made unbiased by applying standard importance weights.
Uniform sampling, however, directly estimates the full-grid mean-squared error over the active pore cells.
This result clarifies why coordinate batching is compatible with operator learning.
GeONet minimizes the expected prediction error over geometries and coordinates while avoiding the memory cost of full-field optimization at every training step.

\subsection{Relationship between training loss and reported metrics}
\label{app:metrics_analysis}

Let
\begin{equation}
    \mathbf{e}
    =
    \widehat{\mathbf{p}}-\mathbf{p}
\end{equation}
denote the vector of prediction errors on the evaluated fluid cells, and let $N$ be the number of evaluated cells.
The root mean squared error is
\begin{equation}
    \mathrm{RMSE}
    =
    \frac{\|\mathbf{e}\|_2}{\sqrt{N}}.
\end{equation}
The normalized RMSE reported in the one-to-one pressure comparisons is
\begin{equation}
    \mathrm{NRMSE}
    =
    \frac{\|\mathbf{e}\|_2}
    {\sqrt{N}\left(P_{\max}-P_{\min}\right)}.
\end{equation}
Thus, minimizing the mean-squared training loss directly reduces the numerator of RMSE and NRMSE.

The coefficient of determination can be written as
\begin{equation}
    R^2
    =
    1
    -
    \frac{\|\mathbf{e}\|_2^2}
    {\|\mathbf{p}-\bar{P}\mathbf{1}\|_2^2},
\end{equation}
where $\bar{P}$ is the spatial mean of the reference pressure field.
Therefore, a high $R^2$ indicates that the squared error is small relative to the intrinsic spatial variance of the pressure field.

The mean absolute percentage error is
\begin{equation}
    \mathrm{MAPE}
    =
    \frac{100}{N}
    \sum_{i=1}^{N}
    \left|
    \frac{e_i}{p_i+\epsilon}
    \right|,
\end{equation}
where $\epsilon>0$ prevents division by values close to zero.
This is important when pressure is shifted or normalized so that outlet values may approach zero.
MAPE provides an interpretable percent-scale error, while NRMSE and $R^2$ measure error relative to the pressure range and spatial variance.

The pressure histograms and profile plots complement these scalar metrics.
NRMSE, MAPE, and $R^2$ quantify global agreement, whereas histograms and one-dimensional pressure slices assess whether the model preserves distributional tails, localized pressure drops, and geometry-induced gradients.

\subsection{DeepONet-BC Hyperparameter Optimization Diagnostics}
\label{app:deeponet_hpo}

This subsection reports the complete DeepONet-BC HPO diagnostics corresponding to the GeONet analyses in the main text.
DeepONet-BC was evaluated using the same data partitions, trial budgets, optimization objective, and applicable search ranges.
Its trunk formulation emphasizes the prescribed boundary conditions, and its leading configurations use zero or mild random Fourier features.
These figures document how the BC-focused branch--trunk architecture responds to the reduced- and full-data optimization stages.

During Stage~I, the parameter-wise distributions occupy a comparatively narrow loss range (Fig.~\ref{fig:deeponet_param_500}).
Learning rate provides the clearest separation among trials, while changes in weight decay, gradient-loss weight, latent dimensions, Fourier features, and trunk depth produce extensive overlap.
The global analysis reaches the same conclusion (Fig.~\ref{fig:deeponet_sensitivity_500}a): learning rate has the largest rank correlation with validation loss ($\rho=+0.60$), followed by much weaker associations for trunk depth ($\rho=-0.24$) and trunk width ($\rho=-0.20$).
Architecture-level distributions overlap strongly (Fig.~\ref{fig:deeponet_sensitivity_500}b), and the isolated lowest-loss trial visible in the search trajectory (Fig.~\ref{fig:deeponet_sensitivity_500}c) does not alter the broader pattern of limited separation among configurations.

\begin{figure}
    \centering
    \includegraphics[width=0.92\textwidth]{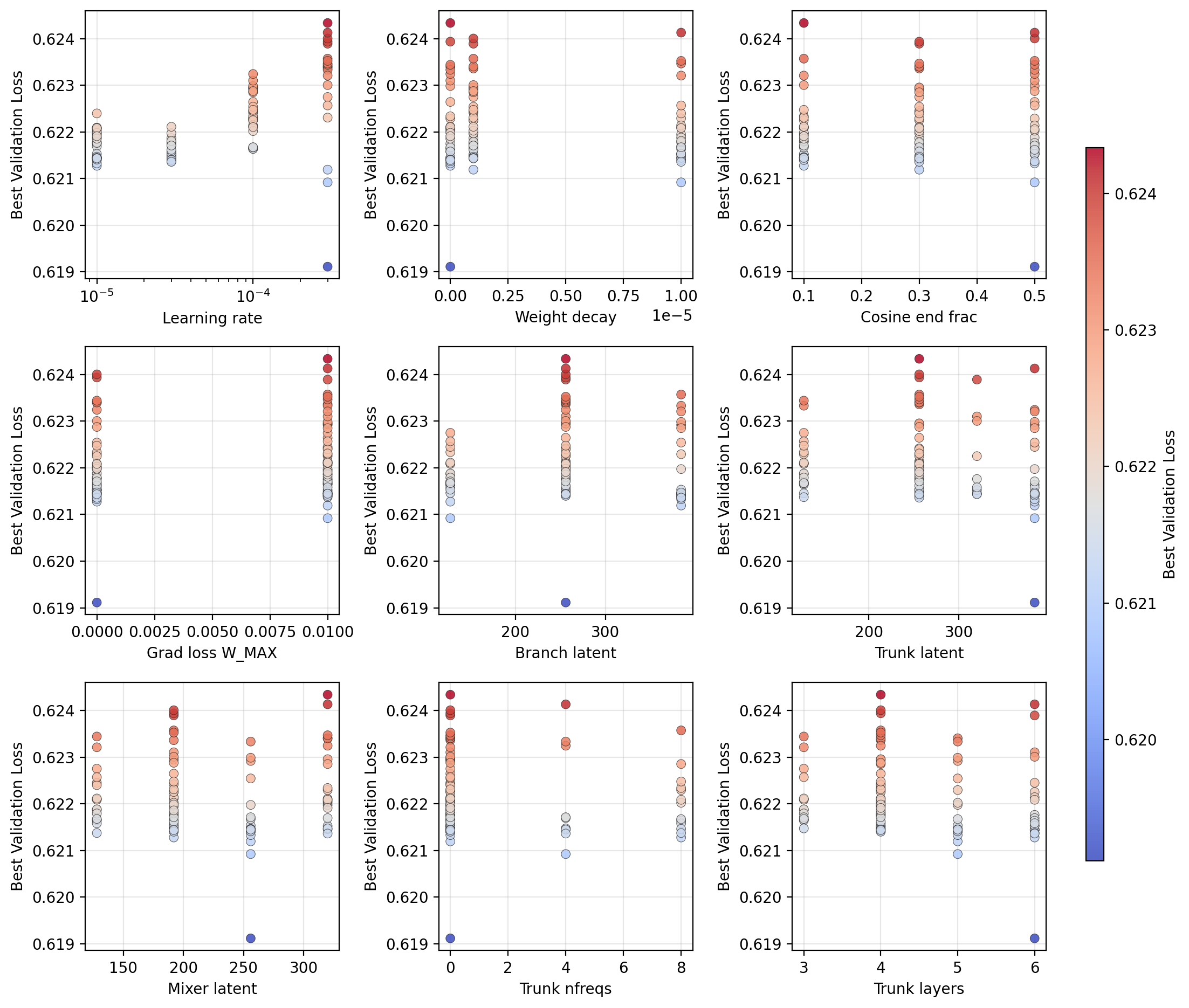}
    \caption{\textbf{Parameter-wise variation in Stage~I DeepONet-BC validation performance ($N=500$).}
    Best validation loss is shown against each sampled optimization and architecture hyperparameter; marker color denotes validation loss.
    Learning rate produces the most visible separation, whereas the remaining parameters exhibit strongly overlapping loss distributions.
    The broad overlap indicates limited sensitivity to isolated changes in model width, Fourier features, or trunk depth within the evaluated Stage~I space.
    The minimum loss of $0.61911$ is attained by free-form trial~040, which uses a wide, six-layer trunk and $n_f=0$.}
    \label{fig:deeponet_param_500}
\end{figure}

\begin{figure}
    \centering
    \includegraphics[width=0.99\textwidth]{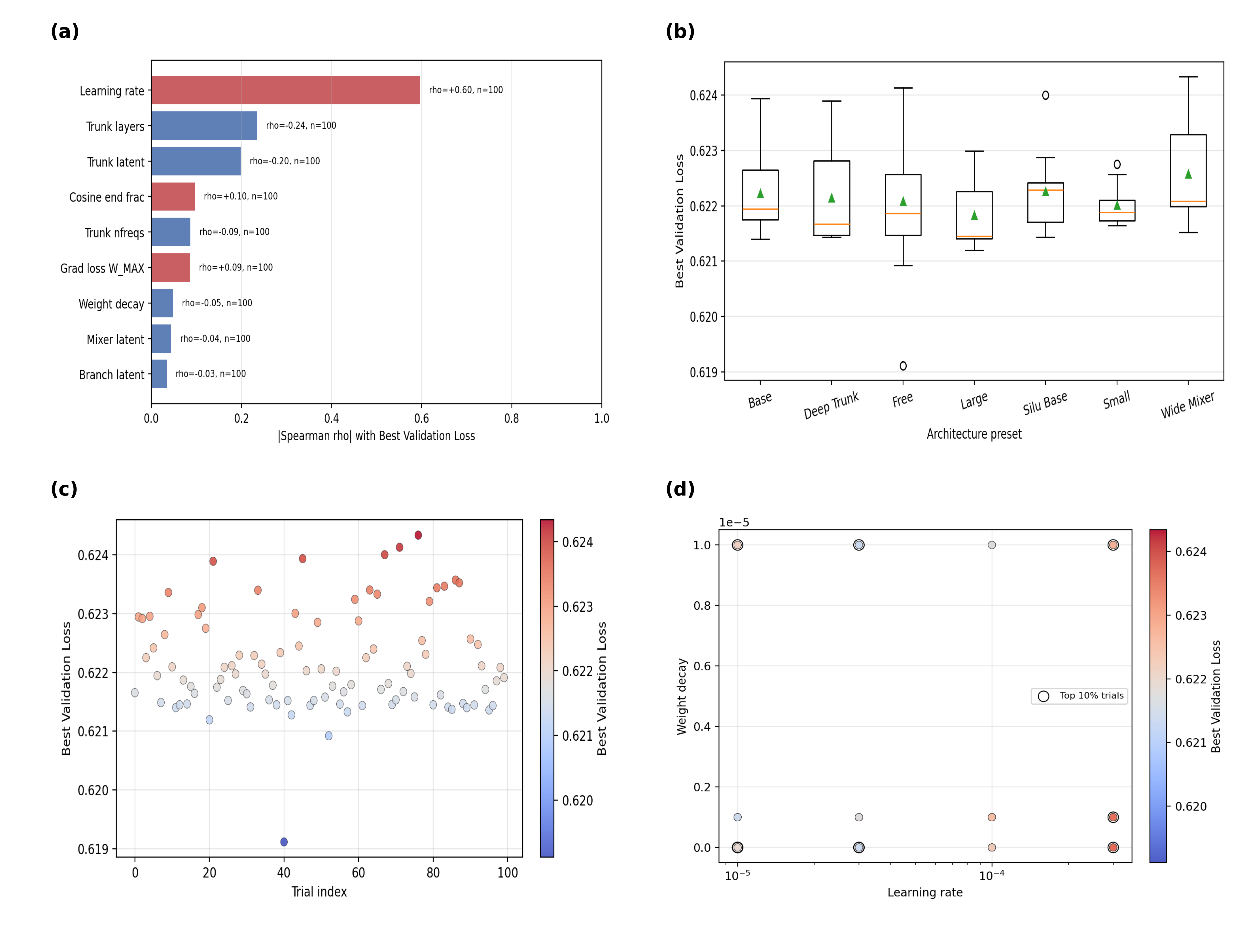}
    \caption{\textbf{Stage~I sensitivity and architecture analysis for DeepONet-BC ($N=500$).}
    (A) Absolute Spearman rank correlations between sampled hyperparameters and best validation loss; bar color indicates correlation sign.
    (B) Validation-loss distributions by architecture preset, with orange median lines and green mean markers.
    (C) Best validation loss across the 100-trial search.
    (D) Joint learning-rate--weight-decay projection, with the top 10\% of trials outlined.
    Learning rate is the dominant sampled control, while the extensive overlap among architecture families indicates weak architecture-level separation.}
    \label{fig:deeponet_sensitivity_500}
\end{figure}

Following transfer to $N=5000$, the DeepONet-BC trials separate into a lower-loss band near $0.567$--$0.579$ and a distinct high-loss band near $0.623$ (Fig.~\ref{fig:deeponet_param_5000}).
The poorest trials occur at the lowest learning rate, while several architecture and optimizer settings map to nearly the same lower band.
Accordingly, the Stage~II rank analysis identifies learning rate as the dominant variable ($\rho=-0.74$), followed by weight decay ($\rho=-0.45$); the remaining correlations are modest (Fig.~\ref{fig:deeponet_sensitivity_5000}a).
The architecture distributions and trial sequence separate configurations into two performance regimes (Fig.~\ref{fig:deeponet_sensitivity_5000}b,c).
The best loss of $0.56718$ represents an $8.4\%$ reduction from the Stage~I optimum, demonstrating that the BC-focused formulation benefits from the larger training cohort.
The separation also shows that this improvement depends on selecting an effective transferred seed and optimization schedule.

\begin{figure}
    \centering
    \includegraphics[width=0.92\textwidth]{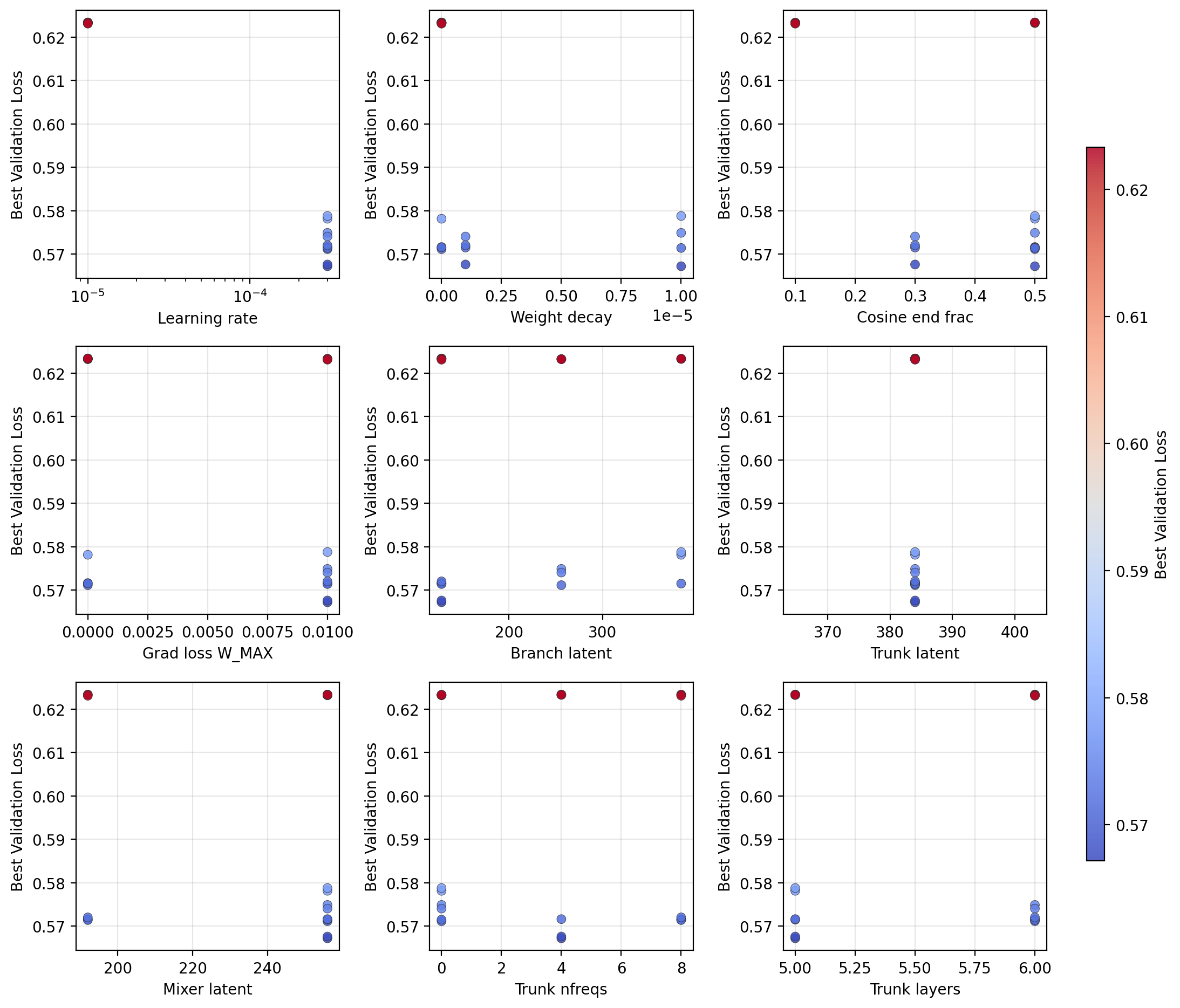}
    \caption{\textbf{Parameter-wise variation in Stage~II DeepONet-BC validation performance ($N=5000$).}
    Best validation loss is plotted against the optimization and architecture settings evaluated during transferred refinement; marker color denotes validation loss.
    The trials separate into a narrow lower-loss band and a high-loss regime associated most clearly with the lowest learning rate.
    The Stage~II minimum of $0.56718$ is achieved by trial~001 using a trunk width of 384, $n_f=4$, $\sigma_f=1$, and $\eta=3\times10^{-4}$.}
    \label{fig:deeponet_param_5000}
\end{figure}

\begin{figure}
    \centering
    \includegraphics[width=0.99\textwidth]{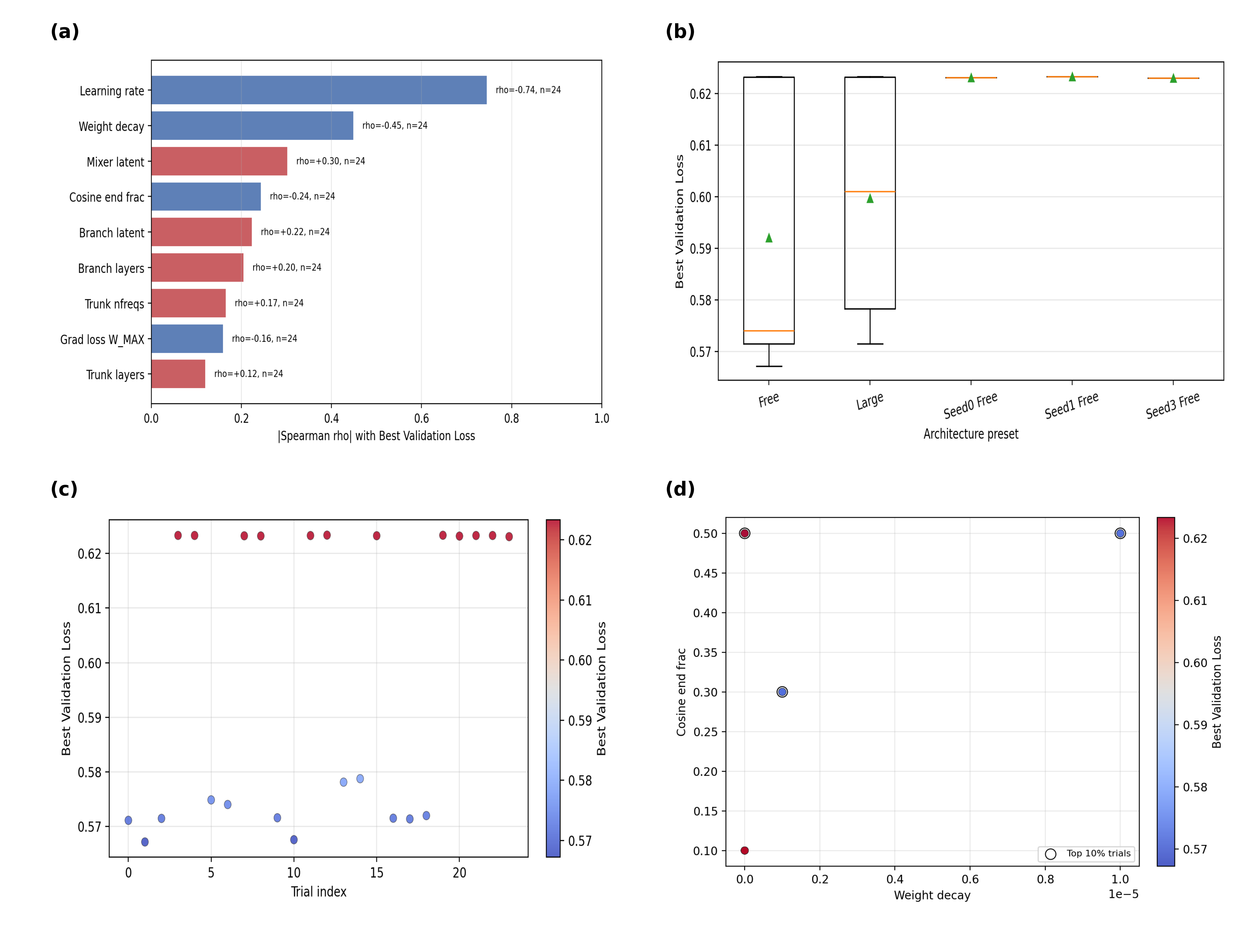}
    \caption{\textbf{Stage~II sensitivity and architecture analysis for DeepONet-BC ($N=5000$).}
    (A) Absolute Spearman rank correlations between locally varied hyperparameters and best validation loss; bar color indicates correlation sign.
    (B) Validation-loss distributions for transferred architecture families, with medians and means shown by orange lines and green triangles.
    (C) Best validation loss across the 24 refinement trials.
    (D) Weight-decay--cosine-end-fraction projection, with the top 10\% of trials outlined.
    Learning rate and weight decay provide the strongest rank separation between the two performance regimes.
    As in the GeONet analysis, these correlations are descriptive of a small, restricted transfer pool and should not be interpreted as independent causal effects.}
    \label{fig:deeponet_sensitivity_5000}
\end{figure}

\subsection{Implications for the reported GeONet results}
\label{app:implications}
The mathematical analysis above supports four claims made by the empirical results.
First, the pressure operator is geometry-conditioned.
The weak formulation and finite-volume graph interpretation show that the solid--void mask changes the fluid domain, the admissible boundary partition, and the connectivity of the discrete pressure system.
Therefore, the geometry input is not cosmetic; it changes the physical operator being learned.
Second, omitting geometry can create an irreducible identifiability error.
Theorem~\ref{thm:irreducible_geometry_error} shows that if two RCP geometries share the same material-property inputs but have different pressure fields, no Branch--Trunk model that ignores the mask can recover both solutions exactly.
Third, the gated Leaf--Branch--Trunk fusion gives GeONet a mechanism for local geometry-coordinate interactions.
These interactions allow pressure features to move with pore boundaries, constrictions, and connectivity changes, rather than being represented only by a fixed coordinate basis.
Fourth, the residual and stability estimates explain why cross-domain generalization is plausible within the admissible RCP distribution.
If unseen geometries remain close to the training support under the geometry-material metric and do not violate connectivity or ellipticity assumptions, then the underlying pressure solution varies continuously.
A trained GeONet can therefore interpolate this geometry-conditioned operator in a physically meaningful way.
Together, these points provide a mathematical rationale for the observed behavior of GeONet: the model preserves sharp pressure transitions, distributional tails, and geometry-controlled pressure profiles more faithfully than a coordinate-centric baseline, while maintaining strong agreement on unseen RCP geometries.